%% file: main.tex
\def\llncs{0}
\def\comments{0}

\ifnum\llncs=0
\documentclass[11pt]{article}
\usepackage{fullpage}
\else
\documentclass[11pt]{llncs}
\fi
\input{macros}
\usepackage{graphicx} % Required for inserting images
\usepackage{authblk}
\definecolor{corlinks}{RGB}{200,0,0}
\definecolor{cormenu}{RGB}{200,0,0}
\definecolor{corurl}{RGB}{200,0,0}

\hypersetup{
 colorlinks=true,
 urlcolor=corlinks,
 linkcolor=corlinks,
 menucolor=cormenu,
 citecolor=corlinks,
 pdfborder= 0 0 0
}

\ifnum\llncs=0
\fi

\title{CountCrypt: Quantum Cryptography between QCMA and PP}
\ifnum\llncs=1
\author{}
\institute{}
\else
\author[1]{Eli Goldin}
\author[2]{Tomoyuki Morimae}
\author[3]{Saachi Mutreja}
\author[4,2]{Takashi Yamakawa}
\affil[1]{\small New York University (\texttt{eli.goldin@nyu.edu})}
\affil[2]{\small Yukawa Institute for Theoretical Physics, Kyoto University (\texttt{tomoyuki.morimae@yukawa.kyoto-u.ac.jp})}
\affil[3]{\small Columbia University (\texttt{saachi@berkeley.edu})}
\affil[4]{\small NTT Social Informatics Laboratories (\texttt{takashi.yamakawa@ntt.com})}
\fi

\begin{document}

\maketitle
\input{abstract}

\newpage
\setcounter{tocdepth}{2}
\tableofcontents
 \newpage

\input{intro}
\input{techoverview}

\input{prelims/prelims}
\input{prelims/tricks}

\input{presampling}

\input{Unitaryseparation}

\input{money/unitarylightning}
\input{unitaryQCCCKE}
\input{unitaryQCCCcommit}
\input{unitaryBQPQCMA}
\input{sgmix/ketoowpuzz}
\input{sgmix/commitments}
\input{2qkdtostatepuzzle}

\ifnum\llncs=0
\section*{Acknowledgement}
Eli Goldin is supported by a NSF Graduate Student Research Fellowship.
Tomoyuki Morimae is supported by JST CREST JPMJCR23I3, JST Moonshot JPMJMS20615-1-1, JST FOREST, MEXT QLEAP, the Grant-in Aid for Transformative Research Areas (A) 21H05183,
and the Grant-in-Aid for Scientific Research (A) No.22H00522.
\fi

% \ifnum\llncs=1
\input{main.bbl}

% \bibliographystyle{splncs04}
% \else
% \bibliographystyle{references/halpha}
% \fi
% \bibliography{references/refs,references/abbrev0,references/crypto}

\appendix
\input{sgsimlemma}
\input{lemmas/ortholem}

\end{document}

%% file: macros.tex
\usepackage{amsmath,amssymb}
\ifnum\llncs=0
\usepackage{amsthm}
\usepackage{geometry}
\fi
\usepackage{mathtools}
\usepackage{algorithm}
\usepackage{algorithmic}
\usepackage[table]{xcolor}
\usepackage{xcolor}
\usepackage{tcolorbox}
\usepackage{tikz}
\usepackage{graphicx}
\usepackage{tikz}
\usepackage{hyperref}
\usepackage{cleveref}
\usepackage{tikz-cd}
\usepackage{minibox}
\usepackage{braket}

\makeatletter
\newcommand\ketbra[2][]{%
  \def\ketbra@firstarg{#1}%
  \def\ketbra@secondarg{#2}%
  \ifx\ketbra@firstarg\empty%
    \left\lvert\ketbra@secondarg\middle\rangle \! \middle\langle \ketbra@secondarg\right\rvert%
  \else%
    \left\lvert\ketbra@firstarg\middle\rangle \! \middle\langle\ketbra@secondarg\right\rvert%
  \fi%
}
\makeatother

\makeatletter
\def\metadef#1#2{%
  \def\metadef@iter##1{\ifx##1;\else \expandafter\newcommand\csname#1\endcsname{#2}\expandafter\metadef@iter\fi}%
  \expandafter\metadef@iter%
}
\makeatother

\metadef{vec#1}{\mathbf #1}abcdefghijklmnopqrstuvwxyz;

\metadef{ten#1}{\mathbf #1}ABCDEFGHIJKLMNOPQRSTUVWXYZ;

\metadef{mat#1}{\mathbf #1}ABCDEFGHIJKLMNOPQRSTUVWXYZ;

\metadef{frak#1}{\mathfrak #1}ABCDEFGHIJKLMNOPQRSTUVWXYZ;

\metadef{c#1}{\mathcal #1}ABCDEFGHIJKLMNOPQRSTUVWXYZ;

\metadef{s#1}{\mathsf #1}ABCDEFGHIJKLMNOPQRSTUVWXYZ;

\ifnum\llncs=0
\newtheorem{all}{Theorem}[section]
\newtheorem{theorem}[all]{Theorem}
\newtheorem{lemma}[all]{Lemma}
\newtheorem{definition}[all]{Definition}
\newtheorem{corollary}[all]{Corollary}

\newtheorem{remark}[all]{Remark}
\newtheorem{claim}[all]{Claim}

\if0
\newtheorem{theorem}{Theorem}
\newtheorem{lemma}{Lemma}
\newtheorem{claim}{Claim}
\newtheorem{corollary}{Corollary}

\newtheorem{remark}{Remark}

\newtheorem{definition}{Definition}

\fi

\fi
\newcommand{\Tr}{Tr}
\newcommand{\TD}{\mathsf{TD}}

\ifnum\comments=1
\newcommand{\mor}[1]{$\ll$\textsf{\color{red} Tomoyuki: { #1}}$\gg$}
\newcommand{\todo}[1]{\noindent \textcolor{red}{(\textbf{TODO:} #1\noindent)}}
\newcommand{\eli}[1]{\noindent \textcolor{orange}{(\textbf{Eli:} #1\noindent)}}
\newcommand{\takashi}[1]{\noindent \textcolor{blue}{(\textbf{Takashi:} #1\noindent)}}
\newcommand{\saachi}[1]{\noindent \textcolor{green}{(\textbf{Saachi:} #1\noindent)}}
\else
\newcommand{\mor}[1]{}
\newcommand{\todo}[1]{}
\newcommand{\eli}[1]{}
\newcommand{\takashi}[1]{}
\newcommand{\saachi}[1]{}
\fi

\newcommand\algo{\mathcal}

\newcommand{\poly}{\mathsf{poly}}
\newcommand{\negl}{\mathsf{negl}}

\newcommand{\randfrom}{\gets}
\newcommand{\U}{U}
\newcommand{\RO}{\mathsf{RO}}
\newcommand{\E}{\mathop{\mathbb{E}}}
\newcommand{\wt}[1]{\widetilde{#1}}

\newcommand{\abs}[1]{\left|#1\right|}

\newcommand{\w}{\omega}

\newcommand{\argmax}{\mathsf{argmax}}
\DeclareMathOperator{\Sim}{Sim}
\newcommand{\norm}[1]{\left\lVert#1\right\rVert}

\newcommand{\N}{\mathbb{N}}
\newcommand{\C}{\mathbb{C}}

\newcommand{\Mint}{\mathsf{Mint}}
\newcommand{\Veri}{\mathsf{Ver}}
\newcommand{\Comm}{\mathsf{Com}}
\newcommand{\Rece}{\mathsf{Rec}}

\ifnum\llncs=0
\fi

\makeatletter
\newcommand\ExpOwsg[4][]{%
  \def\ExpOwsg@firstarg{#1}%
  \ifx\ExpOwsg@firstarg\empty%
    \mathsf{Exp}_{#2,#3}(#4)
  \else%
    \mathsf{Exp}_{#2,#3,#1}(#4)
  \fi%
}
\makeatother

\newcommand{\ol}[1]{\overline{#1}}
\newcommand{\xor}{\oplus}
\newcommand{\A}{\mathcal{A}}

\newcommand{\Samp}{\mathsf{Samp}}
\newcommand{\Rec}{\mathsf{Rec}}
\newcommand{\Com}{\mathsf{Com}}
\newcommand{\first}{\mathsf{QKDFirst}}
\newcommand{\second}{\mathsf{QKDSecond}}
\newcommand{\decode}{\mathsf{QKDDecode}}

\newcommand{\Ver}{\mathsf{Ver}}

\newcommand{\BigPr}[2]{
\Pr\left[
\begin{array}{c}
#1
\end{array}
:
\begin{array}{c}
#2
\end{array}
\right]
}
\newcommand{\BigE}[2]{
\E\left[
\begin{array}{c}
#1
\end{array}
:
\begin{array}{c}
#2
\end{array}
\right]
}

\newcommand{\secpa}{n} %command for the security parameter

\newcommand{\HaarSt}[1]{\mu_{#1}^s}
\newcommand{\HaarUn}[1]{\mu_{#1}}

\newcommand{\SGSamp}{SG_{\text{Samp}}}
\newcommand{\SGSampn}{SG_{\text{Samp},n}}
\newcommand{\SGSampl}{SG_{\text{Samp},\ell}}
\newcommand{\VerSamp}{Ver_{\text{Samp}}}
\newcommand{\VerSampn}{Ver_{\text{Samp},n}}
\newcommand{\VerSampl}{Ver_{\text{Samp},\ell}}
\newcommand{\SimSG}{\text{Sim}_{SG}^{\SGSamp}}
\newcommand{\SimSGn}{\text{Sim}_{SG,n}^{\SGSampn}}
\newcommand{\SimSGl}{\text{Sim}_{SG,\ell}^{\SGSampl}}
\newcommand{\SimMix}{\text{Sim}_{Mix}^{\VerSamp}}
\newcommand{\SimMixn}{\text{Sim}_{Mix,n}^{\VerSampn}}
\newcommand{\SimMixl}{\text{Sim}_{Mix,\ell}^{\VerSampl}}

%% file: abstract.tex
\begin{abstract}
    We construct a unitary oracle relative to which $\mathbf{BQP}=\mathbf{QCMA}$ but 
    quantum-computation-classical-communication (QCCC) commitments and 
    QCCC multiparty non-interactive key exchange exist. 
    We also construct a unitary oracle relative to which $\mathbf{BQP}=\mathbf{QMA}$, but quantum lightning (a stronger variant of quantum money) exists. 
    This extends previous work by Kretschmer [Kretschmer, TQC22], which showed that there is a quantum oracle relative to which $\mathbf{BQP}=\mathbf{QMA}$ but
   pseudorandm unitaries 
%    pseudorandom state generators\mor{pseudorandom unitaries?} (a quantum variant of pseudorandom generators) 
exist. 

    We also show that (poly-round) QCCC key exchange, QCCC commitments, and two-round quantum key distribution can all be used to build one-way puzzles. One-way puzzles are a version of ``quantum samplable'' one-wayness and are an 
    intermediate primitive between pseudorandom state generators and EFI pairs, the minimal quantum primitive. In particular, one-way puzzles cannot exist if $\mathbf{BQP}=\mathbf{PP}$.

    Our results together imply that aside from pseudorandom state generators, there is a large class of quantum cryptographic primitives which can exist even if $\mathbf{BQP} = \mathbf{QCMA}$, 
    but are broken if $\mathbf{BQP} = \mathbf{PP}$. Furthermore, one-way puzzles are a minimal primitive for this class. We denote this class ``CountCrypt''.
\end{abstract}

%% file: intro.tex
\section{Introduction}

Nearly all cryptographic primitives require some form of computational assumption~\cite{FOCS:ImpLub89}. For example, the existence of one-way functions (OWFs) implies that $\mathbf{P}\neq \mathbf{NP}$. However, some primitives seemingly require stronger computational assumptions than others. For instance, it appears that key exchange (KE) requires stronger computational assumptions than OWFs. In particular, OWFs can be built from 
KE~\cite{C:BelCowGol89}, but there exist strong barriers to building KE from OWFs~\cite{STOC:ImpRud89}. In general, the relationships between classical cryptographic primitives are well-understood. Notably, nearly all primitives can be used to build OWFs, and so the existence of any sort of cryptography necessitates that $\mathbf{P}\neq \mathbf{NP}$.

Everything changes once quantum computers enter the picture. It is not hard to define quantum versions of classical cryptographic primitives, where we allow either the protocol description and/or protocol inputs/outputs to be quantum. For instance, ~\cite{C:JiLiuSon18} defines pseudorandom state generators (PRSGs), a quantum analogue of pseudorandom generators where the output is a quantum state, and it is indistinguishable from a Haar random quantum state. There exists a classical oracle relative to which single-copy-secure PRSGs\footnote{A single-copy-secure PRSG is a variant of PRSGs where the adversary receives only a single copy of a pseudorandom state.} 
exist~\cite{STOC:KQST23} or quantumly-computable trapdoor functions exist~\cite{STOC:KreQiaTal25} but $\mathbf{P}=\mathbf{NP}$, or there exists a quantum oracle relative to which pseudorandom unitaries (PRUs) exist but $\mathbf{BQP}=\mathbf{QMA}$~\cite{TQC:Kre21}. This gives evidence that some quantum primitives require weaker assumptions even than OWFs.

Since~\cite{TQC:Kre21,C:MorYam22,C:AnaQiaYue22}, a slew of recent work has begun to map out the relationships between different forms of cryptographic hardness. To quickly summarize recent progress, it seems that most quantum cryptographic primitives fall into one of several buckets:
\begin{enumerate}
    \item ``QuantuMania'': quantum cryptographic primitives which are broken if $\mathbf{BQP} =\mathbf{QCMA}$. This class includes three types of primitives. The first type is just quantumly-secure classical primitives, such as post-quantum classical OWFs or post-quantum classical public-key encryption (PKE). 
    The second type is primitives with ``genuinely quantum'' additional functionalities such as PKE with (publicly verifiable) certified deletion~\cite{TCC:KitNisYam23} and digital signatures with revocable signing keys~\cite{EPRINT:MorPorYam23}.\mor{Maybe only for publicly verifiable certified deletion case?} The third type
    is most non-interactive quantum-computation-classical-communication (QCCC) primitives, where
    locally quantum computing is possible but all communications are classical,
    such as PKE or digital signatures with classical keys, ciphertexts, or signatures~\cite{C:ChuGolGra24}.
    \mor{We can construct QCCC PKE with quantum sk, so shall we change the text to say that all keys are classical?}
    \item ``CountCrypt'': quantum cryptographic primitives which may exist  if $\mathbf{BQP}=\mathbf{QCMA}$, but are broken by a $\mathbf{PP}$ oracle. That is, CountCrypt primitives do not exist if $\mathbf{BQP} = \mathbf{PP}$. 
    This category includes primitives such as PRSGs~\cite{C:JiLiuSon18,TQC:Kre21}, (pure output) one-way state generators (OWSGs)~\cite{C:MorYam22,TQC:MorYam24}, and one-way puzzles (OWPuzzs)~\cite{STOC:KhuTom24,cavalar2023computational}.\footnote{
    OWSGs and OWPuzzs are natural quantum analogues of OWFs. A OWSG is a pair $(\mathsf{Gen},\mathsf{Ver})$ of quantum polynomial-time (QPT) algorithms. $\mathsf{Gen}(k)\to\phi_k$ 
    takes a random classical bit string $k$ as input and outputs a quantum state $\phi_k$.
    The security requires that no QPT adversary, given many copies of $\phi_k$, can output $k'$ such that
    $\mathsf{Ver}(k',\phi_k)\to\top$. A OWPuzz is a pair $(\mathsf{Samp},\mathsf{Ver})$ of algorithms. $\mathsf{Samp}$ is a QPT algorithm that outputs two classical bit strings $s$ and $k$.
    The security requires that no QPT adversary, given $s$, can output $k'$ such that $\mathsf{Ver}(s,k')\to\top$, where $\mathsf{Ver}$ is not necessarily QPT.} 
    There is evidence that these primitives may exist even if $\mathbf{BQP}=\mathbf{QCMA}$~\cite{TQC:Kre21}. 
    On the other hand, post-quantum classical OWFs do not exist if $\mathbf{BQP}=\mathbf{QCMA}$. 
    Thus, it seems plausible that these primitives could be built from weaker assumptions than those necessary for classical cryptography.\footnote{Unlike the classical MicroCrypt, CountCrypt does not collapse to a single primitive. There exists a black-box separation between OWSG and OWPuzz, two CountCrypt primitives~\cite{microsep,bostanci2024oracle}.}
    \item ``NanoCrypt'': quantum cryptographic primitives which can exist even if $\mathbf{BQP}=\mathbf{PP}$, and which we do not know how to break using \textit{any} classical oracle. This category contains exclusively (single-copy) ``quantum output'' primitives. Examples of this category are EFI pairs~\cite{ITCS:BraCanQia23}\footnote{An EFI pair is a two efficiently generatable quantum states that are statistically far but computationally indistinguishable.}, quantum non-interactive bit commitments~\cite{AC:Yan22}, multiparty computation, and single-copy-secure PRSGs~\cite{C:MorYam22}.
    There exists an oracle relative to which these primitives are secure against one-query to an arbitrary classical oracle~\cite{STOC:LomMaWri24}. Nevertheless, it is possible that these primitives all can indeed be broken by (multiple queries to) a $\mathbf{PP}$ oracle, and it is just that no attack has been discovered.  There are also some black-box separations between some primitives in NanoCrypt and those in CountCrypt. In particular, there is an oracle relative to which single-copy-secure PRSGs (and therefore EFI pairs) exist, but (multi-copy secure) PRSGs do not~\cite{chen2024powersinglehaarrandom,EPRINT:AnaGulLin24b}.
\end{enumerate} %%\takashi{I'm curious where mixed OWSG is located. Maybe Nanocrypt? (I'm not sure if we should write it in the paper, but I'm just curious.)}\eli{mixed OWSG with inefficient verification is in NanoCrypt per Batra Jain, but the efficient verification does make it feel different.}
Note that each of these categories also has an associated minimal primitive. Nearly every primitive in QuantuMania implies the existence of efficiently verifiable one-way puzzles (EV-OWPuzzs)~\cite{C:ChuGolGra24}.\footnote{An EV-OWPuzz is a variant of OWPuzzs where $\mathsf{Ver}$ is quantum polynomial-time (QPT).} Like-wise, (inefficiently verifiable) one-way puzzles (OWPuzzs) appear to be minimal for CountCrypt~\cite{STOC:KhuTom24,STOC:KhuTom25}, while EFI pairs seem minimal for NanoCrypt~\cite{TQC:MorYam24,ITCS:BraCanQia23}. From this perspective, we can define QuantuMania as the class of primitives which can be used to build EV-OWPuzzs. We can define CountCrypt as the class of primitives which can be used to build OWPuzzs but not EV-OWPuzzs. And we can define NanoCrypt as the class of primitives which can be used to build EFI pairs but not OWPuzzs. For a figure illustrating the relationships between primitives in these three classes, see~\Cref{fig:graph}.
\input{figure}

While a great amount of progress has been made in recent years, there are a few primitives which we still do not know how to place within these broad categories.

\paragraph*{Interactive QCCC primitives.} One interesting class of quantum primitives are those where communication is required to be done over classical channels, namely, QCCC primitives. This class contains variants of most classical cryptographic primitives, where the only difference is that we allow the protocol itself to perform quantum operations. Examples include classical-communication key exchange (KE), classical-communication digital signatures, and EV-OWPuzz. 
%We will denote such primitives as ``QCCC'' primitives~\cite{EC:ChuLinMah23}, referring to the fact that they use quantum computation and classical communication.\footnote{In the field of entanglement theory, such a setup has been called LOCC (which stands for local operation and classical communication). However, LOCC is usually used for statistical settings, i.e., local parties can do any unbounded computation.} 
\cite{STOC:KhuTom24,C:ChuGolGra24} showed that many QCCC primitives can be used to build EV-OWPuzzs. 
In particular, QCCC versions of encryption, signatures, and pseudorandomness all imply EV-OWPuzzs and so are broken if $\mathbf{BQP} = \mathbf{QCMA}$.
\mor{If the secret key is quantum, they exist, so we should say that all keys are classical.}

However, there are a few notable exceptions to this result. In particular, it is not clear how to build EV-OWPuzzs from QCCC primitives whose security game includes multiple rounds of communication. Key examples of such primitives, which we consider in this work, are QCCC KE and QCCC commitments. 
It is easy to see that QCCC KE and QCCC commitments can be broken with a $\mathbf{PP}$ oracle. Thus, QCCC KE and QCCC commitments should lie in either QuantuMania or CountCrypt. If we look from the perspective of minimal primitives, even less is known. It is unknown whether or not QCCC KE and QCCC commitments can be used to build any of EV-OWPuzzs, OWPuzzs, or EFI pairs\footnote{Note that~\cite{STOC:KhuTom24} constructs OWPuzzs from QCCC commitments with a {\it non-interactive} opening phase. However, this is not enough for our purposes. In particular, our construction of QCCC commitments in the oracle setting has an {\it interactive} opening phase.}.

\paragraph*{2-round quantum key distribution.} One of the most surprising initial results in quantum cryptography is that, with quantum communication, KE can exist \textit{unconditionally}~\cite{Wie83,BB84}\footnote{Assuming the existence of classical authenticated channel.}. 
Such protocols, termed quantum key distribution or QKD, are secure when performed over an authenticated classical channel and an \textit{arbitrary, corrupted} quantum channel.
Formally, a QKD protocol is a protocol where Alice and Bob agree on a shared key over a public authenticated classical channel and a public arbitrary quantum channel. In particular, a good QKD protocol should satisfy three properties. 
\begin{enumerate}
    \item (Correctness): If the protocol is executed honestly, then Alice and Bob output the same key.
    \item (Validity): An adversarial party Eve, who controls the quantum channel, cannot make Alice and Bob output different keys without 
    either of them noticing and outputting $\bot$.
    \item (Security): If either Alice or Bob does not output $\bot$, then Eve can not learn their final key.
\end{enumerate}
A two-round QKD protocol, or 2QKD is simply a QKD protocol where Alice and Bob each send at most one message. It is known that 2QKD protocols can be built from post-quantum OWFs~\cite{C:KMNY24,C:MalWal24}.

%\mor{This paragraph is not clear to me. Isn't it obvious that QKD is secure under corrupted quantum channel by definition?
%Also, what do you mean by ``such a protocol is called a key distribution protocol"? Are you distinguishing key exchange and key distribution protocol?}\eli{classical key exchange assumes all communication is authenticated, since otherwise it is impossible. But you can't authenticate quantum channels, so this model doesn't make much sense in the quantum setting.}

Interestingly, QKD can be performed in three rounds and still maintain information-theoretic security~\cite{C:GarYueZha17}\footnote{Note that in this protocol, only one party learns if the execution succeeds.}. On the other hand, it is known that 2-round QKD (2QKD) is impossible information-theoretically~\cite{C:MalWal24}. However, there is no other known lower bound on the computational hardness of 2QKD. The best known upper bound is that 2QKD can be constructed from 
OWFs~\cite{C:KMNY24,C:MalWal24}. It is entirely an open question whether 2QKD lies in QuantuMania, CountCrypt, or NanoCrypt.

%\paragraph{Public key encryption.}
%\mor{I will add some texts}

\paragraph*{Quantum money/quantum lightning.} Quantum money~\cite{Wie83,C:JiLiuSon18,STOC:AarChr12} is one of the most fundamentally ``quantum'' cryptographic primitives. Quantum money essentially models digital cash. In a quantum money scheme, a bank has the ability to mint quantum states which represent some currency. Money states should be verifiable: either the bank or other users should be able to tell if a money state is valid. Money states should also be uncloneable: it should not be possible for users to take several copies of a money state and turn it into more copies. Note that quantum money is fundamentally impossible if we require that states are classical. This is because it is easy to copy any classical string.

For the purposes of this work, we will consider a stronger variant of quantum money, known as quantum lightning~\cite{JC:Zhandry21}. A quantum lightning scheme consists of a minter and a verifier. The minter produces a money state along with a corresponding serial number. The verifier, given a money state and a serial number, checks if the serial number matches. Security says that no adversary can construct two money states which verify under the same serial number. This is a stronger definition than standard quantum money, since adversary can 
choose which state they are trying to clone.

It is again easy to see that quantum lightning is broken if $\mathbf{BQP} = \mathbf{PP}$, since given the verification algorithm a QPT algorithm querying the
$\mathbf{PP}$ oracle can synthesize a valid state~\cite{STOC:KhuTom25,cavalar2023computational,CCC:INNRY22}. Furthermore, recent work has shown that OWPuzzs can be built from quantum lightning~\cite{STOC:KhuTom25}, but given a quantum lightning scheme, it is not known how to build an EV-OWPuzz. Thus, it seems that quantum lightning lies either in CountCrypt or QuantuMania, but it is not clear which is the case. Note that the same reasoning applies to quantum money schemes as well.

\subsection{Our Results} 
The major question of our work is the following.
\begin{center}
    \textit{Do these primitives (QCCC KE, QCCC commitments, 2QKD, quantum lightning) lie in CountCrypt or QuantuMania?}
\end{center}
More specifically,
\begin{center}
    \textit{Do these primitives exist even if $\mathbf{BQP}=\mathbf{QMA}$ or $\mathbf{BQP}=\mathbf{QCMA}$?}
    Do these primitives imply EV-OWPuzzs or OWPuzzs?
\end{center}

In this paper, we show the following results.
\begin{enumerate}
    \item There exists a unitary oracle relative to which $\mathbf{BQP}=\mathbf{QCMA}$, but 
    QCCC commitments (with a single-message commit phase and a two-message reveal phase) and multiparty QCCC non-interactive KE (NIKE) (and therefore 2QKD as well) exist.
    \item  There exists a unitary oracle relative to which $\mathbf{BQP}=\mathbf{QMA}$, but quantum lightning exists.
    \item If any of (poly-round) QCCC KE, QCCC commitments, or 2QKD exist, then OWPuzzs exist. (It is already known that quantum lightning can be used to build 
    OWPuzzs~\cite{STOC:KhuTom24}. A concurrent work also showed that 2QKD implies OWPuzz through similar techniques~\cite{QRZ25}.)
\end{enumerate}
Together, these results imply that quantum lightning, quantum money, QCCC (NI)KE and QCCC commitments, as well as 2QKD are all CountCrypt primitives.

%% file: figure.tex
%eprint

\usetikzlibrary{positioning} % for position relative to node
\usetikzlibrary{calc} % for computing coordinates
\usetikzlibrary {quotes}
\tikzset{>=latex} % set default arrow head as latex

% TIKZ STUFF
\tikzstyle{mysmallarrow}=[->,black,line width=1.6]
\tikzstyle{mydbarrow}=[<->,black,line width=1.6]
% JET CATEGORIES with straight lines
\begin{figure}
\begin{center}
    \begin{tikzpicture}[scale=0.9,every edge quotes/.style = {font=\footnotesize,fill=white}]
    % Nodes
      \def\h{-2.0} % vertical space between rows
      \def\w{2.6} % horizontal space between two main branches

        % Big Quantum
        \node[] (unc) at (3.5*\w,0.3*\h) {PKE with certified deletion};
        \node[] (qcpke) at (2.5*\w,0.8*\h) {QCCC PKE};
        \node[] (pqOWF) at (1.5*\w,\h) {pqOWF};
        \node[] (evowp) at (3*\w,2*\h) {EVOWPuzz};

        \node[] (bigquant) at (-0.5*\w,\h) {\textbf{QuantuMania}};
        \node[] (bigquant2) at (-0.5*\w,1.2*\h) {$\mathbf{BQP} \neq \mathbf{QCMA}$};
    
        % CountCrypt
        \node[] (prs) at (0.6*\w,3.1*\h) {PRSG};
        \node[] (owsg) at (0.8*\w,4.1*\h) {OWSG};
        \node[] (qmon) at (4.2*\w,3.6*\h) {Quantum Lightning};
        \node[] (qcke) at (1.3*\w,3*\h) {QCCC KE};
        \node[] (qccom) at (3.6*\w,3*\h) {QCCC Commitments};
        \node[] (2qkd) at (2.5*\w, 3.5*\h) {2QKD};
        \node[] (owp) at (3*\w,5*\h) {OWPuzz};
        
        \node[] (count) at (-0.5*\w,3.5*\h) {\textbf{CountCrypt}};
        \node[] (count2) at (-0.5*\w,3.7*\h) {$\mathbf{BQP} \neq \mathbf{PP}$};

        % NanoCrypt
        \node[] (1prs) at (\w,6*\h) {1-PRSG};
        \node[] (com) at (4*\w,7*\h) {Bit commitments};
        \node[] (MPC) at (1*\w,7*\h) {MPC};
        \node[] (EFI) at (3*\w, 7*\h) {EFI};
        
        \node[] (nano) at (-0.5*\w,6.5*\h) {\textbf{NanoCrypt}};
        
    % Arrows

        % Big Quantum
        \draw[mysmallarrow=black] (unc) edge["\tiny{\cite{TCC:KitNisYam23}}"] (evowp);
        \draw[mysmallarrow=black] (pqOWF) edge["\tiny{\cite{C:ChuGolGra24}}"] (evowp);
        \draw[mysmallarrow=black] (qcpke) edge["\tiny{\cite{C:ChuGolGra24}}"] (evowp);
        \draw[] (-\w,2.5*\h) -- (4.5*\w,2.5*\h);

        % between
        \draw[mysmallarrow,dashed,color=red][bend right=10] (pqOWF) edge["\tiny{\cite{C:ACCFLM22}}"] (qcke);
        \draw[mysmallarrow][bend right=40] (pqOWF) edge["\tiny{\cite{C:JiLiuSon18,TCC:BraShm19}}"] (prs);
        \draw[mysmallarrow] (pqOWF) edge["\tiny{\cite{C:KMNY24,C:MalWal24}}",pos=0.8] (2qkd);
        \draw[mysmallarrow,dashed,color=blue] (qcke) edge (evowp);
        \draw[mysmallarrow,dashed,color=red][bend left=15] (prs) edge["\tiny{\cite{C:ChuGolGra24}}",pos=0.55] (evowp);
        \draw[mysmallarrow,dashed,color=blue] (qccom) edge (evowp);
        \draw[mysmallarrow,dashed,color=blue][bend right=85] (qmon) edge (evowp);
        \draw[mysmallarrow,dashed,color=blue][bend right=15] (2qkd) edge (evowp);
        \draw[mysmallarrow,color=blue] (2qkd) edge (owp);
    
        % CountCrypt
        \draw[mysmallarrow=black] (prs) edge["\tiny{\cite{C:MorYam22,TQC:MorYam24,cavalar2023computational}}"] (owsg);
        \draw[mysmallarrow=black][bend right=10] (owsg) edge["\tiny{\cite{STOC:KhuTom24}}",pos=0.35] (owp);
        \draw[mysmallarrow,dashed,color=red][bend right=15] (owp) edge["\tiny{\cite{microsep,bostanci2024oracle}}",pos=0.56] (owsg);
        \draw[mysmallarrow=black,color=blue][bend left=15] (qcke) edge (owp);
        \draw[mysmallarrow=black,color=blue] (qccom) edge (owp);
        \draw[mysmallarrow=black] (qmon) edge["\tiny{\cite{STOC:KhuTom25}}"] (owp);
        \draw[] (-\w,5.5*\h) -- (4.5*\w,5.5*\h);

        % between
        \draw[mysmallarrow=black][bend right=30] (prs) edge (1prs);
        \draw[mysmallarrow,dashed,color=red][bend left=70] (1prs) edge["\tiny{\cite{chen2024powersinglehaarrandom,EPRINT:AnaGulLin24b}}"] (prs);
        \draw[mysmallarrow,dashed,color=red][bend right=10] (1prs) edge["\tiny{\cite{EPRINT:AnaGulLin24b}}",pos=0.3] (qcke);
        \draw[mysmallarrow,dashed,color=red] (1prs) edge["\tiny{\cite{EPRINT:AnaGulLin24b}}",pos=0.3] (qccom);
        \draw[mysmallarrow] (owp) edge["\tiny{\cite{STOC:KhuTom24}}"] (EFI);
        \draw[mysmallarrow,dashed,color=red] (1prs) edge["\tiny{\cite{bostanci2024oracle}}",pos=0.65] (owsg);

        % NanoCrypt
        \draw[mysmallarrow] (1prs) edge (EFI);
        \draw[mydbarrow] (MPC) edge["\tiny{\cite{ITCS:BraCanQia23}}"] (EFI);
        \draw[mydbarrow] (com) edge (EFI);
    \end{tikzpicture}
\end{center}
\caption{A graph of some known implications between primitives in QuantuMania, CountCrypt, and NanoCrypt. Dashed lines represent black-box separations. Blue lines are new in our work.
pqOWF means post-quantum (i.e., quantumly-secure) OWFs. OWSG refer to pure state one-way state generators~\cite{C:MorYam22}.
}
\label{fig:graph}
\end{figure}
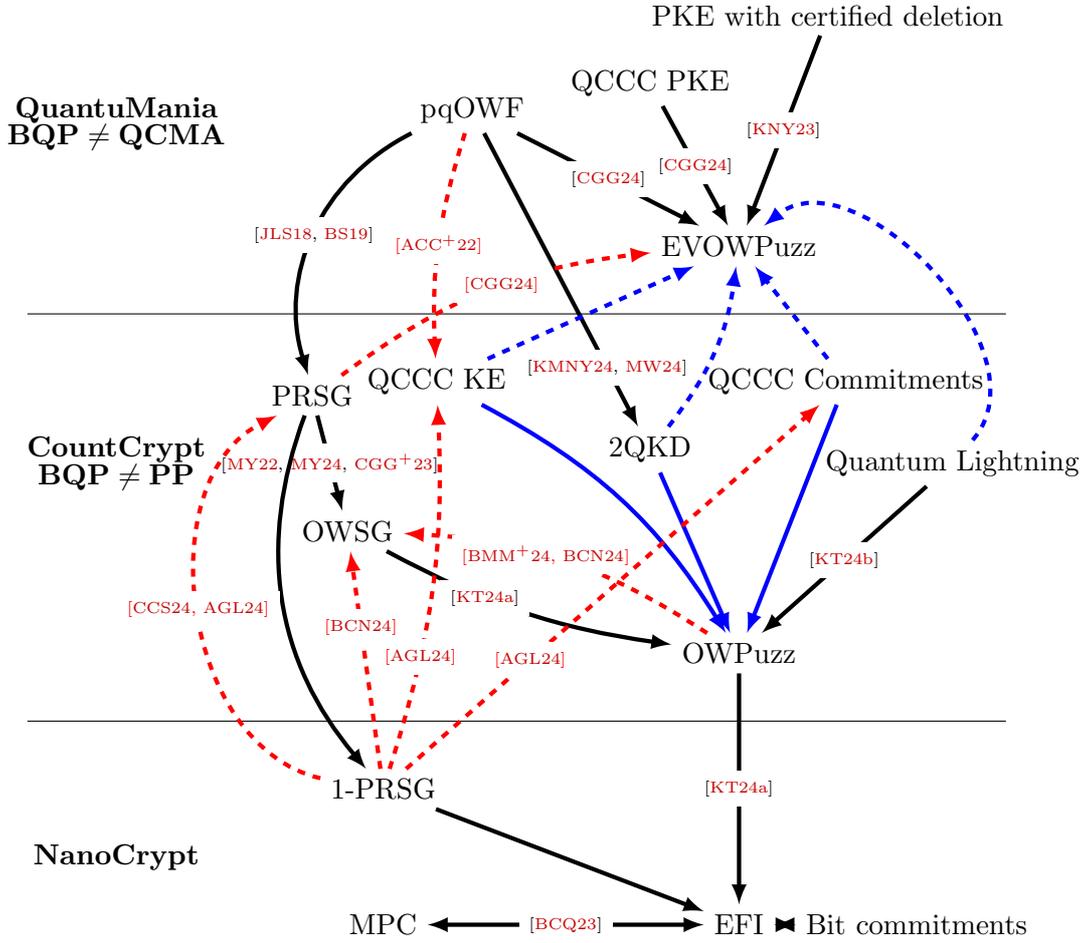

%% file: techoverview.tex
\subsection{Technical Overview}
Here we provide a high-level overview of our proofs.

\paragraph*{Oracle relative to which QCCC NIKE exist and $\mathbf{BQP}=\mathbf{QCMA}$.}
The key idea behind this oracle is that it functions as an idealized NIKE, but where the internal state is forced to be quantum. As a toy example, consider the simple idealized NIKE protocol defined as follows.\footnote{Although we can actually construct multiparty NIKE, here for simplicity we explain the two-party case. 
The multiparty case is similar.}
\begin{enumerate}
    \item Let $\RO:\{0,1\}^\secpa\times \{0,1\}^\secpa\to \{0,1\}^\secpa$ be a random function accessible only to Alice and Bob.
    \item Alice chooses $x$ and sends $x$ to Bob.
    \item Bob chooses $y$ and sends $y$ to Alice.
    \item Alice and Bob both output $\RO(x,y)$.
\end{enumerate}
Note that this protocol is secure against any adversary who does not have access to $\RO$. This is because just given $x$ and $y$, the adversary has no way to compute $\RO(x,y)$.

In order to define an oracle relative to which this NIKE protocol is secure, we need an oracle which allows Alice and Bob to query $\RO(x,y)$, while other parties cannot. To do this, we associate to each $x$ and $y$ a quantum key $\ket{\phi_x}$ and $\ket{\phi_y}$. We then guarantee that the only way for a party to compute $\ket{\phi_x}$ is if it chose $x$ itself. We mandate that in order to query $\RO(x,y)$, a party must have access to either $\ket{\phi_x}$ or $\ket{\phi_y}$.

As a first attempt, we will define a CPTP oracle (that is, a quantum oracle which is allowed to make measurements) as follows. 
Our oracle will be defined respective to a Haar random unitary $\U$ and a random function $\RO:\{0,1\}^\secpa \times \{0,1\}^\secpa \to \{0,1\}^n$.
Our oracle consists of the following three oracles, $SG^{CPTP}$, $Mix^{CPTP}$, and $PSPACE$:\footnote{Precisely speaking, we define 
sequences $\{SG^{CPTP}_n\}_{n\in\mathbb{N}}$ and $\{Mix^{CPTP}_n\}_{n\in\mathbb{N}}$ of oracles defined for each $n\in\mathbb{N}$, but for simplicity, here we ignore that point.}
\begin{enumerate}
    \item $SG^{CPTP}$: samples a bit string $x$ uniformly at random and outputs $(x,\ket{\phi_x})$. Here $\ket{\phi_x}\coloneqq \U|x\rangle$.
    \item $Mix^{CPTP}(x,y,\ket{\psi})$: applies $\U^{\dagger}$ on $\ket{\psi}$, measures the state in the computational basis, and outputs $\RO(x,y)$ if and only if the measurement result is $x$ or $y$.
    \item $PSPACE$: an oracle for a $\mathbf{PSPACE}$-complete problem.\footnote{We use the normal font for the oracle, and the bold font for the complexity class.}
\end{enumerate}

This leads to the following natural NIKE protocol.
\begin{enumerate}
    \item Alice runs $SG^{CPTP}\to (x,\ket{\phi_x})$ and sends $x$ to Bob.
    \item Bob runs $SG^{CPTP}\to (y,\ket{\phi_y})$ and sends $y$ to Alice.
    \item Alice outputs $Mix^{CPTP}(x,y,\ket{\phi_x}) \to \RO(x,y)$.
    \item Bob outputs $Mix^{CPTP}(x,y,\ket{\phi_y}) \to \RO(x,y)$.
\end{enumerate}

Security follows from the fact that given $x$ and $y$, it is impossible for an eavesdropper to compute $\ket{\phi_x}$ or $\ket{\phi_y}$, since with high probability queries to $SG^{CPTP}$ will never return these values. However, without $\ket{\phi_x}$ or $\ket{\phi_y}$, it is impossible to query $\RO(x,y)$, and therefore the agreed key appears uniformly random. Note that since this security holds information-theoretically, it also holds against adversaries with access to the $PSPACE$ oracle.

\paragraph{Making the oracle unitary.} While the oracles $SG^{CPTP},Mix^{CPTP},PSPACE$ do indeed give an oracle where QCCC cryptography exists but $\mathbf{BQP}=\mathbf{QCMA}$, the fact that the oracles are CPTP is a major limitation. In particular, CPTP oracle separations only rule out black-box reductions which do not purify the adversary or query its inverse. These are very common techniques in the quantum setting. Thus, the quantum cryptography community has a strong preference for unitary (or even better, classical) oracle separations, and much work has been done to avoid weaker models~\cite{goldin2025translating,microsep,bostanci2024oracle}.
\mor{It would be better to explain why a CPTP oracle is not sufficient and we want to improve it to a unitary oracle.}

Our first observation is that the oracles $SG^{CPTP}$ and $Mix^{CPTP}$ have natural unitary ``purifications''.

\begin{enumerate}
    \item $SG$: the unitary which swaps the state $\ket{0}$ with the state $|\psi\rangle\coloneqq\sum_k\ket{k}\ket{\phi_k}$.
    More formally, it is the unitary, $|\psi\rangle\langle0|+|0\rangle\langle\psi|+(I-|0\rangle\langle0|-|\psi\rangle\langle\psi|)$.\footnote{For the notational simplicity, for an $n$, we often write $|0^n\rangle$ just as $|0\rangle$.}
    \item $Mix$: the unitary\footnote{$X$ is the Pauli $X$ operator. For an $n$-bit string $x$, $X^x\coloneqq \bigotimes_{i=1}^n X^{x_i}$.}
        \begin{align}
    &\sum_{x,y}|x\rangle\langle x|\otimes|y\rangle\langle y|\otimes(|\phi_x\rangle\langle\phi_x|+|\phi_y\rangle\langle\phi_y|)\otimes X^{\RO(x,y)}\otimes X    \\
    &+\sum_{x,y}|x\rangle\langle x|\otimes|y\rangle\langle y|\otimes(I-|\phi_x\rangle\langle\phi_x|-|\phi_y\rangle\langle\phi_y|)\otimes I\otimes I    
    \end{align}
   acting on five registers.
    Intuitively, it works as follows.
   If the third register is $|\phi_x\rangle$ or $|\phi_y\rangle$ when the first register is $x$ and the second register is $y$,
   then it adds $\RO(x,y)$ to the fourth register, and flips the bit of the fifth register: 
     \begin{align}
    \ket{x}\ket{y}\ket{\phi_x}\ket{z}\ket{b}&\mapsto \ket{x}\ket{y}\ket{\phi_x}\ket{z\xor \RO(x,y)}\ket{\ol{b}},\\
    \ket{x}\ket{y}\ket{\phi_y}\ket{z}\ket{b}&\mapsto \ket{x}\ket{y}\ket{\phi_y}\ket{z\xor \RO(x,y)}\ket{\ol{b}},
    \end{align}
   Otherwise, it acts as the identity.
\end{enumerate}

We will show the following.

\begin{theorem}[\Cref{thm:keexists} restated]
\label{thm:MQNIKE}
QCCC NIKE exists relative to $(SG,Mix,PSPACE)$.
\end{theorem}
Note that two-party QCCC NIKE is a special case of 2QKD where no quantum channel is used.
We therefore obtain the following corollary.
\begin{corollary}
    2QKD exists relative to $(SG,Mix,PSPACE)$.
\end{corollary}
\begin{remark}
    Note that by allowing $Mix$ to take in an arbitrary number of inputs, we can define new oracles such that multi-party NIKE exists but $\mathbf{BQP}=\mathbf{QCMA}$. We will omit this argument, since it follows analogously to~\Cref{thm:MQNIKE}.
\end{remark}

The key idea to show \Cref{thm:MQNIKE} is that we can simulate $SG$ and $Mix$ by the following CPTP oracles:
\begin{enumerate}
    \item On initialization, sample a fresh random unitary $\wt{U}$. Sample strings $k_1,\dots,k_t$.
    \item $\SGSamp$: on the $i$th query, output 
    $|k_{i}\rangle\otimes(\wt{U}|k_{i}\rangle)$.
    \item $\VerSamp$: apply the unitary
    \begin{align}
    \sum_{i\in[t]} |k_i\rangle\langle k_i|\otimes \wt{U}|k_i\rangle\langle k_i|\wt{U}^\dagger \otimes X   
    +\left(I-\sum_{i\in[t]} |k_i\rangle\langle k_i|\otimes \wt{U}|k_i\rangle\langle k_i|\wt{U}^\dagger\right) \otimes I,
    \end{align}
    which, intuitively, works as follows: If the second register is
    $\wt{U}|k_i\rangle$ when the first register is $k_i$, then it flips the bit of the third register
    as
    \begin{align}
    \ket{k_i}\otimes \wt{U}|k_i\rangle \otimes \ket{b} \mapsto \ket{k_i}\otimes \wt{U}|k_i\rangle \otimes \ket{\ol{b}}.
    \end{align}
    Otherwise, it acts as the identity.
\end{enumerate}
\if0
    \mor{Below, you define oracles for each $n$. However, in the all above arguments, you omitted $n$ for simplicity.
    I suggest to do the same thing here as well.}
\begin{enumerate}
    \item On initialization, sample a fresh random family $\{\wt{U}_n\}$\mor{$\wt{U}$}. Sample strings $k_1,\dots,k_t$.
    \item $\SGSamp{_{n}}$\mor{$\SGSamp$}: on the $i$th query, output 
    $\ket{k_{i}} \wt{U}_n\ket{k_{i}}$.
    \mor{$|k_{i}\rangle\otimes(\wt{U}|k_{i}\rangle)$}
    \item $\VerSamp{_{n}}$\mor{$\VerSamp$}: apply the map which, for $i\in [t]$,
    \begin{align}
    \ket{k_i}\otimes \wt{U}_n\ket{k_i} \otimes \ket{b} \mapsto \ket{k_i}\otimes \wt{U}_n\ket{k_i} \otimes \ket{\ol{b}}
    \end{align}
    \mor{$\ket{k_i}\otimes \wt{U}|k_i\rangle \otimes \ket{b} \mapsto \ket{k_i}\otimes \wt{U}|k_i\rangle \otimes \ket{\ol{b}}$}
    and acts as identity on all other \mor{orthogonal} states.
\end{enumerate}
\fi
That is, there exists a simulator $(\SimSG,\SimMix)$ which is indistinguishable from $(SG,Mix)$ under bounded query access, even against adversaries with non-uniform advice that can depend on $U$ and $\RO$.

In particular, to analyze an adversary for our NIKE protocol, we can replace all instances of $(SG,Mix)$ in the security game with $(\SimSG,\SimMix)$. Now, a similar argument to before can be used to show security.

\paragraph{Constructing the simulator.}
Our goal is to show the following theorem.
\begin{theorem}\label{thm:keylemmatechoverview}(Lemma \ref{keylem}, restated)
    With high probability over the oracles $U$ and $\RO$, there exists an efficient simulator such that for all oracle algorithms $\A$ 
    and poly-length advice strings $c$, 
    \begin{align}
    \abs{\Pr[\A^{SG,Mix,PSPACE}(1^\secpa,c)\to 1] - \Pr[\A^{\SimSG,\SimMix,PSPACE}(1^\secpa,c)\to 1]} \leq \negl(n).
    \end{align}
\end{theorem}

To show the theorem, we
first take inspiration from the super-concentration bound for Haar random unitaries from~\cite{TQC:Kre21}, which observes that with all but doubly exponential probability over $\U$, any efficient algorithm querying $\U$ acts identically if the algorithm instead samples $\U$ itself. Thus, by a union bound over all non-uniform advice strings, $\A^{SG,Mix,PSPACE}(1^\secpa,c)$ acts the same on all advice when we replace $\U$ with a freshly sampled Haar random unitary $\wt{\U}$. 

Next, it is known from~\cite{ITCS:Zhandry24b,microsep,goldin2025translating} that for a Haar random state $\ket{\phi}$, the map
\begin{align}
\ket{0}\mapsto \ket{\phi}
\end{align}
can be simulated using many copies of $\ket{\phi}$.
A straightforward generalization of this approach then shows that the map
\begin{align}
\ket{0}\mapsto \sum_{k}\ket{k}\ket{\phi_k}
\end{align}
can be simulated using many samples of $(k,\ket{\phi_k})$. In particular, $SG$ utilizing a fresh random unitary $\wt{U}$ can be simulated using $\SGSamp$.

Next we observe that since $\A^{\SimSG,Mix,PSPACE}(1^\secpa,c)$ can never receive $\ket{\phi_k}$ for a $k$ not sampled by $\SimSG$, we can enforce $Mix$ to always ignore inputs not produced by $\SimSG$.

The final observation we make is that we can replace $\RO(k_1,k_2)$ for $k_1,k_2$ sampled by $\SGSamp$ with random values. The key intuition is as follows: \textit{if you can only query a fixed random function at random values, then that fixed random function appears like a uniform random function}. We call this a concentration bound for random functions on random inputs.
We will explain more details about it in the next paragraph.

Hence we can simulate $Mix$ by using $\VerSamp$ to coherently check whether the inputs are $(x,y,\ket{\phi_x})$ or $(x,y,\ket{\phi_y})$ for some $x,y$ sampled by $\SGSamp$, and then if so outputting some fixed $z_{x,y}$ sampled uniformly at random during initialization.

Putting these together, we construct the simulator $(\SimSG,\SimMix)$. For full details on the simulator, see~\Cref{sec:unisep}.

\paragraph*{Concentration bound for random functions.} 
%One key step in this argument is the so-called concentration bound for random oracles. 
Let us explain here more details about the concentration bound for random functions.
That is, with high probability over the choice of some fixed random oracle $\RO$, $\RO$ will be indistinguishable under random query access from a freshly sampled random oracle $\wt{\RO}$. The intuition here stems from the similarity between a fixed random oracle and the auxiliary-input random oracle model (AI-ROM)~\cite{C:Unruh07}. 

The AI-ROM is a variant of the random oracle model which allows adversaries to have non-uniform advice \textit{depending on the chosen random oracle}. One of the most useful ways to prove things in the AI-ROM is the so-called presampling technique~\cite{C:Unruh07,EC:CDGS18,TCC:GLLZ21}. A bit-fixing, or presampling, oracle is a random function where its evaluation on polynomially many fixed points is chosen adversarially. The presampling technique states that any protocol secure with bit-fixing oracles is also secure in the AI-ROM. 

When adversaries have quantum query access to the oracle, presampling becomes a bit more complicated.~\cite{TCC:GLLZ21} shows that in the quantum setting, any protocol secure with a so-called ``postselected'' oracle is also secure in the AI-ROM. In particular, a ``postselected'' oracle is an oracle which has been sampled from the conditional distribution satisfying some efficient quantum predicate.

Intuitively, the only way to distinguish a fixed random function $\RO$ from a fresh $\wt{\RO}$ is by hardcoding some information about $\RO$ into the distinguisher. This is morally equivalent to giving the distinguisher non-uniform advice depending on $\RO$, i.e. the AI-ROM. In fact, a straightforward argument shows that presampling/postselecting lemmas also hold for fixed random functions under classical/quantum query access respectively. See~\Cref{cor:fixedpresampling} for details.

It thus remains to show that for a postselected random oracle, queries to random points look random. This follows from a straightforward application of the compressed oracle framework~\cite{C:Zhandry19}. In this framework, a random oracle is instead represented by a database containing a superposition of lists of queried points. In particular, after postselection the database can only contain at most a polynomial number of points, so any point chosen at random will necessarily be absent from the database. Queries to points outside of the database will always return uniform random values, and therefore queries to random points on a postselected oracle will return uniform random values.

\paragraph*{$\mathbf{BQP}=\mathbf{QCMA}$.}
Since we know that NIKE exists relative to $(SG,Mix,PSPACE)$, it remains to show that 
\begin{theorem}[\Cref{thm:bqpqcma} restated]
\begin{align}
    \mathbf{BQP}^{SG,Mix,PSPACE}=\mathbf{QCMA}^{SG,Mix,PSPACE}.
\end{align}
\end{theorem}

In particular, we show that for any input $x$ and oracle verifier $\Ver^{SG,Mix,PSPACE}$ for a language in $\mathbf{QCMA}^{SG,Mix,PSPACE}$, there exists a verifier $\wt{\Ver}$ such that for all $x,w$,
\begin{align}
\Pr[\wt{\Ver}^{PSPACE}(x,w) \to 1] \approx \Pr[\Ver^{SG,Mix,PSPACE}(x,w)\to1].
\end{align}
Thus, given $x$ and $\Ver^{SG,Mix,PSPACE}$, we can use the $PSPACE$ oracle to check if $\wt{\Ver}^{PSPACE}$ has a witness $w$ accepting $x$. Since $\wt{\Ver}^{PSPACE}$ acts similarly to $\Ver^{SG,Mix,PSPACE}$, this also reveals whether $\Ver^{SG,Mix,PSPACE}$ has a witness accepting $x$. In other words, we can use the $PSPACE$ oracle to decide languages in $\mathbf{QCMA}^{SG,Mix,PSPACE}$.

The proof of this result is a straightforward application of~\Cref{thm:keylemmatechoverview}. However, some care must be taken to simulate $SG,Mix$ on small domains. Following the approach of~\cite{TQC:Kre21}, for $\ell = O(\log n)$ we construct simulators to emulate $SG,Mix$ efficiently by applying a learning algorithm to learn their behavior up to minimal error.

\paragraph*{Existence of QCCC commitments.}
We can observe that under the same oracle, we can also build QCCC commitments. Our protocol goes as follows.
\begin{enumerate}
    \item The committer queries $|0\rangle$ to the unitary oracle $SG$ to generate
    $SG\ket{0}=\sum_k|k\rangle|\phi_k\rangle$ and measures the first register to obtain the pair $(x,\ket{\phi_x})$ of the measurement result $x$ and
    the post-measurement state $|\phi_x\rangle$. It repeats this process until the first bit of $x$ is the bit that it is trying to commit to. The committer then sends all but the first bit of $x$ to the receiver, which we will denote $x_{>1}$.
    \item To open to a bit $b$, the receiver queries $|0\rangle$ to $SG$ to generate $SG|0\rangle=\sum_k|k\rangle|\phi_k\rangle$,
    measures the first register, obtains $(y,\ket{\phi_y})$, and sends $y$ to the committer.\footnote{If we wish, this step can be done in the commit phase instead of the opening phase. 
    One advantage of doing it in the opening phase is that the receiver does not need to store state long term and the commit phase becomes non-interactive.}
    \item Upon receiving $y$, the committer queries 
    $\ket{x}\ket{y}\ket{\phi_x}\ket{0}\ket{0}$ 
    to the unitary oracle $Mix$ to generate
    $Mix\ket{x}\ket{y}\ket{\phi_x}\ket{0}\ket{0}$, 
    measures the fourth register to get the result $c$, and sends $c$ to the receiver.
    \item The receiver sets $x = (b, x_{>1})$, queries 
    $\ket{x}\ket{y}\ket{\phi_y}\ket{0}\ket{0}$ 
    to the unitary oracle $Mix$ to generate
    $Mix\ket{x}\ket{y}\ket{\phi_y}\ket{0}\ket{0}$, 
    measures the fourth register, and
    accepts if the result is $c$.
\end{enumerate}
Hiding of this protocol is trivial, since the committer just sends a random string in the commit phase. Binding is slightly more difficult, but the proof follows very similar lines to the security of NIKE. In particular, the only way that the committer can open to a bit $b$ is if it knows $\ket{\phi_{b,x_{>1}}}$. However, it is impossible for the receiver to know both $\ket{\phi_{0,x_{>1}}}$ and $\ket{\phi_{1,x_{>1}}}$, since it would have needed to receive both of these separately from $SG$. Thus, once the committer knows $\ket{\phi_{b,x_{>1}}}$, it can not change its mind and decide to open to $\ol{b}$.

We therefore oftain the following theorem.
\begin{theorem}[\Cref{thm:bitcommitment}, restated]
    QCCC commitments exist relative to $(SG,Mix,PSPACE)$.
\end{theorem}

\paragraph*{Existence of quantum lightning.}
Note that it is easy to see that, for any key $k$, no adversary can obtain two copies of $\ket{\phi_k}$. 
This follows immediately from the fact that the oracle $\SGSamp$
only ever returns at most one copy of $\wt{U}\ket{k}$ for each $k$ sampled during initialization. 
\mor{This explanaion is not clear. Please improve it. In particular, we did not explain details of $\SimSG$, so it is not immediate for readers who do not know the definition of $\SimSG$.} 
This leads to the natural quantum lightning protocol as follows.
\begin{enumerate}
    \item $\textbf{Mint}$: Query $|0\rangle$ to $SG$ to generate $SG\ket{0}=\sum_k|k\rangle|\phi_k\rangle$, measure the first register to obtain 
    the pair $(k,\ket{\phi_k})$ of the measurement result $k$ and the post-measurement state $|\phi_k\rangle$.
    Output $k$ as the serial number and $|\phi_k\rangle$ as the bolt state.
    \item $\textbf{Verify}(k,\ket{\psi})$: Query
    $\ket{k}\ket{0}\ket{\psi}\ket{0}\ket{0}$ 
    to $Mix$ to obtain
    $Mix\ket{k}\ket{0}\ket{\psi}\ket{0}\ket{0}$, 
    and
    measure the last register. Accept if the result is 1, and reject otherwise.
\end{enumerate}

The security goes via a reduction to the complexity-theoretic no-cloning theorem defined by~\cite{STOC:AarChr12}. In particular, this theorem states that with only polynomially many queries to a verification oracle, it is impossible to clone a Haar random state with non-negligible success probability. To clone a Haar random state $\ket{\psi}$ using a quantum lightning breaker, we will run the breaker on $(\SimSG,\SimMix)$, with a random $U\ket{k_i}$ replaced by $\ket{\psi}$. With probability approximately $\frac{1}{\#\text{ queries made by the lightning breaker}}$, the quantum lightning breaker will clone $\ket{\psi}$, violating complexity-theoretic no-cloning.

\begin{theorem}[\Cref{thm:lightning}, restated]
    Quantum lightning exists relative to $(SG,Mix,PSPACE)$.
\end{theorem}

\begin{remark}
    Note that for the setting of quantum lightning, the random oracle is entirely redundant. Thus, we can replace $Mix$ with an oracle $Ver \ket{k}\ket{\phi_k}\ket{b}\mapsto \ket{k}\ket{\phi_k}\ket{\ol{b}}$. Note that relative to this oracle, the proof of~\cite{TQC:Kre21} can be easily adapted to show that $\mathbf{BQP}^{SG,Ver,PSPACE} = \mathbf{QMA}^{SG,Ver,PSPACE}$.
\end{remark}

\paragraph*{Building OWPuzzs from QCCC KE.}
To construct OWPuzzs from (poly-round) QCCC KE, we look to the construction of OWFs from classical KE~\cite{FOCS:ImpLub89}. In particular, this paper constructs distributional OWFs from KE, and shows that distributional OWFs and OWFs are equivalent. 
A distributional OWF is a OWF which is hard to \textit{distributionally} invert. 
That is, given $y = f(x)$, it is hard to sample from the set $f^{-1}(y)$.

Recent work~\cite{C:ChuGolGra24} analogously defines a distributional variant of OWPuzzs. 
In particular, a distributional OWPuzzs consists of a sampler producing puzzle-key pairs $(s,k)$ 
such that given a puzzle $s$, it is hard to sample from the conditional distribution over keys $k$. It turns out that distributional OWPuzzs and OWPuzzs are equivalent~\cite{C:ChuGolGra24}.

It is then not hard to build a distributional OWPuzzs from a QCCC KE. In particular, the puzzle will be the transcript (i.e., the sequence of messages exchanged) between Alice and Bob
while the OWPuzz key will be the corresponding shared key of the KE produced by Alice. 
For a random transcript, the conditional distribution over Alice's key will be close to the constant distribution, as otherwise Bob will disagree with Alice on the shared key. Distributionally inverting this puzzle is enough to recover the shared key and thus break the KE.

\begin{theorem}[\Cref{thm:ketoowpuzz}, restated]
If there exists (poly-round) QCCC KE, then there exist OWPuzzs.
\end{theorem}

\paragraph*{Building OWPuzzs from QCCC commitments.}
To build OWPuzzs from (general poly-round) QCCC commitments, we again take advantage of distributional OWPuzzs. In particular, we observe that if distributional OWPuzzs do not exist, 
then it is possible to do conditional sampling for any efficiently sampleable distribution. 
This gives an approach to break hiding for QCCC commitment schemes. 

In particular, given the transcript $\tau$ of the commit phase for a commitment to a bit $b$, 
a corrupted receiver can sample the committer's first message in the opening phase conditioned on $\tau$.\footnote{Remember that we are now considering general commitments where the opening phase is not necessarily non-interactive.} The receiver can then repeat this process to sample the committer's second and third messages, up until the entire opening phase has been simulated. Since the receiver can run the entire opening phase on its own, it can learn the output $b$. By binding, this should be the same as the committed bit, and therefore this attack breaks hiding.

\begin{theorem}[\Cref{thm:commtoowpuzz}, restated]
If there exist QCCC commitments, then there exist OWPuzzs.
\end{theorem}

\paragraph*{Building OWPuzzs from 2QKD.}
\if0
Formally, a QKD protocol is a protocol where Alice and Bob agree on a shared key over a public authenticated classical channel and a public arbitrary quantum channel. In particular, a good QKD protocol should satisfy three properties. 
\begin{enumerate}
    \item (Correctness): If the protocol is executed honestly, then Alice and Bob output the same key.
    \item (Validity): An adversarial party Eve, who controls the quantum channel, cannot make Alice and Bob output different keys without them noticing and outputting $\bot$.
    \item (Security): If either Alice or Bob does not output $\bot$, then Eve can not learn their final key.
\end{enumerate}
A two-round QKD protocol, or 2QKD is simply a QKD protocol where Alice and Bob each send at most one message. It is known that 2QKD protocols can be built from post-quantum OWFs~\cite{C:KMNY24,C:MalWal24}.
\fi
A concurrent work~\cite{STOC:KhuTom25} defines a quantum version of OWPuzzs, state puzzles, and shows that state puzzles are equivalent to OWPuzzs. 
A state puzzle is essentially a OWPuzz where the key is a quantum state. Since the key is quantum, there is a canonical verification algorithm which just projects onto the target state.\footnote{Note that we can always assume that the key is a pure state, since adding the extra qubits from its purification into the key only improves security.}

We observe that the existence of a 2QKD protocol immediately gives a state puzzle. In particular, the puzzle will consist of Alice's first classical message to Bob, 
and the key will be Alice's internal state as well as her first quantum message.\footnote{Remember that in QKD we have authenticated classical channel and non-authenticated quantum channel. We therefore assume that Alice's first message consists of a quantum part and a classical part.} In particular, if this is not a state puzzle, then we can attack the 2QKD protocol by replacing Alice's quantum message with the result from the state puzzle breaker. This will allow an adversary to impersonate Alice and obtain a corresponding private state. Since Alice does not send any more messages beyond this, Bob will not be able to detect this impersonation and so the adversary will be able to predict Bob's output.

\begin{theorem}[\Cref{thm:2qkd}, restated]
If there exist 2QKD, then there exist OWPuzzs.
\end{theorem}

\paragraph{Paper organization} We begin with preliminaries and definitions in~\Cref{prelims-sec}. In~\Cref{sec:presampling}, we outline and modify the quantum version of the presampling lemma from~\cite{TCC:GLLZ21}. In~\Cref{sec:sepsec}, we introduce our oracles $SG,Mix,PSPACE$ and prove the key lemma we will use to simulate this oracle by a stateful efficient procedure. This argument contains the novel key ideas needed to show all of our results relative to $SG,Mix,PSPACE$. We include the proofs of security for quantum lightning, key exchange, and commitments relative to $SG,Mix,PSPACE$ in~\Cref{sec:lightning,app:ke,sec:commitoracle}. \Cref{sec:compequivunit} contains the proof that $\mathbf{BQP}^{SG,Mix,PSPACE}=\mathbf{QCMA}^{SG,Mix,PSPACE}$. We conclude with the constructions of one-way puzzles from CountCrypt primitives in~\Cref{sec:owpz}.

%% file: prelims/prelims.tex
\section{Preliminaries}
\label{prelims-sec}
In this section, we provide a refresher on basics of quantum information, Haar measure and its concentration, as well as definitions of primitives in quantum cryptography.  

\subsection{Notation}
Throughout the paper, $[n]$ denotes the set of integers $\{1,2,\dots, n\}$. If $X$ is a set, we use $x \gets X$ to denote that $x$ is sampled uniformly at random from $X$. 
Similarly, if $A$ is an algorithm and $x$ is an input, we use $y \gets A(x)$ to say that $y$ is the output of $A$ on input $x$. 
For a distribution $D$, $x\gets D$ means that $x$ is sampled according to $D$. A function $f$ is \emph{negligible} if for every constant $c>0$, $f(n)\leq \frac{1}{n^c}$ for all sufficiently large $n\in\mathbb{N}$. We will use $\negl(n)$ to denote an arbitrary negligible function, and $\poly(n)$ to denote an arbitrary polynomially bounded function. For two interactive algorithms $A$ and $B$ that interact over a classical or quantum channel,
$A\rightleftarrows B$ means the execution of the interactive protocol, or the distribution of the transcript, i.e. the sequence of messages exchanged between $A$ and $B$. $A\leftrightarrows B \to y$ means that the execution of the interactive protocol produces output $y$.
We use the abbreviation QPT for a uniform quantum polynomial time algorithm. We will also use non-uniform QPT to refer to quantum polynomial time algorithms with polynomial \textit{classical} advice. We use the notation $A^{(\cdot)}$ to refer to a (classical or quantum) algorithm
that makes queries to an oracle. For an algorithm $\A$, we use $|\A|$ to refer to the length of some canonical description of $\A$. We use $\TD(\rho,\sigma)=\frac{1}{2}\|\rho-\sigma\|_1$ to denote the trace distance between density matrices $\rho$ and $\sigma$. 
For two distributions $D_1,D_2$, we use $\Delta(D_1,D_2) = \frac{1}{2}\sum_{x}\abs{\Pr[x\gets D_1] - \Pr[x\gets D_2]}$ to denote the total variation distance (aka statistical distance) between $D_1$ and $D_2$. 
%We write $D_1 \approx_\epsilon D_2$ to say that $D_1$ and $D_2$ are $\epsilon$ computationally indistinguishable. 
%That is, for all QPT $\A$, 
%$\abs{\Pr_{x\gets D_1}[1\gets\A(x)] - \Pr_{x\gets D_2}[1\gets \A(x)]} \leq \epsilon$. We use $D_1 \approx D_2$ to mean $D_1 \approx_{\negl(n)} D_2$. 
For algorithms which take in a unary security parameter $1^n$, we will sometimes omit $1^n$ as input when the value is implicit. We will use $\omega_T$ to denote a (canonical) $T$th root of unity. That is, $\omega_T^T = \omega_T^0 = 1$, and for all $0 < i < T$, $\omega_T^i\neq 1$. For a projector $\Pi$, we will say "apply the measurement $\{\Pi,I-\Pi\}$" to mean apply the map $\rho\mapsto \Pi\rho \Pi + (I-\Pi)\rho(I-\Pi)$.

\subsection{Quantum Information}
%We use $\mathsf{TD}(\rho,\sigma)$ to denote the trace distance between density matrices $\rho$ and $\sigma$. 

A quantum channel $\mathcal{A}$ is any completely-positive trace-preserving (CPTP) map. Intuitively, quantum channels represent all ``physical'' processes. That is, the set of CPTP maps is the quantum analogue of the set of all functions.

For a quantum channel $\mathcal{A}$, we let $\|A\|_{\diamond}$ denote its \emph{diamond norm}. The diamond norm and trace distance satisfy the following relation:
\begin{theorem}[\cite{NC}]\label{fact}
     Let $\mathcal{A}$ and $\mathcal{B}$ be quantum channels and $\rho$ be a density matrix. Then,
     $\mathsf{TD}(\mathcal{A}(\rho),\mathcal{B}(\rho)) \leq ||\mathcal{A}-\mathcal{B}||_{\diamond}$.
 \end{theorem}
%% \mor{Do we use this often? If not, we can just mention it when this is used. In that case, we can remove this entire subsection.}
Hence if two quantum channels are close in diamond distance, triangle inequality implies that they are indistinguishable via bounded query access.

We note the following relationship between trace distance and fidelity of pure states.
\begin{lemma}\label{lem:tlem}
    Let $\ket{\phi},\ket{\psi}$ be two pure states. Then,
    \begin{align}
    \TD(\ketbra{\phi},\ketbra{\psi}) = \sqrt{1 - \abs{\braket{\phi|\psi}}^2}.
    \end{align}
\end{lemma}

In this work, we will often consider quantum channels corresponding to performing some projective measurement, which we call
measurement channels.
\begin{definition}[Measurement Channels]
    Let $\mathcal{M} = \{\Pi, I - \Pi\}$ be a binary measurement corresponding to a projector $\Pi$. 
    The measurement channel corresponding to $\mathcal{M}$ is the following channel:
    \begin{align}
    \rho \mapsto \Pi \rho \Pi + (I - \Pi) \rho (I - \Pi).
    \end{align}
\end{definition}

We make use of the following fact about measurement channels.
\begin{lemma}[\cite{diamond}]\label{lem:diamond}
    Let $\mathcal{A} = \{\Pi_A, I - \Pi_A\}$, $\mathcal{B} = \{\Pi_B, I - \Pi_B\}$ be two measurement channels. Then
    \begin{align}
    ||\mathcal{A}-\mathcal{B}||_{\diamond} = 2 \max_{\rho} \abs{{\rm Tr}(\Pi_A \rho) - {\rm Tr}(\Pi_B \rho)}.
    \end{align}
\end{lemma}

 \subsection{The Compressed Oracle Technique/The Recording Method}
\label{sec:comporacle} 
Here we explain the compressed oracle technique and the recording method~\cite{C:Zhandry19}. 
First, for the random oracle $H:\{0,1\}^m \rightarrow \{0,1\}^n$, we denote the usual quantum random oracle, or the standard oracle by $\mathsf{Sto}$. In this model, $H$ is represented as a truth table: a vector of size $2^m$ where each component is an $n$ bit string.  
The compressed  oracle for $H:\{0,1\}^m \rightarrow \{0,1\}^n$, $\mathsf{CStO}$, acts on the query and database registers $(\mathbf{Q}, \mathbf{D})$, where the database register $\mathbf{D}$ will be a collection $D$ of input output pairs  $(x,y)$, and $(x,y)\in D$ denotes that the random oracle has been specified to have value $y$ on input $x$. We will write $D(x)=y$ in this case. If, for an input $x$, there is no $(x,y)\in D$, we will write $D(x)=\bot$, indicating that the function has not been specified on $x$.
\def\CStO{\mathsf{CStO}}
\def\Decomp{\mathsf{Decomp}}
\begin{align*}
    \CStO \coloneqq \Decomp \circ \CStO' \circ \Decomp \circ \mathsf{Increase},
\end{align*}
We will refer to the unitary $\CStO$ as the compressed standard oracle. 
where $\mathsf{Increase}$ is the procedure which initializes a new register $\ket{(\bot, 0^n)}$ and appends it to the end: $\mathsf{Increase} \ket{x, y}_{\mathbf{Q}} \otimes \ket{D}_{\mathbf{D}} = \ket{x, y} \otimes \ket{D} \ket{(\bot, 0^n)}$.\mor{Isn't there any typo here? The number of regisers are increased.} 
 $\Decomp$ and $\CStO'$ are defined as follows:
\begin{itemize}
    \item $\Decomp$ is defined by $\ket{x, u}_\mathbf{Q} \otimes \ket{D}_\mathbf{D} \mapsto \ket{x, u}_{\mathbf{Q}} \otimes \Decomp_x\ket{D}_{\mathbf{D}}$, where $\Decomp_x$ is defined by its action on basis states as follows. Let $D$ be such that $D(x) = \bot$:
    \begin{align}
        \ket{D} &\mapsto \frac{1}{\sqrt{|\mathcal{Y}|}} \sum_{y \in \mathcal{Y}} \ket{D \cup \{(x,y)\}}\\
        \frac{1}{\sqrt{|\mathcal{Y}|}} \sum_{y \in \mathcal{Y}} \ket{D \cup \{(x,y)\}} &\mapsto \ket{D}\\
        \frac{1}{\sqrt{|\mathcal{Y}|}} \sum_{y \in \mathcal{Y}} (-1)^{y \cdot u} \ket{D \cup \{(x,y)\}} &\mapsto \frac{1}{\sqrt{|\mathcal{Y}|}} \sum_{y \in \mathcal{Y}} (-1)^{y \cdot u} \ket{D \cup \{(x,y)\}} \quad \text{ for } u \neq 0
    \end{align}
    In words, $\Decomp_x$ swaps the states $\ket{D}$ and $\frac{1}{\sqrt{|\mathcal{Y}|}} \sum_{y \in \mathcal{Y}} \ket{D \cup \{(x,y)\}}$, and acts as the identity on the space orthogonal to these two states.
    \mor{Is the third state orthogonal to the first and second states?}
    \item $\CStO'$ maps $\ket{x, u}_\mathbf{Q} \otimes \ket{D}_\mathbf{D} \mapsto \ket{x, u \oplus D(x)}_{\mathbf{Q}} \otimes \ket{D}_{\mathbf{D}}$.
\end{itemize}
\begin{lemma}[Compressed oracle method~\cite{C:Zhandry19}]\label{lem:comporacle}
$\mathsf{Ctso}$ and $\mathsf{Sto}$ are perfectly indistinguishable. That is, for any adversary $\mathcal{A}$, we have $\Pr[A^{\mathsf{Ctso}}(\cdot)=1]=\Pr[A^{\mathsf{Sto}}(\cdot)=1]$.
\end{lemma}
Next, we will state the following theorem, which gives a general formulation of the overall state of $\mathcal{A}$ and the compressed standard oracle %\mor{the word "compressed standard oracle" was not defined} 
after $\mathcal{A}$ makes $t$ queries, even conditioned on arbitrary measurement results.
\begin{theorem}[\cite{C:Zhandry19,TCC:GLLZ21}]\label{thm:presampledatabase}
   % Define the map $CRO$ over registers $\mathbf{A},\mathbf{B},\mathbf{D}$\mor{registers $\mathbf{A},\mathbf{B},\mathbf{D}$} where $D\in ([N])^{M}$ by
  %  \begin{align}
    %CRO\ket{x,y}_{\mathbf{A},\mathbf{B}}\ket{D}_{\mathbf{D}} \mapsto \ket{x,y}_{\mathbf{A},\mathbf{B}}\ket{D\xor \delta_{x}\cdot u}_{\mathbf{D}},
    %\end{align}
  %  \mor{$\delta_x\cdot y$? What is $u$? Also I suggest $D\xor (\delta_x\cdot y)$ to avoid misunderstanding it as $(D\xor \delta_x)\cdot y$}
  %  where $\delta_{x}\cdot u$ is the function which is $u$ on $x$ and $0$ everywhere else.
 %   \mor{where $x\in[M]$ and $y\in[N]$}
   Let $M,N\in \N$, and let $f:[M]\rightarrow [N]$. For any algorithm $\A$ making $t$ queries to a compressed standard oracle%\mor{the word "compressed standard oracle" was not defined}
   , assuming the overall state of $\A$ and the compressed standard oracle%\mor{the word "compressed standard oracle" was not defined} 
   is $\ket{A^{\mathsf{Ctso}}}=\sum_{z,D}\alpha_{z,D}\ket{z}_{\textbf{Q}}\ket{D}_{\textbf{D}}$ where $\textbf{Q}$  is the query register and $\textbf{D}$ is the database register of the oracle, 
    \begin{align}
    \E_{f:[M]\to[N]}\left[|A^f\rangle\langle A^f|\right] = {\rm Tr}_{\mathbf{D}}\left[\ketbra{\A^{\mathsf{Ctso}}}\right].
    \end{align}
    where $\ket{\A^f}$ is the state of $\A$ after making $t$ queries to $f$.
    Furthermore, define 
    \begin{align}
    \Pi_{\leq t}\coloneqq  \text{projection onto the set of databases $D$}: |D|\leq t. 
    \end{align}
    Then the $\textbf{D}$ register will only have support on all databases $D$ which lie within the subspace specified by $\Pi_{\leq t}$. 
    That is, 
    \begin{align}
    {\rm Tr}\left((I\otimes \Pi_{\leq t})\ketbra{\A^{\mathsf{Ctso}}}\right)=1.
    \end{align}
    Moreover, this is true even if the state is conditioned on arbitrary outcomes (with non zero probability) of $\A$'s intermediate measurements. 
\end{theorem}
\mor{Did we explain about the recording method? If the above is the recording method, it should be said so.}

\subsection{One Way To Hiding}
We use the following lemma which combines~\cite[Theorem 3]{cryptoeprint:2018/904} and ~\cite[Lemma 8]{cryptoeprint:2018/904}, where the latter is rooted in Vazirani's \emph{Swapping Lemma}~\cite{Vaz98}. 

\begin{lemma}[One-Way-to-Hiding Lemma,~\cite{cryptoeprint:2018/904}]\label{lem:O2H}
Let $\algo X,\algo Y$ be arbitrary sets and let $\algo S \subseteq \algo X$ be a (possibly random) subset. Let $G,H: \algo X \rightarrow \algo Y$ be arbitrary (possibly random) functions such that $H(x)=G(x)$, for all $x \notin \algo S$. Let $z$ be a classical bit string or a (possibly mixed) quantum state. (Note that $G,H,S,z$ may have arbitrary joint distribution.) Let $\algo A$ be an
oracle-aided quantum algorithm that makes at most $q$ quantum queries. Let $\algo B$ be an algorithm that on input $z$ chooses a random query index $i \leftarrow [q]$, runs $\algo A^H(z)$, measures $\algo A$'s $i$-th query and outputs the measurement
outcome. Then, we have    
\begin{align}
\left|\Pr[\algo A^G(z)=1] - \Pr[\algo A^H(z)=1] \right| \leq 2 q \sqrt{\Pr[\algo B^H(z) \in \algo S]}.
\end{align}
Moreover, for any fixed choice of $G,H,S$ and $z$ (when $z$ is a classical string or a pure state), we get
\begin{align}
\big\| \ket{\psi^H_q} - \ket{\psi^G_q} \big\| \leq 2q\sqrt{\frac{1}{q}\sum_{i = 0}^{q-1} \big\|\Pi_{\algo S}\ket{\psi^H_i}\big\|^2},
\end{align}
where $\ket{\psi^H_i}$ denotes the intermediate state of $\algo A$ just before the $(i+1)$-st query, where the initial state at $i=0$ corresponds to $z$, and $\Pi_{\algo S}$ is a projector onto $\algo S$.
\end{lemma}

In particular, we will actually make use of the slightly more generalized version.
 \begin{lemma}
 \label{cor:o2h}
    Let $U,V$ be any unitaries acting on a Hilbert space $\mathcal{H}$, with the guarantee that, $U \ket{\psi}= V\ket{\psi}$, for all $\ket{\psi}$ $\notin S$, which is a (possibly random) subspace of $\mathcal{H}$.  Let $z$ be a classical bit string or a (possibly mixed) quantum state. Let $\algo A$ be an
oracle-aided quantum algorithm that makes at most $q$ quantum queries. Then, we have

\begin{align}
\big\| \ket{\psi^V_q} - \ket{\psi^U_q} \big\| \leq 2\sum_{i = 0}^{q-1}\sqrt{ \big\|\Pi_{\algo S}\ket{\psi^U_i}\big\|^2},
\end{align}
where $\ket{\psi^U_i}$ denotes the intermediate state of $\algo A$ just before the $(i+1)$-st query, where the initial state at $i=0$ corresponds to $z$, and $\Pi_{\algo S}$ is a projector onto $\algo S$.
 \end{lemma}

\begin{proof}
    This follows directly from gentle measurement. Let $\Pi_S$ be the projector onto the subspace $S$. Define $\ket{\psi_i^{Int,U}}$ to be the state generated as follows:
    \begin{itemize}
        \item Run $\A^{U}$ up until just before the $(i+1)$-st query, but before every query $j \in [i]$, apply the measurement $\{\Pi_S,I-\Pi_S\}$ onto $\ket{\psi_j}$. 
        \item If the result is the first outcome, output $\bot$. 
        \item Otherwise, continue.
    \end{itemize} 
    Note that since $U$ and $V$ agree on ${\sf Im}(I - S)$, $\ket{\psi_i^{Int,U}}=\ket{\psi_i^{Int,V}}$.

    We know that $\ket{\psi_1^U}=\ket{\psi_1^{Int,U}}$. Gentle measurement immediately gives that $\norm{\ket{\psi_2^{U}} - \ket{\psi_2^{Int,U}}} \leq 2\sqrt{\norm{\Pi_S\ket{\psi_1^U}}^2}$. 
    Now for our induction hypothesis, assume that,
    $\norm{\ket{\psi_{q-1}^{U}} - \ket{\psi_{q-1}^{Int,U}}} \leq 2\sum_{i=0}^{q-2} \sqrt{\norm{\Pi_S \ket{\psi_i^U}}^2}$. \\
    By gentle measurement and triangle inequality
    \begin{align*}
        \norm{\ket{\psi_q^U} - \ket{\psi_q^{Int,U}}}&=\norm{\ket{\psi_{q-1}^U} - \ket{\psi_{q-1}^{Int,U}}} + 2\sqrt{\norm{\Pi_S \ket{\psi_{q-1}^U}}^2}\\
        &\leq 2\sum_{i=0}^{q-1}\sqrt{ \norm{\Pi_S \ket{\psi_i^U}}^2}
    \end{align*}
    %Induction and the triangle inequality then gives us that
    %$$\norm{\ket{\psi_q^U} - \ket{\psi_q^{Int,U}}}\leq O\left(\sqrt{\sum_{i=0}^{q-1} \norm{\Pi_S \ket{\psi_i^U}}^2}\right)$$
    A symmetric argument gives
    $$\norm{\ket{\psi_q^V} - \ket{\psi_q^{Int,V}}}\leq 2\sum_{i=0}^{q-1}\sqrt{ \norm{\Pi_S \ket{\psi_i^U}}^2}$$
    and the result follows from triangle inequality.
\end{proof}

\subsection{Haar Measure And Its Concentration}
We use $\mathbb{U}(N)$ to denote the group of $N \times N$ unitary matrices, and $\HaarUn{N}$ to denote the Haar measure on $\mathbb{U}(N)$. We use $\C^N$ to denote the set of quantum pure states of dimension $N$, and $\HaarSt{N}$ to denote the Haar measure on $\C^N$.
Given a metric space $(\mathcal{M},d)$ where $d$ denotes the metric on the set $\mathcal{M}$, a function $f: \mathcal{M}\rightarrow \mathbb{R}$ 
is $L$-lipschitz if for all $x,y \in \mathcal{M},\left|f(x)-f(y)\right|\leq L\cdot d(x,y)$. 
For $d$, we will in particular be concerned about the Frobenius norm. For a complex matrix $M$, its Frobenius norm is $\sqrt{\mbox{Tr}(M^\dagger M)}$.

The following inequality involving Lipschitz continuous functions captures the strong concentration of Haar measure.
\begin{theorem}[\cite{Mec19}]
\label{thm:conc}
Let $k\geq 1$ be an integer. Given $N_1,N_2, \dots, N_k \in \mathbb{N}$, let $X = \mathbb{U}(N_1)\bigoplus \dots \bigoplus  \mathbb{U}(N_k) $ be the space of block diagonal unitary matrices with blocks of size $N_1,N_2,\dots ,N_k$. 
Let $\mu=\mu_{N_1}\times \dots\times \mu_{N_k}$
be the product of Haar measures on $X$, 
where, for each $i\in[k]$, $\mu_{N_i}$ is the Haar measure on $\mathbb{U}(N_i)$.
Suppose that $f: X\rightarrow \mathbb{R}$ is $L$-Lipshitz with respect to the Frobenius norm. Then for every $t>0$,
 \begin{align}
 \Pr_{U \leftarrow \mu}\left[\left|f(U) - \mathbb{E}_{V \leftarrow \nu}[f(V)]\right| \geq t\right]\leq 2\exp\left({-\frac{(N-2)t^2}{24L^2}}\right) ,
 \end{align}
 where $N=\min(N_1,\dots ,N_k)$.
 \end{theorem}

We will use the following lemma.
 \begin{lemma}[\cite{TQC:Kre21}]
 \label{lem:lipshitz}
     Let $A^{(\cdot)}$ be a quantum algorithm that makes $T$ queries to a unitary oracle operating on a $D$-dimensional state space, and let $\ket{\psi}$ be any input to $A^{(\cdot)}$. 
     For $U \in \mathbb{U}(D)$, define $f(U)=\Pr[A^{U}(\ket{\psi})=1]$. Then, $f$ is $2T$-Lipshitz with respect to the Frobenius norm.
 \end{lemma}
 
\subsection{State Tomography}
%\paragraph{State Tomography.}
\begin{theorem}[Quantum State Tomography, Proposition 2.2 from~\cite{HKOT23}]\label{thm:statetom}
    There is a pure state tomography procedure with the following behavior. Let $d\in\mathbb{N}$. 
    For any $n>d,\epsilon>0$, given $O(n/\epsilon)$ copies of a pure state $\ket{\phi}\in \C^d$, it outputs a classical description of an estimate pure state $\ket{\phi'}$ such that with probability $\geq 1 - e^{-5n}$, 
    \begin{align}
    \abs{\braket{\phi|\phi'}}^2 \geq 1 - \epsilon.
    \end{align}
\end{theorem}
\begin{remark}
    The statement from~\cite{HKOT23} specifies an error depending on the dimension of the state. However, this stronger statement is immediate upon observing the proof, and indeed the predecessor of this result from~\cite{GKKT20} makes the dependence on $n$ explicit.
\end{remark}

\subsection{Quantum Complexity}

We will consider promise problem versions of $\mathbf{BQP}$ and $\mathbf{QCMA}$.

A promise problem $\Pi$ is a pair of disjoint sets $\Pi_{yes},\Pi_{no}\subseteq \{0,1\}^*$ with $\Pi_{yes}\cap \Pi_{no} = \emptyset$.

\begin{definition}
    Let $\mathcal{O}$ be an oracle. We say that a promise problem $\Pi = (\Pi_{yes},\Pi_{no})$ is in $\mathbf{BQP}$ if there exists constants $c_1>c_2$ and a randomized polynomial-time quantum oracle algorithm $\mathcal{A}^{\mathcal{O}}(x)$ such that
    \begin{enumerate}
        \item If $x \in \Pi_{yes}$, then $\Pr[\A^{\mathcal{O}}(x) = 1] \geq c_1$
        \item If $x \in \Pi_{no}$, then $\Pr[\A^{\mathcal{O}}(x) = 1] \leq c_2$
    \end{enumerate}
\end{definition}

\begin{definition}
    Let $\mathcal{O}$ be an oracle. We say that a promise problem $\Pi = (\Pi_{yes},\Pi_{no})$ is in $\mathbf{QCMA}$ if there exists constants $c_1>c_2$, a randomized polynomial-time quantum oracle algorithm $\Ver^{\mathcal{O}}(x)$, and a polynomial $p(\cdot)$ such that
    \begin{enumerate}
        \item If $x \in \Pi_{yes}$, then there exists a witness $w\in \{0,1\}^{p(|x|)}$ such that $$\Pr[\Ver^{\mathcal{O}}(x,c) = 1] \geq c_1$$
        \item If $x \in \Pi_{no}$, then for all witnesses $w\in \{0,1\}^{p(|x|)}$, $$\Pr[\Ver^{\mathcal{O}}(x,w) = 1] \leq c_2$$
    \end{enumerate}
\end{definition}

 \subsection{Quantum Cryptography Definitions}
\label{sec:primitive}
In this subsection, we provide definitions of quantum cryptographic primitives we consider in this paper.

\begin{definition}[QCCC Key Exchange (KE)]
    A QCCC key exchange (KE) protocol is an interactive two-party protocol between two parties $\mathsf{A}$ and $\mathsf{B}$. 
    Both $\mathsf{A}$ and $\mathsf{B}$ are QPT, but all communications are classical.
    At the end of the protocol, $\mathsf{A}\leftrightarrows \mathsf{B}$ produces a classical transcript $\tau$, 
    $\mathsf{A}$'s output bit $a$, and $\mathsf{B}$'s output bit $b$. A KE protocol must satisfy the following properties.
    \begin{itemize}
        \item Correctness: At the conclusion of the protocol, both parties agree on the output. Formally, 
        \begin{align}
        \Pr[a = b : \mathsf{A}(1^\secpa) \leftrightarrows \mathsf{B}(1^\secpa) \to (\tau,a,b)] \geq 1 - \negl(\secpa).
        \end{align}
        \item Security: An adversary with the ability to see all communication should not be able to learn the output. 
        That is, for all non-uniform QPT adversaries $\mathcal{E}$,
        \begin{align}
        \Pr[a\gets \mathcal{E}(\tau) : \mathsf{A}(1^\secpa)\leftrightarrows \mathsf{B}(1^\secpa) \to (\tau,a,b)] \leq \negl(\secpa).
        \end{align}
    \end{itemize}

\end{definition}

\begin{definition}[QCCC Multiparty Non-Interactive Key Exchange (NIKE)]
    An $N$-party QCCC multiparty non-interactive key exchange (NIKE) is a pair $(\mathsf{Publish},\mathsf{KeyGen})$ of QPT algorithms with the following syntax:
   \begin{itemize}
       \item 
       $\mathsf{Publish}(1^\secpa,i)\to (x_i,\rho_i)$: It is a QPT algorithm that, on input the security parameter $1^\secpa$ and an integer $i\in[N]$, 
       outputs a classical bit string $x_i$ and a quantum state $\rho_i$.
       \item 
    $\mathsf{KeyGen}(i,x_1,...,x_N,\rho_i)\to k_i$: It is a QPT algorithm that, on input $i$, $(x_1,...,x_N)$ and $\rho_i$, outputs a classical key $k_i$.   
   \end{itemize} 
   We require the following two properties:
   \begin{itemize}
   \item
   Correctness:
   \begin{align}
\Pr[\forall i,j\in[N], k_i=k_j:(x_i,\rho_i)\gets\mathsf{Publish}(1^\secpa,i),k_i\gets\mathsf{KeyGen}(i,x_1,...,x_N,\rho_i)]\ge1-\negl(\secpa).       
   \end{align}
    \item 
    Security:
    For any non-uniform QPT adversary $\cE$,
    \begin{align}
\Pr\left[
\exists i\in[N]~ s.t.~ e=k_i:
\begin{array}{l}
(x_i,\rho_i)\gets\mathsf{Publish}(1^\secpa,i)\\
k_i\gets\mathsf{KeyGen}(i,x_1,...,x_N,\rho_i)\\
e\gets\cE(x_1,...,x_N)
\end{array}
\right]\le\negl(\secpa).       
    \end{align}
    \end{itemize}
\end{definition}
    
\begin{definition}[Quantum Lightning~\cite{JC:Zhandry21}]
    A quantum lightning protocol is a tuple 
    $(\Mint, \Veri)$ 
    of uniform QPT algorithms 
    with the following syntax.
    \begin{itemize}
        \item $\Mint(1^n) \to (\sigma,\ket{\$})$: takes in the security parameter $1^n$ and outputs a serial number $\sigma$ and a bolt state $\ket{\$}$.
        \item $\Veri(1^n,\sigma, \ket{\$}) \to \top/\bot$: takes in the security parameter $1^n$ as well as a serial number $\sigma$ and a money state $\ket{\$}$, 
        and outputs $\top/\bot$.
%        \mor{security parameter is included in the serial number?}\eli{this is insufficient, maybe you can clone on low security parameters, so the security game needs some way of knowing what security parameter the adversary is trying to break}
    \end{itemize}
    We require the following properties.
    \begin{itemize}
        \item Correctness: Honestly generated bolt states verify under their serial number. That is, 
        \begin{align}
        \Pr[\top\gets\Veri(1^n,\Mint(1^n))] \geq 1-\negl(n).
        \end{align}
        \item Security: No attacker can generate two bolt states for the same serial number. That is, for all non-uniform QPT attackers $\A$,
        \begin{align}
        \BigPr{\top\gets \Veri(1^n,\sigma, \rho_{\mathbf{A}}) \\\top\gets \Veri(1^n,\sigma, \rho_{\mathbf{B}})}{(\sigma,\rho_{\mathbf{AB}})\gets \A(1^n) } \leq \negl(n).
        \end{align}
        Here, $\rho_{\mathbf{AB}}$ is a state over two registers $\mathbf{A}$ an $\mathbf{B}$, 
        and $\rho_{\mathbf{A}}$ ($\rho_{\mathbf{B}}$) is the state on the register $\mathbf{A}$ ($\mathbf{B}$).
    \end{itemize}
\end{definition}
\begin{definition}[Distributional One-Way Puzzles~\cite{C:ChuGolGra24}]
    A $\beta$-distributional one-way puzzle (OWPuzz) is a uniform QPT algorithm $\Samp(1^\secpa) \to (k,s)$ which takes in a security parameter $1^\secpa$ and produces two classical bit strings, a key $k$ and a puzzle $s$, such that given a puzzle $s$, it is computationally infeasible to sample from the conditional distribution over keys $k$. More formally, we require that for all non-uniform QPT algorithms $\mathcal{A}$, for all sufficiently large $\secpa\in\mathbb{N}$,
    \begin{align}
    \Delta((k,s),(\mathcal{A}(1^\secpa,s),s)) \geq \beta(\secpa),
    \end{align}
    where $(k,s)\gets\Samp(1^\secpa)$.
\end{definition}

\begin{remark}
    Distributional OWPuzzs are equivalent to OWPuzzs~\cite{C:ChuGolGra24}, defined in~\cite{STOC:KhuTom24}. 
    A OWPuzz is a QPT sampler paired with an inefficient verifier such that given a puzzle $s$, it is hard to find a key $k$ which verifies. Since we will not work directly with OWPuzzs, we omit the formal definition.
\end{remark}

\begin{definition}[QCCC Commitments~\cite{STOC:KhuTom24}]\label{def:com}
    A QCCC commitment scheme is a two-party protocol between a QPT committer $\Comm$ and a QPT receiver $\Rece$ 
    over a classical channel
    consisting of a commit stage and an opening stage operating on a private input $m$ described as follows.
    \begin{enumerate}
       \item Commit stage: Both parties receive the security parameter $1^\secpa$. 
        The committer $\Comm$ receives a private input $m$. 
        $\Comm$ interacts with the receiver $\Rece$ using only classical messages, and together they produce a transcript $z$. At the end of the stage, both parties hold a private quantum state $\rho_{\Comm}$ and $\rho_{\Rece}$, respectively.
        \item Opening stage: Both parties receive the transcript $z$ as well as their private quantum states $\rho_{\Comm}$ and $\rho_{\Rece}$, respectively. The committer $\Comm$ interacts with the receiver $\Rece$ using only classical messages. At the end of the stage, $\Rece$ either outputs a message or the reject symbol $\bot$.
    \end{enumerate}
    We require the following two properties.
    \begin{enumerate}
        \item Correctness: For all messages $m$, when $\Comm$ and $\Rece$ interact honestly, the probability that $\Rece$ outputs $m$ at the end of the opening stage is at least $1 - \negl(\secpa)$.
        \item (Computational) hiding: For all $m\neq m'$ and for all non-uniform QPT adversarial receivers $\Rece'$, 
        the view to $\Rece'$ interacting with $\Comm$ with input $m$ is computationally indistinguishable from the view of
        $\Rece'$ interacting with $\Comm$ with input $m'$. That is,
        for any non-uniform QPT adversary $\Rece'$ that interacts with the honest $\Comm$ in the commit stage and outputs a bit
        after the end of the commit stage,
        \begin{align}
|\Pr[1\gets\Rece'_{1^\secpa,m}]-\Pr[1\gets\Rece'_{1^\secpa,m'}]|\le\negl(\secpa).            
        \end{align}
        Here, $\Rece'_{1^\secpa,m}$ is a non-uniform QPT adversary that interacts with the honest $\Comm$ whose input is $(1^\secpa,m)$.
%        \begin{align}
%        \Comm(1^\secpa,m)\rightleftarrows \Rece' \approx \Comm(1^\secpa,m') \rightleftarrows \Rece'.
%        \end{align}
%        Here 
%        $\Comm(1^\secpa,m)\rightleftarrows \Rece' $ is the distribution of the transcript of the commit stage between $\Comm(1^\secpa,m)$ and $\Rece'$.
        \item (Computational weak honest) binding: For all $m$ and for any non-uniform QPT adversarial committer 
        $\Comm'$, the probability that $\Comm'$ wins the following game is $\leq \negl(\secpa)$.
        \begin{enumerate}
            \item In the commit stage, an honest receiver $\Rece$ interacts with the honest committer $\Comm$ 
            to produce a transcript $z$ and the output states $\rho_{\Comm}$ and $\rho_{\Rece}$.
            \item In the opening stage, the honest receiver $\Rece$ is given $\rho_{\Rece}$ and $z$, while $\Comm'$ is given $z$ (but not $\rho_{\Comm}$). They then proceed to run the opening stage with the committer replaced by $\Comm'$, and $\Rece$ produces a final output $m'$. $\Comm'$ wins if $m' \neq m$ and $m' \neq \bot$. 
        \end{enumerate}
    \end{enumerate}
\end{definition}

\begin{remark}
    Here we use a very weak security definition for QCCC commitments. This only serves to make our construction of OWPuzz from QCCC commitments stronger. For our separation, we use a stronger security notion detailed in~\Cref{sec:commitoracle}.
\end{remark}

%% file: prelims/tricks.tex
\subsection{State Verification Unitaries}

\begin{definition}
    For a state $\ket{\phi}$, we will define two unitaries representing ``verification of the state''
    as
    \begin{equation}
        \begin{split}
            R_{\ket{\phi}} \ket{\phi} \mapsto (-1)\ket{\phi}\\
            P_{\ket{\phi}} \ket{\phi}\ket{b} \mapsto \ket{\phi}\ket{\bar{b}}
        \end{split}
    \end{equation}
    and on all orthogonal inputs both maps will act as the identity. It is well known that $R_{\ket{\phi}}$ can be used to compute $P_{\ket{\phi}}$ and vice-versa. 
\end{definition}

\begin{lemma}[\cite{C:JiLiuSon18}, Theorem 4]\label{lem:simref}
    For any $\epsilon>0$, for any $t \geq \frac{2}{\epsilon^2}$ and an efficient algorithm $\mathcal{B}$ such that for all $\ket{\phi}$,
    \begin{align}
    \|R_{\ket{\phi}}(\cdot)R_{\ket{\phi}}^\dagger - \mathcal{B}(\ket{\phi}^{\otimes t},\cdot)\|_{\diamond} \leq \epsilon.
    \end{align} \takashi{I believe $R_{\ket{\phi}}\cdot  R_{\ket{\phi}}^\dagger$ means the channel that maps $\rho$ to $R_{\ket{\phi}}\rho  R_{\ket{\phi}}^\dagger$, but this may be misunderstood as a product as matrices (which trivially give the  identity). To avoid the misunderstanding, I think writing like $R_{\ket{\phi}}(\cdot)  R_{\ket{\phi}}^\dagger$ would be more common, and we should explain what this notation means for clarity.}
   In particular, let $\Pi_{sym}^{t+1}$ be the projection onto the symmetric subspace $Span(\ket{\phi}^{\otimes t+1})$. 
   Then $\mathcal{B}(|\phi\rangle^{\otimes t},|\psi\rangle)$ will apply the measurement $\{\Pi_{sym}^{t+1},I-\Pi_{sym}^{t+1}\}$ on $\ket{\phi}^{\otimes t}\ket{\psi}$.
   \mor{It is not clear whether $\mathcal{B}$ does $\rho\to\Pi_{sym}\rho\Pi_{sym}/{\rm Tr}[\Pi_{sym}\rho]$ or $\rho\to\Pi_{sym}\rho\Pi_{sym}+(I-\Pi_{sym})\rho(I-\Pi_{sym})$.}
\end{lemma}

The following corollary follows immediately from the fact that $R_{\ket{\phi}}$ can be used to implement $P_{\ket{\phi}}$.
%\begin{corollary}\label{cor:simpro}
 %   For any $\epsilon>0$, there exists $t = O(1/\epsilon)$ and an efficient algorithm $\mathcal{B}$ such that for all $\ket{\phi}$,
%    \begin{align}
 %   \|P_{\ket{\phi}}(\cdot)  P_{\ket{\phi}}^\dagger - \mathcal{B}(\ket{\phi}^{\otimes t},\cdot)\|_{\diamond} \leq \epsilon.
 %   \end{align}
%\end{corollary}
%\takashi{Should be $P_{\ket{\phi}}(\cdot)  P_{\ket{\phi}}^\dagger$.}\eli{I think we don't ever use this corollary lol}

\begin{lemma}[Variant of Lemma 4.1 from~\cite{EC:AGKL24}]\label{lem:randomortho}
    Let $N,M>0$ be integers with $N \leq M$. Let $\rho$ be a quantum state of dimension $NM$ defined by sampling $N$ states $\ket{\phi_k}\randfrom \HaarSt{M}$ for $k\in[N]$. That is, 
    \begin{align}
    \rho\coloneqq\E_{\ket{\phi_k}\gets \HaarSt{M}, k\in [N]}\left[\bigotimes_{k=1}^N \ketbra{\phi_k}\right].
    \end{align}
%    $$\rho = \bigotimes_{k=1}^n \HaarSt{M}.$$
    Let $\rho'$ be the same state but where we require that each $\ket{\phi_k},\ket{\phi_{k'}}$ are orthogonal. Formally, $\rho'$ is the mixed state of dimension $NM$ defined as 
    \begin{align}
    \rho' \coloneqq \E_{U\gets\HaarUn{M}}\left[\bigotimes_{k=1}^N U\ket{k}\langle k|U^\dagger\right].
    \end{align}
    Then ${\rm TD}(\rho,\rho') \leq \frac{N(N+1)}{2M}$.
\end{lemma}

A formal proof of this lemma is included in~\Cref{sec:randomortho}.

\begin{corollary}\label{cor:ortho}
    Let $t=\poly(n)$ be a polynomial.  Let $\A^{(\cdot)}$ be any $t$-query quantum algorithm. Define
    \begin{align}
            p_1&\coloneqq\Pr_{\ket{\phi_k}\gets \HaarSt{2^n}, k\in [t]}\left[\A^{R_{\ket{\phi_1}},\dots, R_{\ket{\phi_t}}}(\ketbra{\phi_1}\otimes \dots \otimes \ketbra{\phi_t})\right],\\
            p_2&\coloneqq\Pr_{U\gets \HaarUn{2^n}}\left[\A^{R_{U\ket{1}},\dots, R_{U\ket{t}}}(U\ketbra{1}U^\dagger\otimes \dots \otimes U\ketbra{t}U^\dagger)\right].
    \end{align}

    Then $\abs{p_1-p_2} = \negl(n)$.
\end{corollary}

\begin{proof}
    This follows from the fact that $R_{\ket{\phi}}$ can be simulated using $2^{n/8}$ copies of $\ket{\phi}$ up to negligible error by~\Cref{lem:simref}. Therefore, plugging in $M=2^n$ and $N=2^{n/8}\cdot t$, by~\Cref{lem:randomortho} each query can perturb the state by at most $O(N^2/M) = \negl(n)$. 
    Thus, by induction, $\A$'s internal state when acting on orthogonal $U\ket{i}$ is $t\cdot \negl(n)=\negl(n)$ close in trace distance to $\A$'s internal state when acting on random $\ket{\phi_i}$. 
\end{proof}

% \begin{remark}
%     This theorem is not identical to Theorem 4 from~\cite{C:JiLiuSon18}. In particular, they show a stronger version of this lemma where $R_{\ket{\phi}}$ is the unitary reflecting around $\ket{\phi}$. But it is standard that a reflection oracle can be used to perform this projection, and so our lemma follows.
% \end{remark}

\subsection{Complexity-Theoretic No-Cloning}

\begin{theorem}[Complexity-Theoretic No-Cloning~\cite{STOC:AarChr12}]\label{lem:comnocloning}
    Let $n\in\N$ and let $\epsilon\in(0,1)$ such that $1/\epsilon = o(2^n)$. Let $\ket{\psi}$ be an $n$-qubit state chosen uniformly at random from the Haar measure. Given one copy of $\ket{\psi}$, as well as oracle access to $P_{\ket{\psi}}$, a counterfeiter $C$ needs $\Omega(\sqrt{\epsilon}2^{n/2})$ queries to prepare a $2n$-qubit state $\rho$ such that a projector $V^{\otimes 2}_{\ket{\psi}}$ onto $\ketbra{\psi}^{\otimes 2}$ accepts $\rho$ with probability at least $\epsilon$. Here the probability is taken over the choice of $\ket{\psi}$ and the behavior of $V_{\ket{\psi}}^{\otimes 2}$ and $C$.
\end{theorem}

%% file: presampling.tex
%\subsection{Presampling}
\section{Quantum Presampling Lemmas}
\label{sec:presampling}
In this subsection, we outline a technique for proving security of indistinguishability style games relative to fixed random functions. In particular, we adapt the postselected random oracle model from~\cite{TCC:GLLZ21}. This model will operate on functions over arbitrary domains, and so for the following definitions let $N,M\in \N$ be arbitrary, and let $\mathcal{F} \coloneqq \{f:[N]\to [M]\}$.
\begin{definition}
[Postselected random oracle model~\cite{TCC:GLLZ21}]
\label{def:pbfqrom}
    Let $\mathcal{O}^{f}_0$ and $\mathcal{O}^{f}_1$ be two (potentially stateful) oracle algorithms querying $f$. We say that $\mathcal{O}^f_0,\mathcal{O}_1^f$\mor{what is the definition of $\mathcal{O}^f$? An oracle that queries $f$?} are $(s,t,\epsilon)$-indistinguishable in the postselected random oracle model if the following holds:

    Let $\A^f=(\A_s^f,\A_t^f)$ be any pair of (possibly inefficient) quantum oracle algorithms making at most $s,t$ queries respectively. For $b\in \{0,1\}$, define the game $G_b(\A)$ as follows.
    \begin{enumerate}
        \item Sample a random function $f$ using the following process.
        \begin{enumerate}
            \item Sample $f$ uniformly at random from $\mathcal{F}$.
            \item Run $\A_s^{f} \to b$.
            \item If $b=0$, restart this process from (a).
        \end{enumerate}
        \item Run $\A_t^{\mathcal{O}^f_b} \to b'$.
        The output of the game is $b'$.
    \end{enumerate}
    Then 
    \begin{align}
    \max_{\mathcal{A}}\left|\Pr[G_0(\A)\to 1] - \Pr[G_1(\A)\to 1]\right| \leq \epsilon.
    \end{align}
\end{definition}

\begin{theorem}[Theorem 1 from~\cite{TCC:GLLZ21}]\label{thm:presampling}
    Let $\mathcal{O}_1^f,\mathcal{O}_2^f$ be two oracle algorithms making at most $t_\mathcal{O}$ queries to $f\in \mathcal{F}$. 
    Let $S,T\in \N$. Let $\epsilon,\gamma \in (0,1)$. If $\mathcal{O}_1^f,\mathcal{O}_2^f$ are $((S+\log \gamma^{-1})\cdot t\cdot t_{\mathcal{O}},t,(\epsilon-\gamma)/2)$-indistinguishable in the postselected random oracle model, then for all inefficient maps $ADV$ from functions to strings of length $S$, for all $T$-query quantum oracle algorithms $\A^{(\cdot)}$,
    \begin{align}
    \abs{\Pr_{f\gets \mathcal{F}}[\A^{\mathcal{O}_1^f}(ADV(f))\to 1] - \Pr_{f\gets \mathcal{F}}[\A^{\mathcal{O}_2^f}(ADV(f))\to 1]} \leq \epsilon.
    \end{align}
\end{theorem}

We adapt this technique to give the main tool we will use for arguing indistinguishability relative to a fixed random oracle.
\begin{corollary}\label{cor:fixedpresampling}
    Let $\mathcal{O}_1^f,\mathcal{O}_2^f$ be two oracle algorithms making at most $t_\mathcal{O}$ queries to $f$. 
    Let $S,T\in \N$. Let $\epsilon,\gamma \in (0,1)$. If $\mathcal{O}_1^f,\mathcal{O}_2^f$ are $((S+\log \gamma^{-1})\cdot t\cdot t_{\mathcal{O}},t,(\epsilon^2-\gamma)/2)$-indistinguishable in the postselected random oracle model, then for a random function $f$, with probability $\geq 1-\epsilon$ over $f \gets \mathcal{F}$, for all $T$-query quantum oracle algorithms $\A^{(\cdot)}$ and advice strings $c$ such that $|\A|+|c|+1\leq S$,
    \begin{align}
    \abs{\Pr[\A^{\mathcal{O}_1^f}(c)\to 1] - \Pr[\A^{\mathcal{O}_2^f}(c)\to 1]} \leq \epsilon.
    \end{align}
\end{corollary}

\begin{proof}
    Assume towards contradiction that for a random function $f$, with probability $> \epsilon$  over $f$, there exists an adversary $\A_f$ and an advice string $c_f$  such that $\abs{\A_f}+\abs{c_f}\leq S$ and
    \begin{align}
    \abs{\Pr[\A_f^{\mathcal{O}_1^f}(c_f)\to 1] - \Pr[\A_f^{\mathcal{O}_2^f}(c_f)\to 1]} > \epsilon.
    \end{align}
    We will call the set of such $f$: $GOOD$.
    
    We will define $b_f$ to be $1$ if and only if 
    \begin{align}
    \Pr[\A_f^{\mathcal{O}_1^f}(c_f)\to 1] > \Pr[\A_f^{\mathcal{O}_2^f}(c_f)\to 1].
    \end{align}
    In particular,
    $$\Pr[\A_f^{\mathcal{O}_1^f}(c_f)\to b_f] - \Pr[\A_f^{\mathcal{O}_2^f}(c_f)\to b_f] > \epsilon$$
    
    We can define an advice function $ADV(f) = \A_f.c_f.b_f$ for $f\in GOOD$ and $ADV(f) = \bot$ for $f\notin GOOD$. We will define $\mathcal{B}(ADV(f)) = \mathcal{B}(\A_f,c_f,b_f)$ to output $\A_f(c_f) \xor b_f \xor 1$. We will define $\mathcal{B}(\bot) = 0$. We have that for all $f\in GOOD$,
    \begin{equation}
        \begin{split}
            \Pr[\mathcal{B}^{\mathcal{O}_1^f}(ADV(f))\to 1]=\Pr[\mathcal{A}_f^{\mathcal{O}_1^f}(c_f)\to b_f]\\
            \Pr[\mathcal{B}^{\mathcal{O}_2^f}(ADV(f))\to b_f\xor 1]=\Pr[\mathcal{A}_f^{\mathcal{O}_2^f}(c_f)\to b_f].
        \end{split}
    \end{equation}

    We can then compute 
    \begin{equation}
    \begin{split}
            &\Pr_f[\mathcal{B}^{\mathcal{O}_1^f}(ADV(f))\to 1] - \Pr_f[\mathcal{B}_f^{\mathcal{O}_2^f}(ADV(f))\to 1]\\
            &= \Pr_f[f \in GOOD]\cdot \left(\Pr_f[\mathcal{B}^{\mathcal{O}_1^f}(ADV(f))\to 1|f\in GOOD] - \Pr[\mathcal{B}_f^{\mathcal{O}_2^f}(ADV(f))\to 1|f\in GOOD]\right)\\
            &+\Pr_f[f\notin GOOD]\cdot \left(\Pr_f[\mathcal{B}^{\mathcal{O}_1^f}(ADV(f))\to 1|f\notin GOOD] - \Pr[\mathcal{B}_f^{\mathcal{O}_2^f}(ADV(f))\to 1|f\notin GOOD]\right)\\
            &\geq \epsilon\cdot \left(\Pr_f[\mathcal{A}^{\mathcal{O}_1^f}(c)\to b_f|f\in GOOD] - \Pr[\mathcal{A}_f^{\mathcal{O}_2^f}(c)\to b_f|f\in GOOD]\right) + 0\\
            &> \epsilon^2,
    \end{split}
    \end{equation}
    which contradicts~\Cref{thm:presampling} and so we are done.
\end{proof}

%% file: Unitaryseparation.tex
\section{Separating Unitary Oracle}\label{sec:sepsec}
In this section, we introduce our oracle, and show a key lemma, \Cref{keylem}, that will be used later.

\subsection{Definition of Oracle}
\label{sec:unisep}
Our oracle will be parameterized by a sequence of unitaries $\{U_m\}_{m \in \N}$ and 
a sequence of random oracles $\{\RO_\ell\}_{\ell\in\mathbb{N}}$, where $\RO_\ell:\{0,1\}^\ell \times \{0,1\}^\ell \to \{0,1\}^\ell$. 
Each $U_\ell$ will be sampled from the Haar distribution over $\ell$-qubit unitaries and each 
$\RO_\ell$ will be uniformly sampled from the set of functions $\{0,1\}^\ell\times \{0,1\}^\ell\to \{0,1\}^\ell$.
Our oracle consists of two sequences of unitary oracles, $SG\coloneqq\{SG_\ell\}_{\ell\in\mathbb{N}}$ and $Mix\coloneqq\{Mix_\ell\}_{\ell\in\mathbb{N}}$, 
and the classical oracle $PSPACE$:
\begin{enumerate}
    \item $SG_{\ell}$: It is the unitary which swaps $\ket{0}$ \footnote{$\ket{0}$ is a shorthand for $\ket{0...0}$ with the appropriate length.}  and $\frac{1}{\sqrt{2^{\ell}}}\sum_{x\in \{0,1\}^{\ell}\setminus \{0\}}\ket{x}\ket{\phi_x}$, 
    where $\ket{\phi_x}\coloneqq U_{\ell}\ket{x}$, and acts as identity everywhere else. 
      That is,
    \begin{align}
     SG_{\ell}\coloneqq |0\rangle\langle\psi_m|+|\psi_{\ell}\rangle\langle0|+ (I-|0\rangle\langle 0|-|\psi_{\ell}\rangle\langle\psi_{\ell}|),  
    \end{align}
    where $|\psi_{\ell}\rangle\coloneqq\frac{1}{\sqrt{2^{\ell}}}\sum_{x\in\{0,1\}^{\ell}\setminus \{0\}}|x\rangle|\phi_x\rangle$.
  
    \item $Mix_{\ell}$: 
    It is the following unitary.
    \begin{align}
    &\sum_{x,y\in\{0,1\}^{\ell}}|x\rangle\langle x|\otimes|y\rangle\langle y|\otimes(|\phi_x\rangle\langle\phi_x|+|\phi_y\rangle\langle\phi_y|)\otimes X^{\RO(x,y)}\otimes X    \\
    &+\sum_{x,y\in\{0,1\}^{\ell}}|x\rangle\langle x|\otimes|y\rangle\langle y|\otimes(I-|\phi_x\rangle\langle\phi_x|-|\phi_y\rangle\langle\phi_y|)\otimes I\otimes I.
    \end{align}
    Here, $X$ is the Pauli $X$ operator and $X^r\coloneqq\bigotimes_{i=1}^\secpa X_i^{r_i}$ for any $\secpa$-bit string $r$.
    In other words, $Mix_{\ell}$ performs the following operation:
    \begin{align}
        \ket{x}\ket{y}\ket{\phi_x}\ket{z}\ket{b}&\rightarrow \ket{x}\ket{y}\ket{\phi_x}\ket{z\oplus \RO(x,y)}\ket{\bar{b}}\\
        \ket{x}\ket{y}\ket{\phi_y}\ket{z}\ket{b}&\rightarrow\ket{x}\ket{y}\ket{\phi_y}\ket{z\oplus \RO(x,y)}\ket{\bar{b}}
    \end{align}
    and acts as identity for other orthogonal states.
    \item $PSPACE$: an oracle for a $\mathbf{PSPACE}$-complete problem.\footnote{For the oracle, we use the normal font, and for the complexity class, we use the bold font.}
\end{enumerate}

\subsection{Key Lemma}
 
We will define the following oracles $SG_{\text{Samp}}\coloneqq\{SG_{\text{Samp},\ell}\}_{\ell \in \mathbb{N}}$ and $Ver_{\text{Samp}}\coloneqq\{Ver_{\text{Samp},\ell}\}_{\ell \in \mathbb{N}}$ which we will use to simulate $SG$ and $Mix$. These oracles will be further parameterized by some function $t$, which will represent the maximum number of times these oracles can be queried.

   \paragraph{Initialization:} 
    \begin{itemize}
        \item For each $\ell\in\mathbb{N}$, sample an $\ell$-qubit Haar random unitary $\wt{U}_\ell$. (Note that $\wt{U}_\ell$ are sampled independently from $U_\ell$).
        \item For each $\ell\in\mathbb{N}$, sample $\wt{\RO_{\ell}}:\{0,1\}^\ell\times\{0,1\}^\ell\to\{0,1\}^\ell$. (Note that
        $\wt{\RO_{\ell}}$ are sampled independently from $\RO_\ell$). 
        \item Sample $k_1, \dots k_t \gets \{0,1\}^\ell\setminus \{0\}$, where $t$ is some sufficiently large polynomial in $\ell$, and set $\wt{\ket{\phi_{k_1}}},...,\wt{\ket{\phi_{k_t}}}$, where 
        $\wt{\ket{\phi_{k_i}}}\coloneqq \wt{U}_\ell\ket{k_i}$. 
    \end{itemize}
     \begin{enumerate}
   \item $SG_{\text{Samp},\ell}$:
     \begin{itemize}
        \item On the $i^{th}$ query for $i\leq t$, output $\ket{k_i}\wt{\ket{\phi_{k_i}}}$.
       % \mor{What happens for $t+1$ th query?}
        \item On the $i^{th}$ query for $i>t$, output $\ket{\bot}$.
    \end{itemize}
    \item $Ver_{\text{Samp},\ell}$:
    \begin{itemize}
        \item For all $i \in [t]$ perform the mapping
        $\ket{k_i}\wt{\ket{\phi_{k_i}}}\ket{b}\rightarrow \ket{k_i}\wt{\ket{\phi_{k_i}}}\ket{\bar{b}}$
    and acts as identity on all other orthogonal states.
%        \item Act as identity on all other inputs
    \end{itemize}

    \end{enumerate}
     Next, we will describe the oracle algorithms 
     $\text{Sim}_{{Mix}}^{Ver_{\text{Samp}}}\coloneqq\{\text{Sim}_{{Mix}_{\ell}}^{Ver_{\text{Samp},\ell}}\}_{\ell \in \mathbb{N}}$ and $\text{Sim}_{{SG}}^{SG_{\text{Samp}}}\coloneqq\{\text{Sim}_{{SG}_{\ell}}^{SG_{\text{Samp},\ell}}\}_{\ell \in \mathbb{N}}$. These algorithms will be parameterized by some function $t'$, which will be specified as required.
%\eli{$z_{x,y}$ is undefined, it should be sampled at initialization of $\SimMix$.}
    \paragraph{$\text{Sim}_{{Mix}_{\ell}}^{Ver_{\text{Samp},\ell}}$}
   
    \begin{enumerate}
        \item Initialize 3 ancilla registers $\ket{0}_{\mathbf{F}}\ket{0}_{\mathbf{G}}\ket{0}_{\mathbf{H}}$.
        \item On input $\ket{x}_{\mathbf{A}}\ket{y}_{\mathbf{B}}\ket{\psi}_{\mathbf{C}}\ket{z}_{\mathbf{D}}\ket{b}_{\mathbf{E}}$, do the following:
        \begin{itemize}
            \item Query $Ver_{\text{Samp},\ell}$ 
            on input $|x\rangle_{\mathbf{A}}|\psi\rangle_{\mathbf{C}}|0\rangle_{\mathbf{F}}$.
            \item Query $Ver_{\text{Samp},\ell}$ 
            on input $|y\rangle_{\mathbf{B}}|\psi\rangle_{\mathbf{C}}|0\rangle_{\mathbf{G}}$.
            \item Compute the OR of the contents of registers $\mathbf{F}$ and $\mathbf{G}$, and XOR the result into the register $\mathbf{H}$.
        \end{itemize}
        \item Controlled on the register $\mathbf{H}$, add $\wt{\RO_\ell}(x,y)$ to the register $\mathbf{D}$, i.e.,
        $|x\rangle_{\mathbf{A}}|y\rangle_{\mathbf{B}}|z\rangle_{\mathbf{D}}\rightarrow |x\rangle_{\mathbf{A}}|y\rangle_{\mathbf{B}}|z\oplus \wt{\RO_\ell}(x,y)\rangle_{\mathbf{D}}$. 
        \item Controlled on the register $\mathbf{H}$, flip the bit of the register $\mathbf{E}$.
        \item Uncompute ancilla registers $\mathbf{F,G,H}$ by querying 
        $Ver_{\text{Samp},\ell}$ 
        again.
    \end{enumerate}
    
    \paragraph{$\text{Sim}_{{SG}_{\ell}}^{SG_{\text{Samp},\ell}}$}
    
    \begin{enumerate}
        \item Sample $\eta\gets Bin(t,1/2)$ where $Bin(t,1/2)$ is the binomial distribution. That is,
        \begin{align}
        \Pr[Bin(t,1/2)=\eta] = \frac{{t\choose \eta}}{2^t}.
        \end{align}
        \item Query 
        $SG_{\text{Samp},\ell}$ 
        a total of $\eta$ times, receiving output states 
        $\ket{k_1}\wt{\ket{\phi_{k_1}}}\dots \ket{k_\eta}\wt{\ket{\phi_{k_\eta}}}$.
    %    \mor{Do you mean "Query $i$ to $SG_{\text{Samp}}$ for each $i\in[\eta]$"? Currently, which integer is queried to the oracle is not clear.}\takashi{$SG_{\text{Samp}}$ doesn't take any input, so I don't find any issue with the current writing.}
        
        \item For $i>\eta$, set $k_i=0$. For ease of notation, we will define 
        $\wt{\ket{\phi_0}}=\ket{0}$.
        
        \item Generate the state
        \begin{align}
        \left(\sum_{\pi \in Sym(t)} \bigotimes_{i=1}^t
        |k_{\pi(i)}\rangle_{\mathbf{A_i}}
        |0\rangle_{\mathbf{B_i}}\right)\otimes 
        \left(\bigotimes_{j=1}^t |k_j\rangle_{\mathbf{C_j}}\wt{|\phi_{k_j}\rangle}_{\mathbf{D_j}}\right).
        \end{align}
        \item For all $i,j$, apply $SWAP_{\mathbf{B_i},\mathbf{D_j}}$ controlled on 
        registers $\mathbf{A_i},\mathbf{C_j}$ containing the same value in the standard basis. This results in the state
        \begin{align}
        \left(\sum_{\pi \in Sym(t)} \bigotimes_{i=1}^t|k_{\pi(i)}\rangle_{\mathbf{A_i}}\wt{|\phi_{k_{\pi(i)}}\rangle}_{\mathbf{B_i}}\right)\otimes \left(\bigotimes_{j=1}^t |k_j\rangle_{\mathbf{C_j}}|0\rangle_{\mathbf{D_j}}\right).
        \end{align}
        \item Next, trace out registers $\mathbf{C_1 D_1\dots C_t D_t}$. This will contain the state
        \begin{align}
        \sum_{\pi \in Sym(t)} \bigotimes_{i=1}^t\ket{k_{\pi(i)}}_{\mathbf{A_i}}\wt{\ket{\phi_{k_{\pi(i)}}}}_{\mathbf{B_i}}.
        \end{align}
       
        \item Finally, let $\mathbf{In}$ be the input register. 
        Apply the measurement $\{\Pi_{sym}^{t'+1},I-\Pi_{sym}^{t'+1}\}$\mor{$t$, not $t'$?} on $\mathbf{In,A_1B_1},\dots,\mathbf{A_tB_t}$, where $\Pi_{sym}^{t'+1}$ is projection onto the symmetric subspace.
        \mor{It is helpful to remind the reader the definition of $\Pi_{sym}$, because its definition is far away, and hidden in a text of a lemma.}
    \end{enumerate}

\begin{comment}
    \begin{enumerate}
 \item $SG_{\text{samp}}:$ sample $x \leftarrow \{0,1\}^n$,  and output $(x, \ket{\phi_x})$. Here $\ket{\phi_x}= \wt{\U}\ket{x, 0^n}$.

\end{enumerate}
\end{comment}

\begin{comment}
\begin{lemma}

    With probability at least $1-\text{negl}(n)$ over  $\{U_\secpa\}_{\secpa\in \mathbb{N}}$ and  $\{\RO_{\secpa}\}_{\secpa\in \mathbb{N}}$,  for all $T$ query algorithms $\mathcal{A}(1^{\secpa}, c)$, $\forall c \in \{0,1\}^{poly(\secpa)}$, where $c$ is the advice string, for all polynomials $p$, there exist  simulators $\SimSG,\SimMix$ such that
    \[
    \Big|\Pr[\mathcal{A}^{Mix, SG}(1^n,  c)=1]-\Pr[\mathcal{A}^{\Sim^{Ver_{\text{Samp}}}_{Mix},\Sim^{SG_{\text{Samp}}}_{SG}}(1^n, c)=1]\Big|\leq \frac{1}{p(\secpa)}
    \]
   % \eli{Should actually be that for all polynomials $p$, there exists simulators $\SimSG,\SimMix$ such that the difference is $1/p(n)$, since the difference depends on the number of copies used for the reflection oracle, not negligible.}
\end{lemma}
\end{comment}

\begin{lemma}
\label{keylem}
 Let $n$ be the security parameter, and let
 $S(n),T(n)$ be any functions such that $T\leq 2^{n/10}$. With probability at least $1-O(S\cdot 2^{-n/4})$ over 
 $U_n$ and $\RO_n$, 
 for all $T$ query algorithms $\mathcal{A}(1^{\secpa}, c)$ with $|\A|+|c|\leq S$\mor{needs an explanation for the definition of $|\A,c|$}\eli{I added this to notation section}\mor{Yes, still $|\A,c|$ is not defined. Do you mean $|\A|+|c|$?}, where $c$ is the advice string, for all $\epsilon$, there exist  simulators $\text{Sim}_{{Mix}_{n}}^{Ver_{\text{Samp},n}},\text{Sim}_{{SG}_{n}}^{SG_{\text{Samp},n}}$ running in time $O(n,1/\epsilon)$ such that \saachi{ fixed params here}
    \begin{align}
    \Big|\Pr[\mathcal{A}^{Mix_n, SG_n}(1^n,  c)=1]-\Pr[\mathcal{A}^{\text{Sim}_{{Mix}_{n}}^{Ver_{\text{Samp},n}},\text{Sim}_{{SG}_{n}}^{SG_{\text{Samp},n}}}(1^n, c)=1]\Big|\leq \epsilon + O\left(\frac{T^{2.5}+S}{2^{n/8}}\right).
    \end{align}

%    \mor{Don't $\SimSG$ and $\SimMix$ take $\epsilon$ as input? (Otherwise, they do not know $\epsilon$.)}
 %   \eli{Should actually be that for all polynomials $p$, there exists simulators $\SimSG,\SimMix$ such that the difference is $1/p(n)$, since the difference depends on the number of copies used for the reflection oracle, not negligible.}
\end{lemma}

If both $S$ and $T$ are polynomials, we obtain the following corollary.
\begin{corollary}\label{cor:keylem}
Let $n$ be the security parameter.
    Let $S(n),T(n)$ be two polynomials. With probability at least $1-\negl(n)$ over  $U_n$ and  $\RO_n$,  for all $T$ query algorithms $\mathcal{A}(1^{\secpa}, c)$ with $|\A|+|c|\leq S$\mor{definition of $|\A,c|$ should be given}, where $c$ is the advice string, for all polynomials $1/\epsilon$, there exist simulators $\text{Sim}_{{Mix}_{n}}^{Ver_{\text{Samp},n}},\text{Sim}_{{SG}_{n}}^{SG_{\text{Samp},n}}$ running in time $\poly(n,1/\epsilon)$ such that for all sufficiently large $n$,
  \begin{align}
    \Big|\Pr[\mathcal{A}^{Mix_n, SG_n}(1^n,  c)=1]-\Pr[\mathcal{A}^{\text{Sim}_{{Mix}_{n}}^{Ver_{\text{Samp},n}},\text{Sim}_{{SG}_{n}}^{SG_{\text{Samp},n}}}(1^n, c)=1]\Big|\leq \epsilon.
    \end{align} 
\end{corollary}

Before jumping into the proof of~\Cref{keylem}, we will formally state an important lemma
\begin{lemma}\label{lem:sgsim}
    There exists an efficient oracle algorithm $\Sim^{(\cdot)}(1^t)$ such that the following is true:
    Let $n$ be any natural number.
    Let $\{\ket{\phi_k}\}_{k\in \{0,1\}^n}$ be any collection of states. For a function $f:\{0,1\}^n \to \{0,1\}^n=[2^n]$, define the state 
    \begin{align}
    \ket{\phi^f}\coloneqq\frac{1}{\sqrt{2}}\ket{0} - \frac{1}{\sqrt{2}}\frac{1}{\sqrt{2^n-1}}\sum_{k\in \{0,1\}^n\setminus \{0\}}\omega_{2^n}^{f(k)}\ket{k}\ket{\phi_k}.
    \end{align}
    Let $Gen^{\{\ket{\phi_k}\}_k}$ be the CPTP oracle which acts as follows
    \begin{enumerate}
        \item Sample $k\gets \{0,1\}^n\setminus \{0\}$.
        \item Output $\ket{k}\ket{\phi_k}$.
    \end{enumerate}
    Then for all $t$, we have
    \begin{align}
    \E_f\left[|\phi^f\rangle\langle\phi^f|^{\otimes t}\right] = \Sim^{Gen^{\{|\phi_k\rangle\}_k}}(1^t).
    \end{align}
\end{lemma}

\Cref{lem:sgsim} is proven and explained in more detail in~\Cref{sec:proof_lem_sgsim}.

We then proceed to our proof of~\Cref{keylem}.

\begin{proof}
    We will let $\VerSampn,\SGSampn,\SimMixn,\SimSGn$ be as in the previous section with parameters $t=T$ and $t'= \frac{c_{\epsilon}}{\epsilon}$ for some sufficiently large $c_{\epsilon}$.
    
    We will establish the proof of Lemma \ref{keylem} through the following sequence of hybrids:\\
    
    \noindent{\textbf{Hybrid 0:}} $\mathcal{A}^{Mix_n, SG_n}(1^n,c)$. \\
    \begin{comment}
    \noindent{\textbf{Hybrid 1:}} Same as \textbf{Hybrid 0}, except for all $\ell \in [d]$, where $d=?$  all calls to $\U_{\ell}$  in $Mix, SG$ are replaced in the following manner:
    \begin{itemize}

    \item  \saachi{fill this in }. For all $\ell\in [d]$,  perform tomography on each $U_{\ell}$, producing estimates $\U_{\ell}'$.
    Every query made to $U_{\ell}$ is simulated by a query to $\U_{\ell}'$. The estimates are produced as follows:
    \begin{itemize}
     \item Run $SG$ until, $\forall k \in [2^\ell]$, there exist $t>2^{d}\cdot ?$  \saachi{fill this in} copies of $\ket{\phi_k}$.
    \item  For every $k \in [2^{\ell}]$, run the state tomography algorithm from Theorem \ref{thm:statetom} on inputs $n=?$ \saachi{fill this in } and $\epsilon = ?$ \saachi {fill this in}. This produces an estimate $\ket{\phi'_k}$ such that $|\langle \phi_k|\phi'_k\rangle|^2 \geq 1-\epsilon$ with probability at least $1-?$  \saachi{fill this in}.
    \end{itemize}
    \end{itemize}
    \end{comment}
    \noindent{\textbf{Hybrid 1:}} 
    Same as \textbf{Hybrid 0}, but all calls to $U_{n}$ in $Mix_n$ and $SG_n$ are replaced with a freshly sampled Haar random unitary $\wt{U}_{n}$ over $n$ qubits.  Define $\wt{SG}_n$ and $\wt{Mix}_n$ to be the same as $SG_n$ and $Mix_n$, except that queries to $\U_{n}$ are replaced by queries to $\wt{U}_{n}$.\\ 
    \begin{comment}
    \noindent{\textbf{Hybrid 3:}} Same as \textbf{Hybrid 2}, but for all $\ell \in [d+1, \text{poly}(n)]$, all calls to $\wt{U}_{\ell}$ are simulated by a path recording oracle $V$. \\
    \end{comment}
     \begin{claim}
    With probability $\geq 1 - 2^S\exp\left(-2^{n/2}\right)$ over $U_{n}$,  and with probability 1 over  $\RO_{n}$, for all sufficiently large $n$,
    \begin{align}
    \max_{\mathcal{A}, c}\abs{\Pr[1\gets{\bf Hybrid~1}(1^\secpa)] - \Pr[1\gets{\bf Hybrid~0}(1^\secpa)]} \leq \frac{1}{2^{n/8}}.
    \end{align} 
\end{claim}
\begin{proof}
Let $f(U_n)=\Pr[1\gets\textbf{Hybrid 0}(1^\secpa)]$. Then,
~\Cref{lem:lipshitz} implies that $f$ is $2T$ Lipshitz. Now observe that $\Pr[1\gets\textbf{Hybrid 1}(1^\secpa)]=\mathbb{E}_{{U_n}}\Pr[1\gets\textbf{Hybrid 1}(1^\secpa)]$.

Invoking the strong concentration of the Haar measure in \Cref{thm:conc} with $t = \delta$ and $L = 2T$ gives us 
%$U_n$ 
%But it is clear that
%\begin{align}
%\Pr[1\gets\textbf{Hybrid 1}(1^\secpa)] = \E_{U_m}[\Pr[1\gets\textbf{Hybrid 1}(1^\secpa)]]
%\end{align}
\begin{align}
        \Pr_{U_n}\left[
        \left|\Pr_{\wt{U}_n}[1\gets\textbf{Hybrid 1}(1^\secpa)] - \Pr[1\gets\textbf{Hybrid 0}(1^\secpa)]\right| \geq \delta\right] \\
        \leq 2\exp\left(-\frac{(2^{n} - 2)\delta^2}{96T^2}\right)
        \leq \exp\left(-2^{n/2}\right).
\end{align}

for all sufficiently large $n$.
\mor{Are you actually showing not lemma but corollary? I mean are you assuming $T=poly$?}\takashi{I got the same impression. Otherwise, $T$ should remain in the above bound.}
Setting $\delta = 2^{-n/8}$ and taking the union bound over all adversaries $\A$ and all $c$ gives us that, with probability at least $1-2^{S}\cdot \exp\left(-2^{n/2}\right)$ 
over the choice of $U_{n}$ and with probability $1$ over the choice of $\RO_n$, for all sufficiently large $n$
\begin{align}
\max_{\mathcal{A}, c}\abs{\Pr[1\gets\textbf{Hybrid 1}(1^\secpa)] - \Pr[1\gets\textbf{Hybrid 0}(1^\secpa)]} \leq \frac{1}{2^{n/8}}+\exp\left(-2^{(n/2)+1}\right).
\end{align}
\mor{for all sufficiently large $n$?}
\end{proof}

    \noindent{\textbf{Hybrid 2:}} Same as \textbf{Hybrid 1}, except that the oracle $\wt{SG}_n$ is replaced 
    with the following oracle $SG'_n$:
    \begin{itemize}
        \item  $SG'_n:$ Let 
        $\wt{\ket{\phi_x}}=\wt{U}_n\ket{x}$. 
        Then swap $\ket{0}$ and 
        $\frac{1}{\sqrt{2^n}}\sum_{x\in \{0,1\}^n}\omega_{2^n}^{\RO^*_n(x)}\ket{x}\wt{\ket{\phi_x}}$, 
        and act as identity everywhere else. Here, $\RO^*_n$ is a freshly sampled random oracle that's independent of $\RO_n$. \footnote{Here $\omega_{2^n}$ is the $2^n$-th root of unity. 
        Moreover, we interpret $\RO_n(x)\in\{0,1\}^n$ as an element of $[2^n]$.} 
    \end{itemize}
    \begin{claim}
 For all $n$ and for all $U_n,\RO_n$,        
 \begin{align}
\Pr[1\gets{\bf Hybrid~1}(1^\secpa)] = \Pr[1\gets{\bf Hybrid~2}(1^\secpa)].
 \end{align}
 \mor{for all $n$?}
    \end{claim}
    \begin{proof}
    \eli{This is a little sketchy, we need to argue that $\wt{Mix}_n$ is phase invariant otherwise we can't just replace. Here's another attempt:} We observe that ${\bf Hybrid-2}$ is actually identical to ${\bf Hybrid-1}$, except instead of sampling $\wt{U}_n$ as a Haar random unitary over $n$ qubits, we instead sample ${U}'_n$ as a Haar random unitary over $n$ qubits and set $\wt{U}_n = U'_n \cdot S^{\RO^*}$ where $S^{RO^*}$ is the unitary which maps $\ket{x} \to \omega_{2^n}^{\RO^*(x)}\ket{x}$. In particular, since the behavior $\wt{Mix}_n$ is independent of the phase of $\wt{U}\ket{x}$ for all $x$, $\wt{Mix}_n$ acts the same when instantiated with either $\wt{U}_n$ or $U'_n$.\eli{Old version:
        Recall that $\frac{1}{\sqrt{2^n}}\sum_{x\in \{0,1\}^n}\ket{x}\wt{\ket{\phi_x}}=\frac{1}{\sqrt{2^n}}\sum_{x\in \{0,1\}^n}\ket{x} \wt{U}_n\ket{x} $ and when $\wt{U}_n$ is haar-random, for any fixed  $\RO^*$,  $\wt{U}_n\sum_{x\in\{0,1\}^n}\omega_{2^n}^{\RO^*_n(x)}\ket{x}\bra{x}$ is haar random from unitary invariance. So $\wt{SG}_n$ and $SG'_n$ are identical. Since $\wt{Mix}_n$ remains unchanged,  the 2 hybrids are identical. }
    \end{proof}
    \noindent{\textbf{Hybrid 3:}} Same as \noindent{\textbf{Hybrid 2:}}, except $SG'_n$ is simulated by the simulator $\text{Sim}_{{SG}_{n}}^{SG_{\text{Samp},n}}$.
    \begin{claim}
        For all $n$, with probability $1$ over $U_n$ and   
        $\RO_n$, 
        \begin{align}
        \abs{\Pr[1\gets{\bf Hybrid~3}(1^\secpa)] - \Pr[1\gets{\bf Hybrid~2}(1^\secpa)]}\leq \epsilon.
        \end{align}
        \mor{for all $n$?}
    \end{claim}
    \begin{proof}
        First, observe that swapping between $\ket{0}$ and $\frac{1}{\sqrt{2^n}}\sum_{x\in \{0,1\}^n\setminus \{0\}}\omega_{2^n}^{\RO^*_n(x)}\ket{x}\wt{\ket{\phi_x}}$ is the same as applying $R_{\ket{\phi-}}$, where $\ket{\phi-}$ is $\ket{0}-\frac{1}{\sqrt{2^n}}\sum_{x \in \{0,1\}^n\setminus \{0\}}\omega_{2^n}^{\RO^*_n(x)}\ket{x}\wt{\ket{\phi_x}}$.  Let $\mathcal{B}$ be the algorithm from~\Cref{lem:simref}. We know that
        \begin{align}
        \|R_{\ket{\phi-}}(\cdot)R_{\ket{\phi-}}^\dagger - \mathcal{B}(\ket{\phi-}^{\otimes t'},\cdot)\|_{\diamond} \leq O(1/t'),
        \end{align} 
        where $t =\poly(\secpa)$. 
        Now, from \Cref{lem:sgsim}, we can conclude that the final state prepared by $\text{Sim}_{{SG}_n}^{SG_{\text{Samp},n}}$, i.e., 
        \begin{align}
        \left(\sum_{\pi \in Sym(t)} \bigotimes_{i=1}^t\ket{k_{\pi(i)}}_{\mathbf{A_i}}\wt{\ket{\phi_{k_{\pi(i)}}}}_{\mathbf{B_i}}\right)\left(\sum_{\pi \in Sym(t)} \bigotimes_{i=1}^t\bra{k_{\pi(i)}}_{\mathbf{A_i}}\wt{\bra{\phi_{k_{\pi(i)}}}}_{\mathbf{B_i}}\right)
        \end{align}
        is identical as a mixed state to $\ket{\phi-}\bra{\phi-}^{\otimes t}$.  
        By the construction of $\SimSGn$,
        \begin{align}
        \|SG'_n - \text{Sim}_{{SG}_{n}}^{SG_{\text{Samp},n}}\|_{\diamond} \leq O(1/t').
        \end{align}
        Therefore, by setting $t'$ sufficiently large and applying~\Cref{fact},\mor{Link is broken} the claim follows.
    \end{proof}
     \noindent{\textbf{Hybrid 4:}} Same as \noindent{\textbf{Hybrid 3}} except $\wt{Mix}_n$ is replaced with the simulator $\text{Sim'}_{{Mix}_n}^{Ver_{\text{Samp},n}}$, which is exactly the same as $\text{Sim}_{{Mix}_n}^{Ver_{\text{Samp},n}}$, except that controlled on register $\textbf{H}$,  $\RO_n(x,y)$ is added to the register $\mathbf{D}$ instead of $\wt{\RO}_n(x,y)$, i.e.,
        $|x\rangle_{\mathbf{A}}|y\rangle_{\mathbf{B}}|z\rangle_{\mathbf{D}}\rightarrow |x\rangle_{\mathbf{A}}|y\rangle_{\mathbf{B}}|z\oplus \RO_n(x,y)\rangle_{\mathbf{D}}$.
        $\text{Sim'}_{{Mix_n}}^{Ver_{\text{Samp},n}}$ is precisely defined as follows:
       \begin{enumerate}
        \item Initialize 3 ancilla registers $\ket{0}_{\mathbf{F}}\ket{0}_{\mathbf{G}}\ket{0}_{\mathbf{H}}$.
        \item On input $\ket{x}_{\mathbf{A}}\ket{y}_{\mathbf{B}}\ket{\psi}_{\mathbf{C}}\ket{z}_{\mathbf{D}}\ket{b}_{\mathbf{E}}$, do the following:
        \begin{itemize}
            \item Query $Ver_{\text{Samp},n}$ 
            on input $|x\rangle_{\mathbf{A}}|\psi\rangle_{\mathbf{C}}|0\rangle_{\mathbf{F}}$.  
            \item Query $Ver_{\text{Samp},n}$ 
            on input $|y\rangle_{\mathbf{B}}|\psi\rangle_{\mathbf{C}}|0\rangle_{\mathbf{G}}$.
            \item Compute the OR of the contents of registers $\mathbf{F}$ and $\mathbf{G}$, and XOR the result into the register $\mathbf{H}$.
        \end{itemize}
        \item Controlled on the register $\mathbf{H}$, add $\RO_n(x,y)$ to the register $\mathbf{D}$, i.e.,
        $|x\rangle_{\mathbf{A}}|y\rangle_{\mathbf{B}}|z\rangle_{\mathbf{D}}\rightarrow |x\rangle_{\mathbf{A}}|y\rangle_{\mathbf{B}}|z\oplus \RO_n(x,y)\rangle_{\mathbf{D}}$. 
        \item Controlled on the register $\mathbf{H}$, flip the bit of the register $\mathbf{E}$.
        \item Uncompute ancilla registers $\mathbf{F,G,H}$ by querying 
        $Ver_{\text{Samp},n}$ 
        again.
    \end{enumerate}\takashi{The definition seems say that $\text{Sim}_{{Mix}_m}^{Ver_{\text{Samp},m}}$ always uses $\RO$ instead of $\RO'$. I'm wondering if this is really what you intend. From what I can see from its proof, I guess what you actually intend may be that it uses $\RO$ when $x \in \text{Supp}(SG_{\text{Samp},m})$ or $y \in \text{Supp}(SG_{\text{Samp},m})$, but otherwise uses $\RO'$. For this definition to make sense, we should also force all the queries to $SG_{\text{Samp},m}$ to happen at the beginning wlog. 
        }\saachi{I think the simulator for SG has this guarantee. }\takashi{I see, you are right. I misunderstood the definition of $Ver_{\text{Samp},m}$; I thought it applies $\ket{k}\wt{\ket{\phi_{k}}}\ket{b}\rightarrow \ket{k}\wt{\ket{\phi_{k}}}\ket{\bar{b}}$ for all $k$ rather than only for $k_1,...,k_t$. Perhaps, it may help to give a brief explanation on why the difference between $\wt{Mix}_n$ and $\text{Sim'}_{{Mix_n}}^{Ver_{\text{Samp},n}}$ may occur only on those states 
        for avoiding such a misunderstanding.} \saachi{elaborated on this hybrid}
    \begin{claim}
        With probability $\geq 1 - \exp\left(-2^n\right)$ \takashi{How was this probability derived? Do we implicitly do the averaging argument?} over $U_{n},\RO_{n}$,    
        \begin{align}
        \max_{\mathcal{A},c}\abs{\Pr[1\gets{\bf Hybrid~3}(1^\secpa)] - \Pr[1\gets{\bf Hybrid~4}(1^\secpa)]}\leq O\left(\frac{T^2}{2^n}\right).
        \end{align}
    \end{claim}
    \begin{proof}
    Let $\{k_1, \dots k_t\}$ be the set of all inputs to $SG_{\text{Samp,n}}$ chosen during \textbf{Initialization}.  Since $Ver_{\text{Samp},n}$ only performs the mapping 
        $\ket{k_i}\wt{\ket{\phi_{k_i}}}\ket{b}\rightarrow \ket{k_i}\wt{\ket{\phi_{k_i}}}\ket{\bar{b}}$ for $k_i \in \{k_1, \dots k_t\}$,  
     by construction of $\text{Sim'}_{{Mix_n}}^{Ver_{\text{Samp},n}}$, the oracles $\wt{Mix}_n$ and $\text{Sim'}_{{Mix_n}}^{Ver_{\text{Samp},n}}$ are identical on all inputs of the following form: 
       \begin{itemize}
           \item $\ket{x}\ket{y}\wt{\ket{\phi_x}}\ket{z}\ket{b}, \forall x \in \{k_1, \dots k_t\}$.
           \item $\ket{x}\ket{y}\wt{\ket{\phi_y}}\ket{z}\ket{b}, \forall y  \in \{k_1, \dots k_t\}$.
       \end{itemize}
    which implies that $\wt{Mix}_n$ and $\text{Sim'}_{{Mix_n}}^{Ver_{\text{Samp},n}}$ only differ on all inputs of the form $\ket{x}\ket{y}\wt{\ket{\phi_x}}\ket{z}\ket{b}, \forall x \notin \{k_1, \dots k_t\}$ and $\ket{x}\ket{y}\wt{\ket{\phi_y}}\ket{z}\ket{b}, \forall y \notin \{k_1, \dots k_t\}$. Let the subspace spanned by all these differing inputs be $S$. Invoking \Cref{cor:o2h}, we can conclude that 
    \[
   \Big| \Pr[\A^{\wt{Mix}_n}=1]-\Pr[\A^{\text{Sim'}_{{Mix_n}}^{Ver_{\text{Samp},n}}}=1]\Big|\leq 2\sum_{i=0}^{T-1}\sqrt{||\Pi_{\cal{S}}\ket{\psi_i^{\wt{U_n}}}||^2}
    \]
where $\Pi_S$ is the projector onto $S$ and $\ket{\psi_i^{\wt{U}_n}}$ is as in~\Cref{cor:o2h}.
    Since $\wt{\RO}$ is queryable only through $\wt{Mix}_n$, for each $i$, $||\Pi_{\cal{S}}\ket{\psi_i^{\wt{\RO_n}}}||$ is at most the probability that the $n$ qubit state $\wt{\ket{\phi_k}}$ for $k \notin \{k_1, \dots k_t\}$ is cloned by $\cal{A}$. From Theorem \ref{lem:comnocloning}, we know that for any adversary, this probability is upper bounded by $O\left(\frac{T^{2}}{2^{n}}\right)$. Therefore, 
    \[
    \Big| \Pr[\A^{\wt{Mix}_n}=1]-\Pr[\A^{\text{Sim'}_{{Mix_n}}^{Ver_{\text{Samp},n}}}=1]\Big|\leq 2 \sqrt{T}\cdot O\left(\frac{T^2}{2^{n}}\right).
    \]\takashi{Is this trivial? I assume this is done based on some kind of BBBV-style argument or one-way to hiding lemma, but then I feel there may be an additional square loss. I think more clarification is needed unless there's any obvious argument I'm missing.}\saachi{need to  add o2h }
     
    \end{proof}
Thus, so far we have established that, with probability $\geq 1-2^S\exp(-2^n)$ over $U_n, \RO_n$, 
\begin{align}
\max_{\mathcal{A},c}\abs{\Pr[1 \gets \textbf{Hybrid 4}(1^\secpa)]- \Pr[1 \gets \textbf{Hybrid 0}(1^\secpa)]}\leq \epsilon + O\left(\frac{T^{2.5}}{2^{n/8}}\right). 
\end{align}

\takashi{
(This is just a suggestion for improving readability. We do not necessarily have to revise in this way.)
As far as I can tell, what is needed to complete the rest of the proof is the following style claim. For all but exponentially small fraction of $\RO$, for any size bounded adversary $(\mathcal{A},c)$, the adversary (with no query) can distinguish $(k_1,...,k_t,\RO(k_1),....,\RO(k_t))$ and $(k_1,...,k_t,u_1,....,u_t)$
only with an exponentially small advantage, where $k_1,...,k_t$ are uniformly random inputs and $u_1,...,u_t$ are uniformly random outputs. In that case, I feel stating this as a lemma and just invoke it in the main proof would be much easier to follow since in this way we can completely remove the mentioning to compressed oracle in the main proof. (We will use the compressed oracle in the proof of the above lemma anyway, but I feel separating the main proof and the proof of the lemma would be easier to follow.)
}

Now we will introduce additional hybrids to switch to the post selected random oracle model
in order to complete the proof. \\
 \noindent{\textbf{Hybrid 5:}} In this hybrid, we move to the post selected random oracle model (\Cref{def:pbfqrom}) to remove the advice $c$. In particular, we will let $s,t$ be functions to be set later. In \noindent{\textbf{Hybrid 5}} the adversary $\mathcal{A}$ is modelled as $(\mathcal{A}_s,\mathcal{A}_t)$, and has access to the oracle $(\text{Sim'}_{Mix}^{Ver_{\text{Samp}}},\text{Sim}_{SG}^{SG_{\text{Samp}}})$ as defined above. Here $\A_s$ makes at most $s$ queries to its oracle and $\A_t$ makes at most $t$ queries to its oracle. Formally, \textbf{Hybrid 5} is defined as follows:  
 \begin{enumerate}
     \item Sample $\RO_n$ as follows: 
 
 \begin{enumerate}
     
            \item Sample $\RO_n$ uniformly at random.
            \item Run $\A_s^{\RO_n} \to b$.
            \item If $b=0$, restart this process from item (a).

 \end{enumerate}
 \item Run $\mathcal{A}^{\Big(\text{Sim'}_{Mix}^{Ver_{\text{Samp}}^{\RO_n}},\text{Sim}_{SG}^{SG_{\text{Samp}}^{\RO_n}}\Big)}_t\rightarrow b'$.
 \end{enumerate}
\mor{Did we show Hybrid 4=Hybrid 5?}
\takashi{It's a bit confusing, but we are not arguing that Hybrid 4 and Hybrid 5 are indistinguishable. (In fact, even the form of the adversary is different in these hybrids.)
Instead, we are considering Hybrid 5-7 to argue that Hybrid 4 is indistinguishable from Hybrid 4' where $\RO'$ is always used instead of $\RO$ in the simulation of $Mix$. 
For avoiding such confusion, it may be better to use different names for Hybrid 5-7.
}
 
 \noindent{\textbf{Hybrid 6:}} This is the same as {\bf Hybrid 5}, but $\RO$ is modeled as a compressed oracle with database register $\ket{D}$. 
 \begin{claim}
   \begin{align}
   \Pr[1\gets {\bf Hybrid~6}(1^\secpa)]=\Pr[1\gets {\bf Hybrid~5}(1^\secpa)].
   \end{align}
   \mor{for all $n$?}
 \end{claim}
\begin{proof}
    This follows directly from Lemma \ref{lem:comporacle}.
\end{proof}
 
  \noindent{\textbf{Hybrid 7:}} This is the same as {\bf Hybrid 6}, but $(\mathcal{A}_s,\mathcal{A}_t)$,  has access to the oracle $(\text{Sim}_{{Mix}_n}^{Ver_{\text{Samp},n}},\text{Sim}_{{SG}_n}^{SG_{\text{Samp},n}})$. 
 \begin{claim}
     \begin{align}
     \max_{(\mathcal{A}_s,\mathcal{A}_t)}\abs{\Pr[1\gets {\bf Hybrid~6}(1^\secpa)]
     -\Pr[1\gets {\bf Hybrid~7}(1^\secpa)]}\leq O\left(\frac{s}{2^{n/2}}\right).
     \end{align}
     \mor{what is $s$?}
 \end{claim}
 \begin{proof}
 
 First, from Theorem \ref{thm:presampledatabase}, we can conclude that, after $\mathcal{A}_s$ postselects on the $\RO$ such that $\mathcal{A}_s^{\RO}\rightarrow 1$, the state on the database register $\rho_D$ will always satisfy $\Tr(\Pi_{\leq s} \rho_D) = 1$.\mor{"the register lies within the projector" does not make sense} Now, define the projection 
 $\Pi_{\{k_1, \dots, k_t\}\cap D= \emptyset}$ 
 onto databases which contain no elements that are in $\{k_1, \dots k_t\}$. 
 Since $k_1, \dots k_t \leftarrow \{0,1\}^n\setminus \{0\}$ is a set of random elements,  
 for all $D$ such that $|D|\leq s$\mor{what is $s$?},
 \begin{align}
     \Pr[D\cap \{k_1, \dots k_t\}=\emptyset]&= 1-  \Pr[\exists x \in D\cap \{k_1, \dots k_t\}]\\
     &\geq 1-\sum_{x\in \{k_1, \dots k_t\}}\Pr[x \in D]\\
     &\geq 1 -\frac{s^2}{2^n}.
 \end{align}
 Note that, if this projector is applied after initialization, then $\textbf{Hybrid 6}$ and $\textbf{Hybrid 7}$ are identically distributed, since in $\textbf{Hybrid 7}$ the fresh random oracle will only be queried on $(x,y)$ such that either $x\in\{k_1, \dots k_t\}$ or $y\in \{k_1, \dots k_t\}$. 
 And therefore the lemma follows by gentle measurement.
 \end{proof}
 %Now, let $\mathcal{O}_0$ be the oracle $(\text{Sim'}_{{Mix}_n}^{Ver_{\text{Samp},n}},\text{Sim}_{{SG}_n}^{SG_{\text{Samp},n}})$, and let $\mathcal{O}_1$ be the oracle $(\text{Sim}_{{Mix}_n}^{Ver_{\text{Samp},n}},\text{Sim}_{{SG}_n}^{SG_{\text{Samp},n}})$. Let $G_0$ and $G_1$ be the games in the post selected random oracle model (Definition \ref{def:pbfqrom}) with respect to $\mathcal{O}_0$ and $\mathcal{O}_1$ respectively.
  Then, combining $\textbf{Hybrid 5}$, $\textbf{Hybrid 6}$ and $\textbf{Hybrid 7}$ and applying triangle inequality implies that 
 \begin{align}
 \max_{\mathcal{A}_s,\mathcal{A}_t}\abs{\Pr[1\gets {\bf Hybrid~5}(1^\secpa)]-\Pr[1\gets {\bf Hybrid~7}(1^\secpa)]}\leq O(s\cdot 2^{-n/2}),
 \end{align}
 which is to say that the oracles 
 $(\text{Sim}_{{Mix}_n}^{Ver_{\text{Samp},n}},\text{Sim}_{{SG}_n}^{SG_{\text{Samp},n}}),$ $(\text{Sim'}_{{Mix}_n}^{Ver_{\text{Samp},n}},\text{Sim}_{{SG}_n}^{SG_{\text{Samp},n}})$ are $(s,t,O(s\cdot 2^{-n/2}))$ indistinguishable in the post selected random oracle model.   

 Setting $\gamma = 2^{-n}, \epsilon=O\left(\frac{S}{2^{n/4}}\right)$, from \Cref{cor:fixedpresampling}, this implies that, with probability $\geq 1-O\left(\frac{S}{2^{n/4}}\right) $ over $\RO_n$, for all $T$ query algorithms $\mathcal{A}$ and
 advice strings $c$ such that $|A|+|c|\leq S$, 
 \begin{align}\label{eq:sim_to_sim_prime}
\abs{\Pr[\mathcal{A}^{(\text{Sim}_{{Mix}_n}^{Ver_{\text{Samp},n}},\text{Sim}_{{SG}_n}^{SG_{\text{Samp},n}})}=1]-\Pr[\mathcal{A}^{(\text{Sim'}_{{Mix}_n}^{Ver_{\text{Samp},n}},\text{Sim}_{{SG}_n}^{SG_{\text{Samp},n}})}=1]}\leq O\left(\frac{S}{2^{n/4}}\right).
 \end{align}

Now, recall that we have already established that
with probability $\geq 1-2^S \exp(-2^{n/2})$ over $U_n, \RO_n$, 
\begin{align}
\max_{\mathcal{A},c}\abs{\Pr[1 \gets \mathcal{A}^{Mix, SG}(1^n,c)]- \Pr[\mathcal{A}^{(\text{Sim'}_{{Mix}_n}^{Ver_{\text{Samp},n}},\text{Sim}_{{SG}_n}^{SG_{\text{Samp},n}})}(1^n,c)=1]}\leq \epsilon + O\left(\frac{T^{2.5}}{2^{n/8}}\right), 
\end{align} 

which implies that with probability $\geq 1-2^S \exp(-2^{n/2})-O\left(\frac{S}{2^{n/4}}\right)\geq 1-O\left(\frac{S}{2^{n/4}}\right)$ (using the fact that $2^S\exp(-2^{n/2})\leq S\cdot 2^{-n/2}$) over $U_n, \RO_n$, for all $T$ query algorithms $\mathcal{A}$ and advice strings $c$
\begin{align}
\abs{\Pr[1 \gets \mathcal{A}^{Mix_n, SG_n}(1^n,c)]- \Pr[1\gets\mathcal{A}^{(\text{Sim}_{{Mix}_n}^{Ver_{\text{Samp},n}},\text{Sim}_{{SG}_n}^{SG_{\text{Samp},n}})}(1^n,c)]}\\
\leq \epsilon + O\left(\frac{T^{2.5}}{2^{n/8}}\right)+ O\left(\frac{S}{2^{n/4}}\right)\leq \epsilon + O\left(\frac{T^{2.5}+S}{2^{n/8}}\right), 
\end{align}

\end{proof}

%% file: money/unitarylightning.tex
\section{Quantum Lightning Exists Relative to $(SG,Mix,PSPACE)$.}\label{sec:lightning}

In this section, we show quantum lightning exists relative to the oracle $(SG,Mix,PSPACE)$.

\begin{theorem}\label{thm:lightning}
    With probability $1$ over $\{U_n\}_{n \in \N}$ and $\{\RO_n\}_{n \in \N}$, quantum lightning exists relative to $(SG,Mix,PSPACE)$.
\end{theorem}

\begin{proof}
    Our construction is as follows.
    \begin{enumerate}
        \item $\mathsf{Mint}(1^n)$: Query $|0\rangle$ to $SG_n$ to get 
        , where $|\phi_x\rangle\coloneqq U_n|x\rangle$, and measure the first register to
        get the pair $(k,\ket{\phi_k})$ of the measurement result $k$ and the post-measurement state $|\phi_k\rangle$. 
        Output the serial number $\sigma\coloneqq k$ and the bold state $\ket{\$}\coloneqq \ket{\phi_k}$.
        \item $\mathsf{Ver}(1^n,\sigma, \ket{\$})$: Query $\ket{\sigma}\ket{0}\ket{\$}\ket{0}\ket{0}$ to $Mix_n$, and measure the last register.
        Output the measurement result.
    \end{enumerate}

    \if0
    Note that if $SG_n,Mix_n$ are defined with respect to $U_n$ and $\ket{\phi_k}=U_n\ket{k}$, then this exactly implements
    \begin{enumerate}
        \item $\mathsf{Mint}(1^n)$: Sample $k\gets \{0,1\}^n$. Output $\ket{\phi_k}$.
        \item $\mathsf{Ver}(1^n,\sigma, \ket{\$})$: Apply the measurement $\{\ketbra{\phi_\sigma},I-\ketbra{\phi_\sigma}\}$. Output $1$ if and only if the measurement produces the first result.
    \end{enumerate}
    \fi

    Correctness follows by construction.
    Next, we show security. Let $\A$ be an attacker against the quantum lightning security of $(\mathsf{Mint},\mathsf{Ver})$. 
    In particular, we will take $\A^{SG_n,Mix_n,PSPACE}(1^n,c)$ to be a uniform algorithm with non-uniform advice $c$ and making $T$ queries to its oracles, where $|\A|+|c| \leq \poly(n)$. 
    We will first show that with all but negligible probability over $U_n,\RO_n$ (and for any fixed $\{U_\ell\}_{\ell\in \N\setminus \{n\}},\{\RO_\ell\}_{\ell \in \N\setminus \{n\}}$), for all $c$, 
    \begin{align}
    \Pr[\A(1^\secpa,c)\text{ wins the quantum lightning game}]\leq \negl(n). 
    \end{align}
    In particular, let $p(n)$ be any polynomial. 
    We will show that for all sufficiently large $n$, 
    \begin{align}
    \Pr[\A(1^\secpa,c)\text{ wins the quantum lightning game}]\leq 1/p(n).
    \end{align} \takashi{This only holds for all sufficiently large $n$. The same comment applies to many parts of the following proof (and the other sections as well). If it's bothering to state this every time, I'm fine with omitting it, but I feel we should at least mention it at some point.}
    To show it, we define the following sequence of games.

    \begin{enumerate}
        \item $\mathbf{G}_1$: It is the original security game of quantum lightning. 
        \begin{enumerate}
            \item Run $\A^{SG,Mix,PSPACE}(1^n,c)\to (x,\rho_{\mathbf{AB}})$.
            \item Apply the measurement $\{\ketbra{\phi_x,\phi_x}, I - \ketbra{\phi_x,\phi_x}\}$ to $\rho_{\mathbf{AB}}$. 
            Output $\top$ if measurement results in the first outcome.
            Otherwise, output $\bot$.
        \end{enumerate}
        \item $\mathbf{G}_2$: It is the same as $\mathbf{G}_1$ except that it replaces all calls to $SG_n$ and $Mix_n$ with queries to 
      $\text{Sim}_{{SG}_{n}}^{SG_{\text{Samp},n}}$ and $\text{Sim}_{{Mix}_n}^{Ver_{\text{Samp},n}}$ respectively, from~\Cref{cor:keylem} instantiated with error $\epsilon=\frac{1}{6p(n)}$. \takashi{$\epsilon=\frac{1}{6p(n)}$?} 
        \item $\mathbf{G}_3$: It is the same as $\mathbf{G}_2$ except that it outputs $\bot$ 
        if the initialization for $\SGSampn$ samples the same $k$ twice.
        \item $\mathbf{G}_4$: It is the same as $\mathbf{G}_3$ except that it samples each $\ket{\phi_k}$ independently at random instead of guaranteeing they are all orthogonal. In particular, $\mathbf{G}_4$ is the quantum lightning game where the oracles $SG_n,Mix_n$ are replaced by $\text{Sim}_{{SG}_{n}}^{SG'_{\text{Samp},n}},\text{Sim}_{{Mix}_n}^{Ver'_{\text{Samp},n}}$ with $\SGSampn'$ and $\VerSampn'$ defined as follows.
        \begin{enumerate}
            \item On initialization, sample $k_1\neq \dots \neq k_t$ uniformly at random from $\{0,1\}^n\setminus \{0\}$, and sample $\ket{\phi_{k_i}}\gets \HaarSt{2^{n}}$.
            \item On the $i$th query, $\SGSampn'$ will output $\ket{k_i}\ket{\phi_{k_i}}$.
            \item $\VerSampn'$ is the map such that, for all $i \in [t]$, 
            \begin{align}
            \ket{k_i}\ket{\phi_{k_i}}\ket{b}\mapsto \ket{k_i}\ket{\phi_{k_i}}\ket{\bar{b}}
            \end{align}
            and acts as the identity for all orthogonal states.
        \end{enumerate}
    \end{enumerate}

    Now we show each hybrids are close with each other.
    \begin{enumerate}
        \item By~\Cref{cor:keylem}, with probability $1-\negl(n)$ over the choice of oracle $SG_n,Mix_n$, for all $\A,c$,
        \begin{align}
        \abs{\Pr[1\gets\mathbf{G}_1(1^n)] - \Pr[1\gets\mathbf{G}_2(1^n)]}   \leq \frac{1}{6p(n)}.
        \end{align} 
        \item We know that $\abs{\Pr[1\gets\mathbf{G}_2(1^n)] - \Pr[1\gets\mathbf{G}_3(1^n)]} \leq \Pr[\bot\gets\mathbf{G}_3(1^n)]$. But this is exactly the probability that there is a repeat among $\leq T = \poly(n)$ samples of a random value from $\{0,1\}^n$. By the birthday bound, we get that with probability $1$ over the choice of $SG_n,Mix_n$, for all $\A,c$,
        \begin{align}
        \abs{\Pr[1\gets\mathbf{G}_2(1^n)] - \Pr[1\gets\mathbf{G}_3(1^n)]} \leq \frac{T(n)^2}{2^n} \leq \frac{1}{6p(n)}.
        \end{align}
        \item ~\Cref{cor:ortho} immediately gives with probability $1$ over the choice of $SG_n,Mix_n$, for all $\A,c$,
        \begin{align}
        \abs{\Pr[1\gets\mathbf{G}_3(1^n)] - \Pr[1\gets\mathbf{G}_4(1^n)]} \leq  \negl(n).
        \end{align}
    \end{enumerate}

    Therefore by triangle inequality, with probability $1-\negl(n)$ over the choice of $SG,Mix$, for all $\A,c$,
    \begin{align}\label{eq:diff_G1_and_G4}
        \abs{\Pr[1\gets\mathbf{G}_1(1^n)] - \Pr[1\gets\mathbf{G}_4(1^n)]} \leq \frac{1}{2p(n)}.
    \end{align}

So it is sufficient to establish that $\Pr[1\gets\mathbf{G}_4(1^n)] \leq \text{negl}(n)$. Let $\A^{SG,Mix,PSPACE}$ be any adversary such that $\Pr[\mathbf{G}_4\to 1] = \epsilon$. We will construct an inefficient adversary $\A'^{P_{\ket{\psi}}}$ cloning a single Haar-random state $\ket{\psi}$ with access to a verification oracle $P_{\ket{\psi}}$.

    $\A'$ behaves as follows.
    \begin{enumerate}
        \item On input $\ket{\psi}$.
        \item Sample $k_1\neq \dots\neq k_t \gets \{0,1\}^n\setminus \{0\}$.
        \item Sample $i^* \randfrom [t]$. Set $\ket{\phi_{k_{i^*}}} = \ket{\psi}$.
        \item For $i\neq i^*$, sample $\ket{\phi_{k_i}}$ from the Haar distribution.
        \item Define $\wt{\SGSampn},\wt{\VerSampn}$ as follows:
        \begin{enumerate}
            \item On the $i$th query, $\wt{\SGSampn}$ outputs $\ket{k_i}\ket{\phi_{k_i}}$.
            \item $\VerSampn'$ is the map such that, for all $i \in [t]$, 
            \begin{align}
            \ket{k_i}\ket{\phi_{k_i}}\ket{b}\mapsto \ket{k_i}\ket{\phi_{k_i}}\ket{\bar{b}}
            \end{align}
            and acts as the identity for other orthogonal states.
        \end{enumerate}
        \item Run $(y,\rho)\gets\A^{\text{Sim}_{SG_n}^{\wt{\SGSamp,n}},\text{Sim}_{Mix_n}^{\wt{\VerSamp,n}},PSPACE}(1^n,c)$.
        \item If $y = k_{i^*}$, output $\rho$.
    \end{enumerate}

    Note that the view of $\A$ is exactly the same view as in $\mathbf{G}_4$. 
    Therefore as long as $\A'$ correctly guesses the index $\A$ wins at, $\A'$ will also win. That is,
    \begin{align}\label{eq:Aprime_wins_and_G4}
    \Pr[\A'\text{ wins cloning game}] \geq \Pr[1\gets\mathbf{G}_4(1^n)]\cdot \Pr[k_{i^*}=y] \geq \frac{\epsilon}{t}.
    \end{align}

    By~\Cref{lem:comnocloning}, we know that any adversary making at most $2^{n/4}$ queries cannot win the cloning game for a single Haar random state with probability above $2^{-n/2}$. Since $\A'$ makes at most polynomially many queries, we have that $\Pr[\A'(1^\secpa)\text{ wins cloning game}] \leq 2^{-n/2}$. And so $\epsilon = \Pr[\mathbf{G}_4 \to 1] \leq t\cdot 2^{-n/2}=\negl(n)$.

    Thus, by \Cref{eq:Aprime_wins_and_G4},  $\Pr[1\gets\mathbf{G}_4(1^n)] \leq \negl(n)$, and therefore, by \Cref{eq:diff_G1_and_G4}  for all advice $c$, with probability $\geq 1 - \negl(n)$ over $SG_n,Mix_n$, for all $\A,c$, for all polynomials $p(n)$,
    \begin{align}
    \Pr[\A(1^n, c)\text{ wins the quantum lightning game}] \leq \frac{1}{p(n)}.
    \end{align}

    Therefore  with probability $\geq 1-\negl(n)$ over $SG_n,Mix_n$, for all $\A,c$,
    \begin{align}
    \Pr[\A(1^n, c)\text{ wins the quantum lightning game}] \leq \negl(n).
    \end{align}
    
    But note that since $\sum_{n=1}^\infty \negl(n)$ converges, by the Borel-Cantelli lemma~\cite{Borel,Cantelli} $\A$ achieves negligible advantage for all but finitely many input lengths $n\in \N$ with probability $1$ over 
    $\{U_{\ell}\}_{\ell \in \N},\{\RO_\ell\}_{\ell \in \N}$.
    In particular, $(\mathsf{Mint},\mathsf{Ver})$ is a quantum lightning scheme.
\end{proof}

\paragraph{Establishing $\mathbf{BQP}=\mathbf{QMA}$}
The main barrier to showing $\mathbf{BQP}=\mathbf{QMA}$ relative to our choice of oracles $Mix,SG,PSPACE$ is that $Mix$ relies on the random oracle. However, consider the following oracle $\{Ver_{\ell}\}_{\ell \in \mathbb{N}}$:
\begin{itemize}
    \item $Ver_{\ell}$: denotes the unitary which takes as input $(x,\ket{\psi})$, and performs the following mapping:
    \begin{align}
        \ket{k}\ket{\phi_k}\ket{b}\mapsto \ket{k}\ket{\phi_k}\ket{\ol{b}},
    \end{align}
        where $\ket{\phi_k}= U_{\ell}\ket{k}$
    and acts as identity on all orthogonal states. 
\end{itemize}

Now, it is easy to see that quantum lightning still exists relative to $SG,Ver$ even against inefficient adversaries (although the same doesn't apply for the other primitives we consider). We have the following construction:
\begin{enumerate}
        \item $\mathsf{Mint}(1^n)$: Query $|0\rangle$ to $SG_n$ to get 
        , where $|\phi_x\rangle\coloneqq U_n|x\rangle$, and measure the first register to
        get the pair $(k,\ket{\phi_k})$ of the measurement result $k$ and the post-measurement state $|\phi_k\rangle$. 
        Output the serial number $\sigma\coloneqq k$ and the bold state $\ket{\$}\coloneqq \ket{\phi_k}$.
        \item $\mathsf{Ver}(1^n,\sigma, \ket{\$})$: Query $\ket{\sigma}\ket{\$}\ket{0}$ to $Ver_n$, and measure the last register.
        Output the measurement result.
    \end{enumerate}

    The proof of security is essentially identical to the one relative to $SG, Mix$. 
    Unfortunately, it is not clear that $\mathbf{BQP}^{SG,Ver,PSPACE}=\mathbf{QMA}^{SG,Ver,PSPACE}$ directly. However, following the techniques of \cite{TQC:Kre21}, it is not difficult to define a recursive oracle $\mathcal{C}$ (essentially $PSPACE^{\SGSamp,\VerSamp}$) such that $\mathbf{BQP}^{SG,Ver,\mathcal{C}}=\mathbf{QMA}^{SG,Ver,\mathcal{C}}$.

%% file: unitaryQCCCKE.tex
\section{QCCC NIKE Exists Relative To $(SG,Mix,PSPACE)$}\label{app:ke}
\mor{In the following, we construct two-party QCCC NIKE. A generalization to constant multiparty cases can be done similarly.}
\mor{Do we mention about poly multiparty case?}
In the following, we construct two-party QCCC NIKE. Using a modified $Mix_n$ oracle which takes in an arbitrary number of inputs, it is not difficult to give an oracle relative to which multiparty NIKE exists but $\mathbf{BQP}=\mathbf{QCMA}$.

\begin{theorem}\label{thm:keexists}
    With probability $1$ over $\{\U_n\}_{n\in \N}$, $\{\RO_n\}_{n \in \N}$, relative to the oracle $(SG,Mix,PSPACE)$, QCCC NIKE exist.
\end{theorem}

Our NIKE protocol operates as follows.
\begin{enumerate}
    \item On input $1^\secpa$, Alice queries $|0\rangle$ to $SG_n$ to obtain $\frac{1}{\sqrt{2^n}}\sum_{x \in \{0,1\}^n}\ket{x}\ket{\phi_x}$, and measures the first register
    to obtain the measurement result $x$ and the post-measurement state $\ket{\phi_x}$. Alice sends $x$ to Bob.
    \item On input $1^\secpa$, Bob queries $|0\rangle$ to $SG_n$ to obtain $\frac{1}{\sqrt{2^n}}\sum_{x \in \{0,1\}^n}\ket{x}\ket{\phi_x}$, and measures the first register
    to obtain the measurement result $y$ and the post-measurement state $\ket{\phi_y}$. Bob sends $y$ to Alice.
    \item Alice queries $|x\rangle|y\rangle|\phi_x\rangle|0\rangle|0\rangle$ to $Mix_n$, and measures the fourth register to
    get $\RO_n(x,y)$. Alice outputs $\RO_n(x,y)$.
    \item Bob queries $|x\rangle|y\rangle|\phi_y\rangle|0\rangle|0\rangle$ to $Mix_n$, and measures the fourth register to
    obtain $\RO_n(x,y)$. Bob outputs $\RO_n(x,y)$.
\end{enumerate}

Agreement follows by construction.
It only remains to show security. Without loss of generality we assume the adversary is a $T$-query uniform quantum algorithm $\A^{SG_n,Mix_n,PSPACE}(1^n,c,x,y)$ where $(x,y)$ is the transcript and $c\in \{0,1\}^{\poly(n)}$ is the advice string.
 We will first show that with all but negligible probability over $U_n,\RO_n$ (and for any fixed $\{U_\ell\}_{\ell\in \N\setminus \{n\}},\{\RO_\ell\}_{\ell \in \N\setminus \{n\}}$), for all $c$,  and all polynomials $p$, 
\begin{align}
\Pr_{x,y}[\A(1^n,c,x,y)=\RO_n(x,y)]\leq \frac{1}{p(n)}.
\end{align}

We will use hybrids to switch our NIKE security game to one where we can more easily argue that the adversary has negligible advantage. 
\begin{enumerate}
    \item $\mathbf{G}_1$ is the original security game of the NIKE played with 
    $\A^{SG_n,Mix_n,PSPACE}$. More formally, it runs as follows:
    \begin{enumerate}
        \item Query $|0\rangle$ to $SG_n$ to get $\frac{1}{\sqrt{2^n}}\sum_x|x\rangle|\phi_x\rangle$, and measure the first register to obtain
        the pair $(x,\ket{\phi_x})$ of the measurement result $x$ and the post-measurement state $|\phi_x\rangle$.
        \item 
        Query $|0\rangle$ to $SG_n$ to get $\frac{1}{\sqrt{2^n}}\sum_k|k\rangle|\phi_k\rangle$, and measure the first register to obtain
        the pair $(y,\ket{\phi_y})$ of the measurement result $y$ and the post-measurement state $|\phi_y\rangle$.
        \item Run $\A^{SG_n,Mix_n,PSPACE}(1^\secpa,c,x,y)\to guess$.
        \item Output $1$ if and only if $guess = \RO_n(x,y)$.
    \end{enumerate}
    \item $\mathbf{G}_2$ is the same as $\mathbf{G}_1$, except that we will prevent $\A$ from querying $Mix_n$ on $\ket{\phi_x}$. In particular, every time $\A$ queries the oracle $Mix_n$ on $\rho_{\mathbf{ABCDE}}$, we will first apply the measurement $\{\ketbra{\phi_x},I-\ketbra{\phi_x}\}$ on register $\mathbf{C}$, and if the result is the first option the game will output $0$. Otherwise, we apply $Mix_n$ on the residual state.
    \item $\mathbf{G}_3$ is the same as $\mathbf{G}_2$, except that (analogously to the previous hybrid) we prevent $\A$ from querying $Mix_n$ on $\ket{\phi_y}$. Formally, define $\Pi_{bad} = I_{\mathbf{AB}} \otimes (\ketbra{\phi_x}_{\mathbf{C}} + \ketbra{\phi_y}_{\mathbf{C}}) \otimes I_{\mathbf{DE}}$. $\mathbf{G}_3$ will be the same as $\mathbf{G}_1$, but whenever $\A$ queries $Mix_n$ on $\rho_{\mathbf{ABCDE}}$, it will apply the measurement $\{\Pi_{bad},I-\Pi_{bad}\}$. If the result is the first option, the game will output $0$. Otherwise, we apply $Mix_n$ on the residual state. Note that this measurement can be exactly implemented by making a query to $Mix_n$ on $\ketbra{x}_{\mathbf{A}'}\otimes \ketbra{y}_{\mathbf{B}'}\otimes \rho_{\mathbf{C}}\otimes \ketbra{0}_{\mathbf{D}'\mathbf{E}'}$ measuring if register $\mathbf{E'}$ contains $1$, and then applying $Mix_n$ again to clear registers $\mathbf{D'E'}$.
    \item $\mathbf{G}_4$ is the same as $\mathbf{G}_3$, except that we replace  $SG_n,Mix_n$ with $\text{Sim}^{SG_{Samp,n}}_{SG_{n}},\text{Sim}^{Ver_{Samp,n}}_{Mix_{n}}$ with error $\frac{1}{2p(n)}$. In detail, this will be identical to the following game
    \begin{enumerate}
        \item Sample a fresh random unitary $\wt{U}_n$ and a random oracle $\wt{\RO}_n$ for $\SGSampn,\VerSampn$.
        \item Sample $x,y\gets \{0,1\}^n$.
        \item Define $\wt{\Pi}_{bad} = I_{\mathbf{AB}} \otimes (\wt{U}_n \ketbra{x} \wt{U}_n^\dagger + \wt{U}_n \ketbra{y} \wt{U}_n^\dagger)_{\mathbf{C}} \otimes I_{\mathbf{DE}}$.
        \item Run $\A^{\text{Sim}^{SG_{Samp,n}}_{SG_{n}},\text{Sim}^{Ver_{Samp,n}}_{Mix_{n}}}(1^\secpa,c,x,y)\to guess$, but every time $\A$ makes a $\text{Sim}^{Ver_{Samp,n}}_{Mix_{n}}$ query on a state $\rho_{\mathbf{ABCDE}}$, apply the measurement $\{\wt{\Pi}_{bad},I-\wt{\Pi}_{bad}\}$ on registers $\mathbf{ABC}$ and if the result is the first outcome, the game will output $0$.
        \item Output $1$ if and only if $guess = \wt{\RO}_n(x,y)$.
    \end{enumerate}
\end{enumerate}

\begin{lemma}\label{lem:kecloninghybrid}
    With probability $1$ over $SG_n,Mix_n$, for all $\A,c$
    \begin{align}
    \abs{\Pr[1\gets \mathbf{G}_1(1^\secpa)] - \Pr[1\gets\mathbf{G}_2(1^\secpa)]} \leq \negl(n).
    \end{align}
\end{lemma}

\begin{proof}
    We will show that for all $t$, if $\rho_{\mathbf{ABCDE}}$ is the state input to $Mix_n$ in the $t$-th query, then 
    ${\rm Tr}((I_{\mathbf{ABDE}} \otimes \ketbra{\phi_x}_{\mathbf{C}})\rho_{\mathbf{ABCDE}}) \leq \negl(n)$. The lemma then follows immediately from gentle measurement and induction.

    In fact, this claim follows immediately from the security our quantum lightning construction. In particular, let $\A(\cdot,c)$ be any adversary
    with
    advice $c$ such that the state $\rho_{\mathbf{ABCDE}}$ input to $Mix_n$ in the $t$-th query of $\mathbf{G}_1$ satisfies
    \begin{align}
    {\rm Tr}((I_{\mathbf{ABDE}} \otimes \ketbra{\phi_x}_{\mathbf{C}})\rho_{\mathbf{ABCDE}}) \geq \frac{1}{p(n)}
    \end{align}
    for some polynomial $p$. 
    Then we get an adversary 
    $\A'(1^\secpa,c)$
    for quantum lightning defined as follows. 
    \begin{enumerate}
        \item Run $\mathsf{Mint}(1^n)\to (x,\ket{\phi_x})$ and $\mathsf{Mint}(1^n)\to (y,\ket{\phi_y})$.
        \item Run $\A^{SG_n,Mix_n,PSPACE}(1^\secpa,c,x,y)$ up until the $t$-th query to $Mix_n$, generating a query state $\rho_{\mathbf{ABCDE}}$.
        \item Output $\ketbra{\phi_x}\otimes \rho_\mathbf{C}$.
    \end{enumerate}
    Note that the view of $\A$ generated by $\A'$ is identical to its view in $\mathbf{G}_1$. And therefore $\rho_\mathbf{C}$ measures to $\ket{\phi_x}$ with probability $\frac{1}{p(n)}$. 
    But so does $\ket{\phi_x}$, and therefore $\A'$ breaks the lightning protocol with probability $\frac{1}{p(n)}$.
\end{proof}

Analogously, we also have

\begin{lemma}
    With probability $1$ over $SG_n,Mix_n$, for all $\A,c$
    \begin{align}
    \abs{\Pr[1\gets\mathbf{G}_2(1^\secpa)] - \Pr[1\gets\mathbf{G}_3(1^\secpa)]} \leq \negl(n).
    \end{align}
\end{lemma}

\takashi{We have to bound the difference between $\mathbf{G}_3$ and $\mathbf{G}_4$.}\eli{Sorry, you caught me mid-rewriting this proof. The new hybrids are slightly different. In particular, I instead describe how to implement $\mathbf{G}_3$ using queries to $Mix_n$.}\takashi{Sure!}

\Cref{cor:keylem} immediately implies 

\begin{lemma}
    With probability $1-\negl(n)$ over $U_n,\RO_n$, for all $\A,c$
    \begin{align}
    \abs{\Pr[1\gets\mathbf{G}_4(1^\secpa)] - \Pr[1\gets\mathbf{G}_3(1^\secpa)]} \leq \frac{1}{2p(n)}.
    \end{align} 
\end{lemma}

Finally, it is clear that 
\begin{lemma}
    For all $SG_n,Mix_n$, for all $\A,c$
    \begin{align}
    \Pr[1\gets\mathbf{G}_4(1^\secpa)]\leq \frac{1}{2^n}.
    \end{align} 
\end{lemma}
In particular, $\A(1^\secpa,c,x,y)$ must guess $\RO_n(x,y)$, but its oracles are independent of $\RO_n(x,y)$. Thus, by lazy sampling, it is equivalent if $\RO_n(x,y)$ is sampled after $\A$ produces its output. The lemma follows.

Combining these lemmas, we get that with probability $1-\negl(n)$ over $U_n,\RO_n$, for all $\A,c$,
\begin{align}
\Pr_{x,y}[\A^{SG_n,Mix_n,PSPACE}(1^n,c,x,y)=\RO_n(x,y)] \leq \negl(n).
\end{align}
But note that since $\sum_{n=1}^\infty \negl(n)$ converges, by the Borel-Cantelli lemma~\cite{Borel,Cantelli} $\A$ achieves negligible advantage for all but finitely many input lengths $n\in \N$ with probability $1$ over 
$\{U_\ell\}_{\ell \in \mathbb{N}}$ and $\{\RO_\ell\}_{\ell\in \N}$.
In particular, with probability $1$, this protocol is a secure NIKE scheme relative to $(SG,Mix,PSPACE)$.

\paragraph{Generalizing to QCCC Multiparty NIKE }
We can generalize our 2 party protocol for non interactive key exchange to non interactive key exchange with polynomially many parties under a slightly modified (but still unitary) version of  $Mix= \{Mix_{\ell}\}_{\ell \in \mathbb{N}}$. 
\begin{itemize}
    \item Let $n$ be the security parameter, and let  the number of parties in the multiparty NIKE protocol be $p(n)$, where $p$ is some polynomial. The oracle $Mix^*= \{Mix^*_{\ell}\}_{\ell \in \mathbb{N}}$ is the following unitary:
   
        \begin{align}
    &\sum_{x_1, \dots x_{p(n)}\in\{0,1\}^{\ell}}|x_1\rangle\langle x_1|\otimes\dots \otimes|x_{p(n)}\rangle\langle x_{p(n)}|\otimes(|\phi_{x_1}\rangle\langle\phi_{x_1}|\dots +|\phi_{x_{p(n)}}\rangle\langle\phi_{x_{p(n)}}|)\\
    &\otimes X^{\RO(x_1, \dots x_{p(n)})}\otimes X    \\
    &+\sum_{x_1, \dots x_{p(n)}\in\{0,1\}^{\ell}}|x_1\rangle\langle x_1|\otimes\dots \otimes|x_{p(n)}\rangle\langle x_{p(n)}|\otimes(I-|\phi_{x_1}\rangle\langle\phi_{x_1}|-\dots -|\phi_{x_{p(n)}}\rangle\langle\phi_{x_{p(n)}}|)\\
    &\otimes I \otimes I  .
    \end{align}
    Here, $X$ is the Pauli $X$ operator and $X^r\coloneqq\bigotimes_{i=1}^\secpa X_i^{r_i}$ for any $\secpa$-bit string $r$.
    In other words, $Mix_{\ell}$ performs the following operation:
    \begin{align}
        \ket{x_1}\dots \ket{x_{p(n)}}\ket{\phi_{x_{1}}}\ket{z}\ket{b}&\rightarrow \ket{x_1}\dots \ket{x_{p(n)}}\ket{\phi_x}\ket{z\oplus \RO(x_1, \dots x_{p(n)})}\ket{\bar{b}}\\
        \vdots\\
        \ket{x_1}\dots \ket{x_{p(n)}}\ket{\phi_{x_{p(n)}}}\ket{z}\ket{b}&\rightarrow \ket{x_1}\dots \ket{x_{p(n)}}\ket{\phi_x}\ket{z\oplus \RO(x_1, \dots x_{p(n)})}\ket{\bar{b}}
    \end{align}
    and acts as identity for other orthogonal states. 
   
\end{itemize}
$SG$ is defined identically to before. \\
The 2 party NIKE protocol can be generalized to $p(n)$ parties $P_1, \dots P_{p(n)}$ trivially:
\begin{enumerate}
    \item On input $1^\secpa$, $\forall i \in [p(n)], P_i$ queries $|0\rangle$ to $SG_n$ to obtain $\frac{1}{\sqrt{2^n}}\sum_{x_i \in \{0,1\}^n}\ket{x_i}\ket{\phi_{x_i}}$, and measures the first register
    to obtain the measurement result $x_i$ and the post-measurement state $\ket{\phi_{x_i}}$. $P_i$ sends $x_i$ to $\{P_j\}_{j \neq i}$
   
    \item $\forall i \in [p(n)]$, $P_i$ queries $\ket{x_1}\dots \ket{x_{p(n)}}\ket{\phi_{x_{i}}}\ket{z}\ket{b} $ to $Mix^*_n$, and measures the fourth register to
    get $\RO_n(x_1,\dots x_{p(n)})$. $\forall i \in [p(n)], P_i$ outputs $\RO_n(x_1,\dots x_{p(n)})$.
   
\end{enumerate}

%% file: unitaryQCCCcommit.tex
\section{QCCC Commitments Exist Relative To $(SG,Mix,PSPACE)$}\label{sec:commitoracle}

Since this section details a construction, we will show that our construction satisfies stronger security requirements than given in~\Cref{def:com}. Note that if a bit commitment scheme satisfies statistical hiding and sum-binding, then it trivially satisfies computational hiding and weak computational binding.

\begin{definition}[Statistical hiding~\cite{STOC:HaiRei07}]Let $(\Comm,\Rece)$ be a bit-commitment scheme. We say that $(\Comm,\Rece)$ satisfies statistical hiding if for all (inefficient) adversarial receivers $\Rece'$, the transcript of the commit stage between $\Rece'$ and $\Comm$ with input $m$ is statistically close to the transcript of the commit stage between $\Rece'$ and $\Comm$ with input $m'$. That is,
\begin{align}
        \Delta(\Comm(m)\rightleftarrows \Rece', \Comm(m') \rightleftarrows \Rece') \leq \negl(\secpa).
\end{align}
Here $\Comm(m)\rightleftarrows\Rece'$ is the distribution of the transcript of the commit stage between $\Rece'$ and $\Comm$ with input $m$.
\end{definition}

\begin{definition}[Sum-Binding~\cite{EC:Unruh16}]
    Let $(\Comm,\Rece)$ be a bit-commitment scheme. We say that $(\Comm,\Rece)$ satisfies sum-binding if for all QPT adversarial senders 
    $\Comm'$, the probability that $\Comm'$ wins the following game is $\leq \frac{1}{2} + \negl(\secpa)$.
        \begin{enumerate}
            \item In the commit stage, an honest receiver $\Rece$ interacts with the adversarial committer $\Comm'$ to produce a transcript $z$ and the receiver state $\rho_{\Rece}$. $\Comm'$ produces the adversarial committer state $\rho_{\Comm'}$.
            \item The challenger samples a bit $b \randfrom \{0,1\}$.
            \item In the opening stage, the honest receiver $\Rece$ is given $\rho_{\Rece}$ and $z$, while $\Comm'$ is given $z,\rho_{\Comm'},$ and $b$. They then proceed to run the opening stage with the committer replaced by $\Comm'$, and $\Rece$ produces a final output $b'$. $\Comm'$ wins if $b' = b$.
        \end{enumerate}
\end{definition}

\begin{theorem}\label{thm:bitcommitment}
    With probability $1$ over $\{U_m\}_{m\in\mathbb{N}}$ and $\{\RO_m\}_{m\in\mathbb{N}}$, there exists a bit commitment scheme satisfying statistical hiding and sum-binding relative to $SG,Mix,PSPACE$.
\end{theorem}

\begin{proof}
We construct a commitment scheme $(\Com^{SG,Mix,PSPACE},\Rec^{SG,Mix,PSPACE})$ as follows. Note that our commitment scheme has a non-interactive commitment stage and an interactive opening stage.
\begin{enumerate}
    \item To commit to $b$, $\Com$ queries $|0\rangle$ to $SG_n$ to generate $SG_n\ket{0}$ and measures the first register to get $(x, \ket{\phi_x})$. Let $x_1$ be the first bit of $x$ and let $x_{>1}$ be the last $\secpa-1$ bits of $x$. $\Com$ repeats this process until $x_1 = b$. $\Com$ sends $x_{>1}$ to $\Rec$.
    \item To open to a bit $b$, $\Rec$ queries $|0\rangle$ to $SG_n$ to generate $SG_n\ket{0}$ and measures the first register 
    to get $ (y,\ket{\phi_y})$. $\Rec$ sends $y$ to $\Com$.
    \item Upon receiving $y$, $\Com$ computes $Mix_n\ket{x}\ket{y}\ket{\phi_x}\ket{0}\ket{0}$, measures the fourth register to get the measurement result $r$,
    and sends $r$ to $\Rec$.
    \item $\Rec$ computes $Mix_n\ket{y}\ket{b,x_{>1}} \ket{\phi_y}\ket{0}\ket{0}$ and measures the fourth register to get the result $r'$. $\Rec$ accepts if $r=r'$.
    Otherwise, $\Rec$ rejects.
\end{enumerate}

Correctness is clear by construction.
Note that the commitment $x_{>1}$ is a uniformly random string, and therefore $(\Com,\Rec)$ trivially satisfies statistical hiding.

It remains to show computational binding. 
We will define a sequence of hybrids $\mathbf{G}_1,\dots,\mathbf{G}_5$ where each hybrid is defined relative to some $t$-query oracle algorithm $\mathcal{A}^{SG,Mix,PSPACE}(1^n,c)$. 

\begin{enumerate}
    \item $\mathbf{G}_1$ will be the original sum-binding security game. We will rewrite it for simplicity as follows
    \begin{enumerate}
        \item Sample $b\randfrom \{0,1\}$.
        \item Query $|0\rangle$ to $SG_n$ to generate $SG_n\ket{0}$, and measure the first register to get $(y,\ket{\phi_y})$.
        \item Run $(x,\rho_\A)\gets\A^{SG,Mix,PSPACE}(1^\secpa,c)$.
        \item Run $z\gets\A^{SG,Mix,PSPACE}(b,x,y,\rho_{\A}) $.
        \item Generate $Mix_n \ket{y}\ket{b,x}\ket{\phi_y}\ket{0}\ket{0}$ and measure the fourth register. If the result is $z$,
        output 1. Otherwise, output 0.
    \end{enumerate}
    \item $\mathbf{G}_2$ is the same as $\mathbf{G}_1$, except that we will prevent $\A$ from querying $Mix_n$ on $\ket{\phi_y}$. In particular, every time $\A$ queries the oracle $Mix_n$ on $\rho_{\mathbf{ABCDE}}$, we will first apply the measurement $\{\ketbra{\phi_y},I-\ketbra{\phi_y}\}$ on register $\mathbf{C}$, and if the result is the first option the game will output $0$. Otherwise, we apply $Mix_n$ on the residual state. Note that this measurement can be performed by querying $Mix_n$ on $\ketbra{y}_{\mathbf{A}'}\otimes \ketbra{0}_{\mathbf{B}'}\otimes \rho_C \otimes \ketbra{0}_{\mathbf{C'D'}}$, measuring whether register $D'$ contains $\ket{1}$, and then applying $Mix_n$ again to clear out registers $C'$ and $D'$.
    \item $\mathbf{G}_3$ is the same as $\mathbf{G}_2$, except that we replace all calls to $SG_n$ and $Mix_n$ (including the ones used to evaluate $\{\ketbra{\phi_x},I-\ketbra{\phi_x}\}$) with $\SimSGn,\SimMixn$ for $\epsilon = \frac{1}{2p(n)}$.
    \item $\mathbf{G}_4$ is the same as $\mathbf{G}_3$, except that we resample all random oracle outputs $\RO(w, y)$ for all inputs $w$ for which $\A$ does not have access to $\ket{\phi_z}$. In particular, let $k_1,\dots,k_t$ be the inputs for $\SGSampn$ chosen during initialization. After $\A$ outputs the de-commitment $z$, for all $w \notin \{k_1,\dots,k_t\}$, we will replace $\RO'_n(w,y)$ with a fresh random value. 
    \item $\mathbf{G}_5$ will be the same as $\mathbf{G}_4$, except that we require for all $i\neq j \in [t]$, $k_i,k_j$ sampled during initialization must differ in the last $n-1$ bits. That is, $(k_i)_{>1} \neq (k_j)_{>1}$.
\end{enumerate}

\begin{claim}
    With probability $1$ over $SG,Mix$, and for all $\A,c$
    \begin{align}
    \abs{\Pr[1\gets\mathbf{G}_1(1^\secpa)] - \Pr[1\gets\mathbf{G}_2(1^\secpa)]} \leq \negl(n).
    \end{align}
\end{claim}
\begin{proof}
    This follows the exact same proof structure as~\Cref{lem:kecloninghybrid}. Intuitively, if $\A$ ever queries on $\ket{\phi_y}$, since the game itself has another copy of $\ket{\phi_y}$, this would violate security of our quantum lightning scheme.
\end{proof}

\begin{claim}
    For all sufficiently large $n$, with probability $\geq 1 - \negl(n)$ over $U_n,\RO_n$, for all $\A,c$
    \begin{align}
    \abs{\Pr[1\gets\mathbf{G}_3(1^n)] - \Pr[1\gets\mathbf{G}_2(1^n)]} \leq \frac{1}{2p(n)}.
    \end{align}
    \mor{for all sufficiently large $\secpa$?}
\end{claim}

This follows immediately from~\Cref{cor:keylem}

\begin{claim}
    For all $\{U_\ell\}_{\ell \in \N},\{\RO_\ell\}_{\ell\in \N}$, for all $n\in \N$, and for all $\A,c$,
    \begin{align}
    \Pr[1\gets\mathbf{G}_4(1^n)] = \Pr[1\gets\mathbf{G}_3(1^n)].
    \end{align}
    \mor{for all $n$?}
\end{claim}
\begin{proof}
    This follows immediately from the fact that $\A$'s oracle is independent of $\RO'_n(w,y)$ for $w\notin \{k_1, \dots k_t\}$. This is because by construction any query to $\SimMixn$ on input $\rho_{\mathbf{ABCDE}}$ is independent of $\RO(w,y)$ unless either $\rho_{\mathbf{C}}$ has some overlap with $\ketbra{\phi_y}$ or $w \in \{k_1,\dots,k_t\}$. But note that by construction neither game ever queries $\SimMixn$ on $\rho_{\mathbf{ABCDE}}$ where $\rho_{\mathbf{C}}$ has any overlap with $\ketbra{\phi_y}$.
\end{proof}

\begin{claim}
    With probability $1$ over $\{U_\ell\}_{\ell\in\mathbb{N}}$ and $\{\RO_\ell\}_{\ell\in \N}$, for all $\A,c$,
    \begin{align}
    \abs{\Pr[1\gets\mathbf{G}_5(1^n)] - \Pr[1\gets\mathbf{G}_4(1^n)]} \leq \negl(n).
    \end{align}
\end{claim} 

\begin{proof}
    The only way to distinguish these two games is if in $\mathbf{G}_4$ there is some $i\neq j\in [t]$ such that $(k_i)_{>1} = (k_j)_{>1}$. By collision-bound, this occurs with probability $\leq \frac{t^2}{2^{n-1}} \leq \negl(n)$.
\end{proof}

\begin{claim}
    $\Pr[1\gets\mathbf{G}_5(1^n)] \leq \frac{1}{2} + \frac{1}{2^n}$.
\end{claim}

\begin{proof}
    If $b.x$ (i.e. the concatenation of $b$ and $x$) is not in $\{k_1,\dots,k_t\}$, then $\RO_n'(b.x,y)$ will be resampled after the decommitment $z$ is chosen, and so $\A$ will succeed with probability $\frac{1}{2^n}$. However, at the point when $x$ is chosen, $k_1,\dots,k_t$ have all been chosen. If $x\neq (k_i)_{>1}$ for some $i$, then it is necessarily the case that $b.x \notin \{k_1,\dots,k_t\}$. If $x = (k_i)_{>1}$ for some $i$, then unless $b = (k_i)_1$, it will be the case that $b.x \notin \{k_1,\dots,k_t\}$. And so in particular, with probability $\geq \frac{1}{2}$, it will be the case that $b.x \notin \{k_1,\dots,k_t\}$. And so the adversaries success probability is bounded by $\frac{1}{2} + \frac{1}{2^n}$.
\end{proof}

Combining these lemmas, we get that with probability $1-\negl(n)$ over $SG_n,Mix_n$, and for all $\A,c$,
\begin{align}
\Pr[\A(1^n,c,\cdot)\text{ breaks sum-binding}] \leq \frac{1}{2}+\negl(n).
\end{align}

But note that since $\sum_{n=1}^\infty \negl(n)$ converges, by the Borel-Cantelli lemma~\cite{Borel,Cantelli},
$\A$ achieves negligible advantage for all but finitely many input lengths $n\in \N$ with probability $1$ over 
$\{U_\ell\}_{\ell \in \N}$ and $\{\RO_\ell\}_{\ell \in \N}$.

In particular, with probability $1$, this protocol is a computationally sum-binding commitment scheme relative to $(SG,Mix,PSPACE)$.
\end{proof}

%% file: unitaryBQPQCMA.tex
\section{$\mathbf{BQP}^{SG,Mix,PSPACE} = \mathbf{QCMA}^{SG,Mix,PSPACE}$}\label{sec:compequivunit}
In this section, we show the equivalence of $\mathbf{BQP}$ and $\mathbf{QCMA}$ relative to $(SG,Mix,PSPACE)$.
\begin{theorem}\label{thm:bqpqcma}
    With probability at least $1-\negl$ over $\{U_\ell\}_{\ell\in\N}$ and $\{\RO_\ell\}_{\ell\in\N}$,
    \begin{align}
    \mathbf{BQP}^{SG,Mix,PSPACE} = \mathbf{QCMA}^{SG,Mix,PSPACE}.
    \end{align}
\end{theorem}

To show it, we use the following lemma.
\begin{lemma}[Path Recording~\cite{MH24}]\label{thm:pathrecording}
    Let $n,t\in\N$. There exists an efficiently implementable isometry $PR$ operating on an input register $\mathbf{A}$ 
    and a state register $\mathbf{St}$ indistinguishable from a Haar random unitary by oracle access. Formally, let $\A^{(\cdot)}$ be any $t$-query oracle algorithm. Then
    \begin{align}
    \TD\left(\E_{U\gets \HaarUn{2^n}}\left[\A^{U,U^\dagger}\right],{\rm Tr}_{\mathbf{St}}\left(\A^{PR,PR^\dagger}\right)\right) \leq O\left(\frac{t^2}{2^{n/8}}\right).
    \end{align}
    \mor{I guess $\cA^{PR,PR^\dagger}$ is the output state of $\cA^{PR,PR^\dagger}$. Then it does not contain $\mathbf{St}$, so tracing over $\mathbf{St}$ does not make sense.}
\end{lemma}

\begin{proof}[Proof of \Cref{thm:bqpqcma}]
Let $\Pi\in \mathbf{QCMA}^{SG,Mix,PSPACE}$, which implies that there exists a $\mathbf{QCMA}$ verifier $\Ver^{SG,Mix,PSPACE}$ which takes as input 
$x \in \{0,1\}^*$
and a classical witness $w \in \{0,1\}^{q(\abs{x})}$ where $q$ is a polynomial, with completeness $\frac{99}{100}$ and soundness $\frac{1}{100}$. Without loss of generality, we will assume $q(\abs{x}) \geq \abs{x}$.
Let $|x|=n$, and let $T(n)$ be a polynomial upper bound on the running time of $\Ver$ on inputs of length $n$. 

We will now describe a QPT algorithm $\mathcal{B}^{SG,Mix,PSPACE}(x)$ such that, with probability $1-\negl$ over $\{U_\ell\}_{\ell\in\mathbb{N}}$ and $\{\RO_\ell\}_{\ell\in\mathbb{N}}$, $\cal{B}$ decides $\Pi$ on all but finitely many inputs $x$. In particular, we will show that for all but finitely many $x \in \Pi_{yes}$, $\Pr[\mathcal{B}^{SG,Mix,PSPACE}(x)=1]\geq \frac{2}{3}$ and for all but finitely many $x \in \Pi_{no}$, $\Pr[\mathcal{B}^{SG,Mix,PSPACE}(x)=1]\leq \frac{1}{3}$.\mor{What are $\Pi_{yes}$ and $\Pi_{no}$. They were not defined. Are you considering promise problem, not language?}\eli{They are defined in the preliminaries.}  Let $\delta = \frac{1}{600}$ and let $d=\log (c_d q(|x|)^2T(|x|)^4+2)$ for some sufficiently large constant $c_d$. 
\begin{itemize}
    \item For all $\ell\in [d]$,  $\mathcal{B}$ performs tomography on each $U_{\ell}$, producing estimates $\U_{\ell}'$. 
    \takashi{I think putting this sentence here is confusing; In this item we are just explaining how to produce $U'_\ell$, and not explaining how to simulate $U_\ell$. (In fact, at this point, it's unclear for whom $\mathcal{B}$ is simulating the oracle.)} The estimates are produced as follows:
    \begin{itemize}
     \item $\mathcal{B}$ runs $SG_\ell$ on $\ket{0}$ and measures the first register of 
     $\frac{1}{\sqrt{2^\ell}}\sum_{x\in \{0,1\}^\ell \setminus \{0\}}\ket{x}\ket{\phi_x}$.
     $\mathcal{B}$ repeats this until, $\forall k \in [2^\ell]$, there exist $t>2^{d}\cdot \frac{T}{\delta}$ copies of $\ket{\phi_k}$.
    \item  For every $k \in [2^{\ell}]$, $\cB$ runs the state tomography algorithm from Theorem \ref{thm:statetom} on inputs $n=q(n)^2\cdot T(n)$ and $\epsilon = \frac{\delta^2}{2^{2\ell} \cdot T(n)^2}$. This produces an estimate $\ket{\phi'_k}$ such that $|\langle \phi_k|\phi'_k\rangle|^2 \geq 1-\epsilon$ with probability at least $1-e^{-5q(n)^2T(n)}$.
    \item $\mathcal{B}$ defines the map $U'_{\ell}\ket{k}=\ket{\phi'_k}$ for all $k\in [2^\ell]$ \mor{for all $k\in[2^\ell]$}.
 \end{itemize}
    \item Further, for all $\ell\in [d]$, $\mathcal{B}$ learns all functions $\RO_{\ell}$ by querying $SG_\ell$, and produces estimates $\RO_{\ell}'$.  \takashi{Same as above.}
    The estimates are produced as follows:
    \begin{itemize}
        \item $\mathcal{B}$ runs $SG_\ell$ on $\ket{0}$ and measures the first register. $\mathcal{B}$ repeats it until, $\forall k \in [2^\ell]$, there exist $2^d$ copies of $\ket{\phi_k}$. With probability 
        $\geq 1-e^{-q(n)^2}$, 
        this terminates in time polynomial in $n$ for all $\ell$.
        \item For all $x,y\in \{0,1\}^{\ell}$, 
        $\mathcal{B}$ queries $|x\rangle|y\rangle|\phi_x\rangle|0\rangle|0\rangle$ to $Mix_\ell$ and measures the fourth register
        to get $\RO_{\ell}(x,y)$. $\mathcal{B}$ stores the result as $\RO_\ell'(x,y)$.
        \end{itemize}

    \item Next, $\mathcal{B}$ constructs the following $\mathbf{QCMA}^{PSPACE}$ verifier $\wt{\Ver}^{PSPACE}$ which will simulate $\Ver^{SG,Mix,PSPACE}$ by replacing queries to $\{SG_\ell\}_{\ell \in \N}$ and $\{\RO_\ell\}_{\ell \in \N}$ as follows:
    \begin{itemize}
        \item For all $\ell\in [d]$, every query $SG_\ell$ and $Mix_\ell$ make to $U_\ell,\RO_\ell$ are replaced by queries to $U_\ell'$ and $\RO_\ell'$. 
        \item For all $\ell\in [d+1, T(n)]$ simulate all queries to $SG_\ell,Mix_\ell$ with queries to $\SimSGl,\SimMixl$
         with error parameter $\epsilon = \frac{\delta^2}{T(n)}$ with the following modification. Instead of sampling fresh random $\wt{U}_\ell$ and $\wt{\RO}_\ell$ during initialization, $\wt{\Ver}$ will instead simulate these for $\SGSamp$ and $\VerSamp$ using path-recording oracles and compressed oracles respectively. 
    \end{itemize} 
    \item Finally, consider the $\mathbf{QCMA}^{PSPACE}$ problem $\wt{\Pi}$ corresponding to $\wt{\Ver}$ with completeness $\frac{2}{3}$ and soundness $\frac{1}{3}$. Since $\mathbf{QCMA}^A \subseteq \mathbf{PSPACE}^A$ for all classical oracles $A$, and since $\wt{\Pi}\in \mathbf{QCMA}^{PSPACE}\subseteq \mathbf{PSPACE}^{PSPACE}=\mathbf{PSPACE}$, $\mathcal{B}$ will then use the $PSPACE$ oracle to decide $\wt{\Pi}$, and will output $1$ if $x \in \wt{\Pi}_{yes}$ and $0$ if $x\in \wt{\Pi}_{no}$.
\end{itemize}

We will show via a sequence of hybrids that with probability $1-n^3$ over $\{U_\ell\}_{\ell\in\N}$ and $\{\RO_\ell\}_{\ell\in\N}$, for all $(x,w)$,
\begin{align}
\abs{\Pr[\wt{\Ver}^{PSPACE}(x,w) \to 1] - \Pr[\Ver^{SG,Mix,PSPACE}(x,w) \to 1]} \leq \frac{1}{100}.
\end{align}
For a verification algorithm $V_i$, Define $acc(V_i)\coloneqq \Pr[V_i(w,x) \to 1]$, and let $d=\log (3456|x|T(|x|)^4+2)$.
\begin{enumerate}
    \item Let $V_1 = \Ver^{SG,Mix,PSPACE}$.
    
    \item $V_2:$ This is the same as $V_1$ except, for all $\ell \in [d]$, all queries to $\RO_{\ell}$ are replaced with queries to $\RO_{\ell}'$ generated by $\mathcal{B}$. 
    \begin{claim}
        For all $SG,Mix$, for all $\Ver$, with probability $\geq 1-e^{-q(n)^2}$ over the randomness of $\mathcal{B}$, $acc(V_2)=acc(V_1)$.
    \end{claim}
    \begin{proof}
        For all $\ell \in [d]$, $\mathcal{B}$  learns all functions $\RO_{\ell}$ with no error with probability at least $1-e^{-q(n)^2}$, so $acc(V_2)=acc(V_1)$.
    \end{proof}

    \item $V_3$: This is the same as $V_2$ except, for all $\ell \in [d]$, all queries to $\U_\ell$ are replaced with queries to $\U_{\ell}'$ generated by $\mathcal{B}$. 
    \begin{claim}
        For all $SG,Mix$, for all $\Ver$, with probability $\geq 1-e^{-q(n)^2}$ over the randomness of $\mathcal{B}$, $\abs{acc(V_3)-acc(V_2)} \leq 4\delta$.
    \end{claim}
    \begin{proof}
        From~\Cref{thm:statetom}, we have that for each $\ell \in [d]$, for each $k\in \{0,1\}^\ell$, with probability $\geq 1 - e^{-5q(n)T(n)}$, $\abs{\braket{\phi_x|\phi_x'}}^2 \geq 1-\epsilon$, with $\epsilon=\frac{\delta^2}{16\cdot 2^{2\ell}\cdot T^2}$. 
        Taking the union bound over all $d,k$, we get that with probability $\geq 1 - d\cdot 2^d \cdot e^{-5q(n)T(n)} \geq 1-e^{-q(n)}$, for all $\ell \in [d],k\in \{0,1\}^d$, $\abs{\langle\phi_x|\phi_x'\rangle}^2 \geq 1-\epsilon$.

        Let $SG_\ell'$ and $Mix_\ell'$ be $SG_\ell$ and $Mix_\ell$ but with $U_\ell$ replaced by $U'_\ell$. In particular, 
        \begin{align}
                &\abs{\left(\frac{1}{\sqrt{2^\ell-1}}\sum_{k\in \{0,1\}^\ell\setminus \{0\}}\bra{k}\bra{\phi_k'}\right)\left(\frac{1}{\sqrt{2^\ell-1}}\sum_{k\in \{0,1\}^\ell\setminus \{0\}}\ket{k}\ket{\phi_k}\right)}^2\\
                &=\frac{1}{2^\ell - 1}\abs{\sum_{k \in \{0,1\}^\ell \setminus \{0\}} \braket{\phi_k'|\phi_k}}^2\\
                &\geq \frac{2^\ell - 1}{2^\ell - 1}\left(1-\epsilon\right),
        \end{align}
        and therefore by~\Cref{lem:tlem},
        \begin{align}
        \norm{SG_\ell - SG_\ell'}_{\diamond} \leq 2\sqrt{1-\left(1-\frac{\delta}{T(n)}\right)^2}\leq 4\sqrt{\epsilon}.
        \end{align}

        Similarly, by~\Cref{lem:diamond} and the fact that $Mix$ can be seen as applying the unitary version of the measurement $\{\ketbra{\phi_k},I-\ketbra{\phi_k}\}$ in succession for all $k$, we get
        \begin{align}
                \norm{Mix_\ell-Mix_\ell'}_{\diamond} \leq 2\cdot 2^\ell \cdot \sqrt{\epsilon}.
        \end{align}
        \mor{Did we define $Mix_\ell'$?}

        Thus, applying~\Cref{fact}\mor{link is broken} $T$ times and plugging in $\epsilon=\frac{\delta^2}{16\cdot 2^{2\ell}\cdot T^2}$ proves the claim.
    \end{proof}
   \item $V_4:$ This is the same as $V_3$ except, for all $\ell \in [d+1,T(n)]$, queries to $SG_\ell,Mix_\ell$ are replaced with queries to $\SimSGl,\SimMixl$ with error parameter $\frac{\delta^2}{T(n)}$. \takashi{ with error parameter $\epsilon = \frac{\delta^2}{T(n)}$?}    
   \begin{claim}
        With probability $\geq 1-n^{-4}$ over the choice of $SG,Mix$, for all $\Ver$, $\abs{acc(V_4)-acc(V_3)} \leq \delta^2$.
    \end{claim}
    \begin{proof}
        This follows immediately from applying~\Cref{keylem} for each $\ell \in [d+1,T(n)]$. In particular, for each $\ell$, \takashi{Below, shouldn't we replace $n$ with $\ell$?} replacing $SG_\ell,Mix_\ell$ incurs a loss of $\frac{\delta^2}{T(n)} + O\left(\frac{T(n)^{2.5} +q(n)}{2^{d/8}}\right)$  with probability $\geq 1- O(q(n)\cdot 2^{-d/4})$ \takashi{$2^{-d/4}$ instead of $2^{-d/2}$?}  over $\{U_\ell\}_{\ell\in\N},\{\RO_\ell\}_{\ell\in\N}$, and we repeat this $T$ times. 
        For sufficiently large $n,c_d$, $O\left(\frac{T(n)^{2.5} +q(n)}{2^{d/8}}\right) \leq \frac{\delta^2}{2}$ and $T(n)\cdot O(q(n)\cdot 2^{-d/4}) \leq n^{-4}$. \takashi{$2^{-d/4}$ instead of $2^{-d/2}$?}  Note that~\Cref{keylem} relativizes, and so we can apply it even for algorithms which query $SG_{\ell'},Mix_{\ell'}$ for $\ell'\neq \ell$.
    \end{proof}
   \item $V_5:$ This is the same as $V_4$ except for each $\ell \in [d+1,T]$, we replace the unitary sampled by $\SimSGl$
   with the path-recording oracle $PR$ from~\Cref{thm:pathrecording}.
   \begin{claim}
       For all $SG,Mix,\Ver$, $\abs{acc(V_5)-acc(V_4)} \leq \delta$.
   \end{claim}
   \begin{proof}
       This follows directly from~\Cref{thm:pathrecording}. In particular, each replacement incurs a loss of $O\left(\frac{T^2}{2^{d/8}}\right) \leq \frac{\delta}{T}$ for all sufficiently large $n$. And so, the claim follows by repeating for each $\ell\in [d+1,T]$.
   \end{proof}
   \item $V_6$: This is the same as $V_5$, except that for each $\ell \in [d+1,T]$, we replace the random oracle $\wt{\RO}_\ell$ sampled by $\SimSGl$\mor{subscript $\ell$} with a compressed oracle $\mathsf{Ctso}$.   
   \begin{claim}
       For all $SG,Mix,\Ver$, $acc(V_6)=acc(V_5)$.
   \end{claim}
   \begin{proof}
       This follows from Lemma \ref{lem:comporacle}.
   \end{proof}
    \end{enumerate}

    Hence, with probability $1-n^{-3}$ over the choice of $SG,Mix$ we have that for all verifiers $\Ver$, with probability $\geq 4\delta$ over the randomness of $\mathcal{B}$, for all inputs $x\in \{0,1\}^n,w\in \{0,1\}^{poly(n)}$,
    \begin{align}
\abs{\Pr[\wt{\Ver}^{PSPACE}(x,w) \to 1] - \Pr[\Ver^{SG,Mix,PSPACE}(x,w) \to 1]} \leq 6\delta \leq \frac{1}{100}.
\end{align}

    This means that for all $n\in \N$, with probability $1-n^{-3}$ over $SG,Mix$, the following holds
    \begin{enumerate}
        \item For all $x\in \{0,1\}^n \cap \Pi_{yes}$, $\Pr[\mathcal{B}^{SG,Mix,PSPACE}(x) \to 1] \geq 1 - 4\delta$
        \item For all $x \in \{0,1\}^n \cap \Pi_{no}$, $\Pr[\mathcal{B}^{SG,Mix,PSPACE}(x) \to 1] \leq 4\delta$
    \end{enumerate}
    Since $\sum_{i=1}^\infty i\cdot i^{-3} < \infty$, by the Borel-Cantelli lemma~\cite{Borel,Cantelli}, we can swap the quantifiers to say that with probability $1$ over $SG,Mix$, for all but finitely many $n\in \N$,
    \begin{enumerate}
        \item For all $x\in \{0,1\}^n \cap \Pi_{yes}$, $\Pr[\mathcal{B}^{SG,Mix,PSPACE}(x) \to 1] \geq 1 - 4\delta$
        \item For all $x \in \{0,1\}^n \cap \Pi_{no}$, $\Pr[\mathcal{B}^{SG,Mix,PSPACE}(x) \to 1] \leq 4\delta$
    \end{enumerate}
    
    Thus, with probability $1$, there exists some $\mathcal{B}'$ which correctly decides $\Pi$ for all $x$ by hard-coding the points where $\mathcal{B}$ was wrong.

    Therefore, for any fixed $\Pi \in \mathbf{QCMA}^{SG,Mix,PSPACE}$, with probability $1$ over $SG,Mix$, $\Pi \in  \mathbf{BQP}^{SG,Mix,PSPACE}$.

    Taking a union bound over all $\Pi \in \mathbf{QCMA}^{SG,Mix,PSPACE}$, we get that
    \begin{align}
    \mathbf{QCMA}^{SG,Mix,PSPACE} =  \mathbf{BQP}^{SG,Mix,PSPACE}.
    \end{align}
\end{proof}

%% file: sgmix/ketoowpuzz.tex
\section{OWPuzz Constructions}\label{sec:owpz}
\subsection{OWPuzzs From QCCC KE}
\begin{theorem}\label{thm:ketoowpuzz}
    If there exists QCCC KE, then distributional OWPuzzs exist.
\end{theorem}

\begin{proof}
    Let $(\mathsf{A},\mathsf{B})$ be a QCCC KE protocol. We construct a distributional OWPuzz, $\Samp$, as follows:
    \begin{enumerate}
        \item Run the KE protocol $\mathsf{A}\leftrightarrows \mathsf{B}$, producing a transcript $\tau$ as well as outputs $a,b$.
        \item Output $s\coloneqq \tau, k\coloneqq a$.
    \end{enumerate}

    Define $\mathsf{A}_{\tau}$ to be the distribution of $\mathsf{A}$'s output $a$ conditioned on the transcript being $\tau$. 
    Observe that $\mathsf{A}_{\tau}$ is exactly the distribution over $k$ conditioned on $s = \tau$.
    Let $a_{\tau} = \argmax_a \Pr[\mathsf{A}_{\tau} \to a]$. We can see that with probability $\geq \frac{99}{100}$ over $\tau$, $\Pr[\mathsf{A}_{\tau} \to a_{\tau}] \geq \frac{99}{100}$. Otherwise, $\mathsf{A}$ and $\mathsf{B}$ would disagree on the key with probability $\geq \frac{1}{10000}$, which violates correctness.
    In particular, this means that
    \begin{align}\label{eq:ketoowp1}
        \BigPr{a=a_{\tau}}{\mathsf{A}\leftrightarrows \mathsf{B} \to (\tau,a,b)} \geq \frac{98}{100}.
    \end{align}

    Let us assume towards contradiction that $\Samp$ is not a distributional OWPuzz. Then there exists a QPT adversary $\A$ such that
    \begin{equation}\label{eq:ketoowp2}
        \Delta((\tau,\mathcal{E}(\tau)),(\tau,a)) \leq \frac{1}{100}.
    \end{equation}
    \Cref{eq:ketoowp1,eq:ketoowp2} give us
    \begin{equation}\label{eq:ketoowp3}
        \Pr[a_\tau\gets\mathcal{E}(\tau)] \geq \frac{97}{100}
    \end{equation}
    since otherwise we could inefficiently distinguish $(\tau,\mathcal{E}(\tau))$ from $(\tau,a)$.
    Applying the union bound to~\Cref{eq:ketoowp2,eq:ketoowp3} then gives us
    \begin{equation}
        \Pr[\mathcal{E}(\tau) = a: \mathsf{A}\leftrightarrows \mathsf{B} \to (\tau,a,b)] \geq \frac{95}{100}
    \end{equation}
    and therefore $\mathcal{E}$ is also an attacker for the KE protocol.
    Since the KE protocol $(\mathsf{A},\mathsf{B})$ is secure, so is the distributional OWPuzz $\Samp$.
\end{proof}

%% file: sgmix/commitments.tex
\subsection{OWPuzzs From QCCC Commitments}

\begin{theorem}\label{thm:commtoowpuzz}
    If QCCC commitment schemes exist, then distributional OWPuzzs exist.
\end{theorem}

\begin{proof}
Let $(\Com,\Rec)$ be any QCCC commitment scheme with an opening stage taking $t$ rounds. We will assume that in each round, the receiver sends one message $x_i$ and receives a response $d_i$.
    
    We will construct a distributional OWPuzz, $\Samp$, as follows:
    \begin{enumerate}
        \item Sample $i\randfrom [t]$.
        \item Sample $m \randfrom \{0,1\}$.
        \item Run the committing stage on $m$, producing a transcript $z$.
        \item Run the first $i$ rounds of the opening stage, producing a transcript $x_1,d_1,\dots,x_i,d_i$.
        \item Output $s \coloneqq (i,z,x_1,d_1,\dots,x_i)$ and $k\coloneqq d_i$.
    \end{enumerate}

    Assume that $\Samp$ is not a distributional OWPuzzs. Then, for all constants $c\ge1$, there exists an adversary $\A$ such that for infinitely many $\secpa\in\mathbb{N}$,
    \begin{align}
    \Delta((i,z,x_1,d_1,\dots,x_i,d_i),(i,z,x_1,d_1,\dots,x_i,\A(i,z,x_1,d_1,\dots,x_i))) \leq \secpa^{-c}.
    \end{align}
However,
    \begin{align}
      &  \Delta((i,z,x_1,d_1,\dots,x_i,d_i),(i,z,x_1,d_1,\dots,x_i,\A(i,z,x_1,d_1,\dots,x_i)))\\
       & = \frac{1}{t}\sum_{i\in[t]}\Delta((z,x_1,d_1,\dots,x_i,d_i),(z,x_1,d_1,\dots,x_i,\A(i,z,x_1,d_1,\dots,x_i)))\\
       & \geq \frac{1}{t}\Delta((z,x_1,d_1,\dots,x_i,d_i),(z,x_1,d_1,\dots,x_i,\A(i,z,x_1,d_1,\dots,x_i)))
    \end{align}
       for each $i\in[t]$.
    Therefore, for each $i\in [t]$,
    \begin{equation}\label{eq:distbound}
        \Delta((z,x_1,d_1,\dots,x_i,d_i),(z,x_1,d_1,\dots,x_i,\A(i,z,x_1,d_1,\dots,x_i))) \leq t\cdot \secpa^{-c}.
    \end{equation}

    We can use $\A$ to build an adversary $\mathcal{B}$ that breaks hiding as follows.
    \begin{enumerate}
        \item Run the committing stage, producing transcript $z$ and receiver state $\wt{\rho}_{\Rece}^0$.
        \item For each $i \in [t]$
        \begin{enumerate}
            \item Run $
             \wt{x}_i
             \gets
            \Rec(i, z, \wt{x}_1, \wt{d}_1, \dots, \wt{x}_{i-1}, \wt{d}_{i-1}, \wt{\rho}_{\Rece}^{i-1}) 
            $ to produce the receiver's message in the $i$th round. Note that this will update $\wt{\rho}_{\Rece}^{i-1}$ to a new state, which we will denote $\wt{\rho}_{\Rece}^i$.
            \item Run $
             \wt{d}_i
             \gets
            \A(i, z, \wt{x}_1,\wt{d}_1,\dots,\wt{x}_i) 
            $ to simulate the committer's message in the $i$th round.
        \end{enumerate}
        \item 
        Produce the receiver's output $m'\gets\Rec(\wt{\rho}_{\Rece}^t)$,
        and output $m'$.
    \end{enumerate}

Now we show that such $m'$ is equal to the originally committed message $m$ with high probability, and therefore $\cB$ breaks hiding.
Intuitively, if $m'$ is another message than $m$, a malicious committer that runs
$\cA$ can open $m'\neq m$, which breaks the binding.
Moreover, we can also show that the probability that $\cB$ outputs $\bot$ is bounded.

More precisely, the argument is the following.
    Note that an adversarial committer could also simulate its messages in this way. That is, define $\Com'$ as follows
    \begin{enumerate}
        \item Let $z$ be the transcript of the committing stage, and let $x_1,d_1,\dots,x_i$ be the current transcript of the opening stage.
        \item Run $d_i\gets\A(i,z,x_1,d_1,\dots,x_i)$.
        \item Output $d_i$.
    \end{enumerate}
    Observe that the resulting distribution on the output of $\Rec$ when interacting with $\Com'$ on a message $m$ is identical to the output of our attacker when run on a transcript $z$ generated from $\Com(m) \rightleftarrows \Rec$. Thus, by binding we have that 
    \begin{align}\label{eq:mprimenotmbot}
        \Pr[\mathcal{B}(z,\rho_{\Rece}) \to m' \notin \{m,\bot\} :m\randfrom \{0,1\},\Com(m) \rightleftarrows \Rec \to (z, \rho_{\Rece})] \leq \negl(\secpa).
    \end{align}
    
    Consider the mixed state $\mathcal{I}_i = (z,\wt{x}_1,\wt{d}_1,\dots,\wt{x}_i,\wt{d}_i,\wt{\rho}_{\Rece}^i)$ produced by $\mathcal{B}$ after round $i$. We will also define the mixed state $\mathcal{J}_i = (z,x_1,d_1,\dots,x_i,d_i,\rho_{\Rece}^i)$ produced by an honest interaction between $\Com$ and $\Rec$ after $i$ rounds of the opening stage.

    We will show by induction that $\norm{\mathcal{I}_i - \mathcal{J}_i}_1 \leq i t \secpa^{-c}$. 
    The base case is trivial, since $\mathcal{I}_0 = \mathcal{J}_0$. 
    For the inductive step, assume that $\norm{\mathcal{I}_i - \mathcal{J}_i}_1 \leq i t \secpa^{-c}$. Define a new distribution $\mathcal{I}_i'$ obtained by applying the $i$th iteration of $\mathcal{B}$ to a sample from $\mathcal{J}_{i-1}$. In particular, $\mathcal{I}_i' = (z,x_1,d_1,\dots,d_{i-1},x_i',d_i',\rho_{\Rece}^{i})$ where $(z,x_1,d_1,\dots,d_{i-1})\randfrom \mathcal{J}_{i-1}$, $x_i' = \Rec(\rho_{\Rec}^{i-1})$ (producing $\rho_{\Rec}^{i}$), and $d_i' = \A(i,z,x_1,\dots,d_{i-1},x_i')$. Since trace distance is contractive under CPTP maps, we have 
    \begin{equation}\label{eq:iiprime}
      \norm{\mathcal{I}_i - \mathcal{I}_i'}_1 \leq \norm{\mathcal{I}_{i-1} - \mathcal{J}_{i-1}}_1 \leq (i-1)t\secpa^{-c}.  
    \end{equation}
    But note that the internal state of the receiver at any point in time is a deterministic function of the transcript. In particular, $\rho_{\Rec}^i$ is the state achieved by running the protocol in superposition and then postselecting on the transcript being $(z,x_1,d_1,\dots,x_i,d_i)$. Furthermore, $x_i'$ and $x_i$ are sampled identically. Thus,
    \begin{align}\label{eq:iprimei}
            &\norm{\mathcal{I}_i' - \mathcal{J}_i}_1\\
            &=\norm{(z,x_1,d_1,\dots,x_i,\A(i,z,x_1,d_1,\dots,x_i),\rho_{\Rec}^i) - (z,x_1,d_1,\dots,x_i,d_i,\rho_{\Rec}^i)}_1\\
            &=\norm{(z,x_1,d_1,\dots,x_i,\A(i,z,x_1,d_1,\dots,x_i)) - (z,x_1,d_1,\dots,x_i,d_i)}_1\\
            &=\Delta((z,x_1,d_1,\dots,x_i,\A(i,z,x_1,d_1,\dots,x_i)),(z,x_1,d_1,\dots,x_i,d_i)) \leq t \secpa^{-c},
    \end{align}
    where the last inequality follows from~\Cref{eq:distbound}.
    Combining~\Cref{eq:iiprime,eq:iprimei} gives us
    \begin{equation}
      \norm{\mathcal{I}_i - \mathcal{J}_i}_1 \leq \norm{\mathcal{I}_i - \mathcal{I}_i'}_1 + \norm{\mathcal{I}_i' - \mathcal{J}_i}_1 \leq (i-1)t\secpa^{-c} + t\secpa^{-c} = i t \secpa^{-c}.  
    \end{equation}
    In particular $\norm{\mathcal{I}_t - \mathcal{J}_t}_1 \leq t^2 \secpa^{-c}$. But we know that running $\Rec$ on the internal state from $\mathcal{J}_t$ will output something which is not $\bot$ with all but negligible probability by correctness. Thus, as long as $\secpa^{-c} \leq \frac{1}{4t^2}$, 
    \begin{equation}\label{eq:mprimenotbot}
        \Pr[\bot \gets \mathcal{B} ] \leq \frac{1}{4} + \negl(\secpa).
    \end{equation}
    Putting together~\Cref{eq:mprimenotmbot,eq:mprimenotbot}, we get 
    \begin{align}
      \Pr[m\gets\mathcal{B}] 
      &= 1 - \Pr[m\not\gets\mathcal{B}] \\
      &\geq 1 - \Pr[\mathcal{B} \to m' \notin \{m,\bot\}] - \Pr[\bot\gets\mathcal{B}]\\
      &\geq \frac{3}{4} - \negl(\secpa),
    \end{align}
    which gives a contradiction. Thus, $\Samp$ is a distributional OWPuzz.
\end{proof}

%% file: 2qkdtostatepuzzle.tex
\subsection{2 Round QKD Implies State Puzzles}

\begin{definition}[State Puzzles~\cite{STOC:KhuTom25}]\label{def:statepuzzles}
    A state puzzle is defined by a uniform QPT algorithm $\Samp(1^\secpa)$ which outputs a pair $(s,\ket{\psi_s})$ of a bit string $s$ and a quantum state $\ket{\psi_s}$ such that given $s$, it is quantum computationally infeasible to output $\rho$ which overlaps notably with $\ket{\psi_s}$.
    Formally, we define $\sigma_{\secpa}$ to be the mixed state corresponding to the output of $\Samp(1^\secpa)$.
    For a non-uniform QPT $\A$, we define $\sigma_{\secpa}^\A$ to be the mixed state corresponding to the following process
    \begin{align}
    (s,\ket{\psi_s})\gets\Samp(1^\secpa)\\
    \rho\gets\A(1^n,s)\\
    \text{Output }(s,\rho).
    \end{align}

    We say that $\Samp$ is a state puzzle if for all QPT $\A$,
    \begin{align}
   \mathsf{TD} (\sigma_\secpa, \sigma^\A_\secpa) \geq 1 - \negl(\secpa).
    \end{align}
%   \mor{Is this really equal to $\mathbb{E}_{s\gets\Samp(1^n)}\langle\psi_s|\cA(1^n,s)|\psi_s\rangle$, which is the KT definition?}
\end{definition}

\begin{remark}
    Note that this definition at first seems different than the one in~\cite{STOC:KhuTom25}, which says that the expected overlap between $\rho$ and $\ket{\psi_s}$ is negligible. Formally,
    \begin{align}
    \BigE{{\rm Tr}(\ketbra{\psi_s}\rho)}
    {
    (s,\ket{\psi_s})
    \gets
    \Samp(1^\secpa)
    \\\
    \rho\gets \A(1^n,s)} \leq \negl(\secpa).
    \end{align}
    
    However, these definitions are equivalent. In particular, the best distinguisher between $\sigma$ and $\sigma^\A$ is the map which on input $(s,\ket{\psi_s})$ applies the measurement $\{\ketbra{\psi_s}, I-\ketbra{\psi_s}\}$. But the advantage of this distinguisher is exactly 
    \begin{align}
    1 - \BigE{{\rm Tr}(\ketbra{\psi_s}\rho)}{
    (s,\ket{\psi_s})
    \gets
    \Samp(1^\secpa)
    \\\
    \rho\gets \A(1^n,s)}.
    \end{align}
\end{remark}

For simplicity of presentation, we will not consider quantum advice. Observing the proof, our result also holds for quantum advice.

Informally, QKD describes a two party protocol where two QPT parties, Alice and Bob, send each other authenticated classical messages and unauthenticated quantum messages
that satisfies the following. 
\begin{enumerate}
    \item (Correctness): Either Alice and Bob output the same key, or one of the parties outputs $\bot$.
    \item (Security): If either Alice or Bob does not output $\bot$, then Eve cannot learn their final key.
\end{enumerate}

We consider a weaker notion of security, which we call one-sided computational security. In particular, we will only require that if Bob does not output $\bot$, then Eve cannot predict Bob's key. This definition immediately follows from the more standard stronger definition. Since we build OWPuzzs from the weaker definition, it is also possible to build OWPuzzs from the stronger definition.

\begin{definition}[Modified from~\cite{C:MalWal24}]
A two-round quantum key distribution (2QKD) scheme is defined as a tuple of algorithms 
$(\first,\second,\decode)$ with the following syntax.
\begin{enumerate}
    \item $\first(1^\secpa) \to (msg,\mu,st)$: A QPT algorithm which, on input the security parameter $1^\secpa$, outputs a message, consisting of a classical component $msg$ and a quantum state $\mu$, as well as an internal quantum state $st$. We allow $\mu$ to be a mixed state, and in particular $\mu$ and $st$ may be entangled.
    \item $\second(msg,\mu) \to \{(resp,\eta, k), \bot\}$: A QPT algorithm which, on input the first message $(msg,\mu)$, outputs a response consisting of a classical component $resp$ and a quantum mixed state $\eta$, as well as a classical key $k\in \{0,1\}^\secpa$. We also allow $\second$ to output a special symbol $\bot$, denoting rejection.
    \item $\decode(st,resp,\eta) \to \{k,\bot\}$: A QPT algorithm which, on input the internal state $st$ and the response $(resp,\eta)$, outputs a key $k\in\{0,1\}^\secpa$. We also allow $\decode$ to output a special symbol $\bot$, denoting rejection.
\end{enumerate}
\end{definition}

\begin{definition}[2QKD Correctness]
    We say that a 2QKD protocol $(\first,\second,\decode)$ satisfies correctness if there exists a negligible function $\negl$ such that for all $\secpa \in \N$, 
    \begin{align}
    \Pr[
    \bot\gets
    \second(msg,\mu) 
    :
     (msg,\mu,st)
     \gets
    \first(1^\secpa) 
    ] \leq \negl(\secpa)
    \end{align}
    and
    \begin{align}
    \BigPr{
    \decode(st,resp,\eta) = k}
    {
     (msg,\mu,st)
     \gets
    \first(1^\secpa) 
    \\ 
    (resp,\eta,k)
    \gets
    \second(msg,\mu) 
    } \geq 1 - \negl(\secpa).
    \end{align}
\end{definition}

\begin{definition}[2QKD One-sided Computational Security]
    For any non-uniform QPT algorithm $\A$, we define the following security game 
    $\mathbf{QKDSec}(\A,1^\secpa)$ as follows.
    \begin{enumerate}
        \item Run 
        $
        (msg,\mu,st)
        \gets
        \first(1^\secpa) 
        $ and send $(msg,\mu)$ to $\A$.
        \item $\A$ returns a quantum state $\rho$, which may be entangled with $\A$'s internal state.
        \item Run 
        $
         \{(resp,\eta,k),\bot\}
         \gets
        \second(msg, \rho) 
        $.
        \item Forward the response to $\A$.
        \item $\A$ outputs $k' \in \{0,1\}^\secpa$.
        \item The game outputs $1$ if $k' = k \neq \bot$.
    \end{enumerate}

    We say a 2QKD protocol $(\first,\second,\decode)$ satisfies one-sided computational security 
    if 
    for any non-uniform QPT algorithm $\A$,
    there exists a negligible function $\negl$ such that 
    \begin{align}
    \Pr[1\gets \mathbf{QKDSec}(\A,1^\secpa)] \leq \negl(\secpa).
    \end{align}
\end{definition}

\begin{theorem}\label{thm:2qkd}
    If there exists a 2QKD protocol satisfying correctness and one-sided computational security, then OWPuzzs exist.
\end{theorem}

\begin{proof}
    We will assume without loss of generality that the combined quantum state $\mu,st$ is a single pure state $\ket{\phi}_{\mathbf{BC}}$ over two registers. We can do this by running $\first$, deferring all measurements to the end. This will end up with a state over four registers $\mathbf{A,B,C,D}$ where $\mathbf{A}$ is the output register for $msg$, 
    $\mathbf{B}$ is the output register for $\mu$, 
    $\mathbf{C}$ is the output register for $st$, 
    and $\mathbf{D}$ is some ancilla register. We will then set $st$ to be the state over the registers $\mathbf{C}$ and $\mathbf{D}$ combined. $\decode$ will simply ignore the $\mathbf{D}$ register. This leads to a protocol with identical behavior, but where the state $\mu,st$ is pure.

    In particular, we will build a state puzzle. $\Samp$ will be defined as follows.
    \begin{enumerate}
        \item Run 
        $
         (msg,\ket{\phi}_{\mathbf{BC}})
         \gets
        \first(1^\secpa) 
        $.
        \item Output $s = msg$, $\ket{\psi_s} = \ket{\phi}_{\mathbf{BC}}$.
    \end{enumerate}

    We then claim that $\Samp$ is a state puzzle. In particular, let $p$ be a polynomial and let $\A$ be a non-uniform QPT adversary breaking state puzzle security of $\Samp$ with advantage $\epsilon(n)$. Formally, for $\sigma_n,\sigma_n^\A$ from~\Cref{def:statepuzzles}
    \begin{align}
    \TD(\sigma_n, \sigma_n^\A) \leq 1 - \epsilon(n)
    \end{align}
    for infinitely many $n$.

    We will construct an attacker $\A'$ breaking $QKD$ one-sided computational security with advantage dependent on $\epsilon$.
    \begin{enumerate}
        \item On receiving $(msg,\mu)$, run 
        $
         \ket{\phi}_{\mathbf{BC}}
         \gets
        \A(1^n,msg) 
        $.
        \item Output $\ket{\phi}_\mathbf{B}$, keeping internal state $\ket{\phi}_\mathbf{C}$.
        \item On receiving $(resp,\eta)$, run 
        $
         k'
         \gets
        \decode(\ket{\phi}_\mathbf{B}, resp, \eta) 
        $.
        \item Output $k'$.
    \end{enumerate}

    We can construct an inefficient adversary $\A''$ which acts the same as $\A'$, but always perfectly inverts the puzzle. That is, $\A''$ is defined as follows.
    \begin{enumerate}
        \item On receiving $(msg,\mu)$, output $\ket{\phi_{msg}}$, the residual state of $\first$'s quantum output conditioned on the classical message being $msg$.
        \item Output $\ket{\phi_{msg}}_\mathbf{B}$, keeping internal state $\ket{\phi_{msg}}_\mathbf{C}$.
        \item On receiving $(resp,\eta)$, run 
        $
         k'
         \gets
        \decode(\ket{\phi_{msg}}_\mathbf{B}, resp, \eta) 
        $.
        \item Output $k'$.
    \end{enumerate}
    By the definition of $\A$,
    \begin{align}
    \TD((msg,\A(1^n,msg)),(msg,\ket{\phi_{msg}})) \leq 1 - \epsilon(n)
    \end{align}
    and so since the only difference between $\mathbf{QKDSec}(\A')$ and $\mathbf{QKDSec}(\A'')$ is that $\A(msg) = \ket{\phi}$ is replaced by $\ket{\phi_{msg}}$, we get
    \begin{align}\label{eq:qkd1}
        \abs{\Pr[1\gets \mathbf{QKDSec}(\A'',1^n)] - \Pr[1\gets \mathbf{QKDSec}(\A',1^n)]} \leq 1- \epsilon(n).
    \end{align}

    But we can also observe that $\mathbf{QKDSec}(\A'',1^n)$ simply runs the honest QKD protocol, 
    with $\A''$ running $\decode$ on the correct state. Thus, by correctness, we have that in $\mathbf{QKDSec}(\A'',1^n)$, the following hold
    \begin{align}
            \Pr[k = \bot] &\leq \negl(\secpa)\\
            \Pr[k' = k] &\geq 1 - \negl(\secpa).
    \end{align}
    Thus, 
    \begin{align}\label{eq:qkd2}
        \Pr[1\gets \mathbf{QKDSec}(\A'',1^n) ] \geq 1 - \negl(\secpa).
    \end{align}
    \Cref{eq:qkd1,eq:qkd2} together give us
    \begin{align}
    \Pr[1 \gets \mathbf{QKDSec}(\A',1^n)] \geq \epsilon(n) - \negl(\secpa)
    \end{align}
    and therefore by 2QKD security, $\epsilon(n) \leq \negl(\secpa)$.
\end{proof}

\begin{theorem}[Corollary of~\Cref{thm:keexists}]
    There exists an oracle relative to which 2QKD exists, but $\mathbf{BQP} = \mathbf{QCMA}$.
\end{theorem}

\begin{proof}
    This oracle is exactly the same as the QCCC KE oracle. In particular, we observe that our QCCC KE protocol consists of exactly two rounds. But any 2 round QCCC KE protocol is trivially also a 2QKD protocol. Therefore the theorem follows.
\end{proof}

%% file: main.bbl
\newcommand{\etalchar}[1]{$^{#1}$}

%% file: sgsimlemma.tex
\subsection{Proof of \Cref{lem:sgsim}}
\label{sec:proof_lem_sgsim}
Here we give a proof of \Cref{lem:sgsim}, which is used to show \Cref{keylem}.
\Cref{lem:sgsim} is essentially a generalization of Lemma 4.4 from~\cite{goldin2025translating}. To our knowledge, the idea that superpositions of states with a random phase can be simulated using copies of the states first appeared in~\cite{ITCS:Zhandry24b} (see Lemma 5.5). A generalization of this technique also appears in a concurrent (but not independent) work~\cite{notreleasedyet}. For completeness, we include a full proof along with an explicit description of the simulator.

\begin{lemma}[\Cref{lem:sgsim} restated]
    There exists an efficient oracle algorithm $\Sim^{(\cdot)}(1^t)$ such that the following is true:
    Let $n$ be any natural number.
    Let $\{\ket{\phi_k}\}_{k\in \{0,1\}^n}$ be any collection of states. For a function $f:\{0,1\}^n \to \{0,1\}^n=[2^n]$, define the state 
    \begin{align}
    \ket{\phi^f}\coloneqq\frac{1}{\sqrt{2}}\ket{0} - \frac{1}{\sqrt{2}}\frac{1}{\sqrt{2^n-1}}\sum_{k\in \{0,1\}^n\setminus \{0\}}\omega_{2^n}^{f(k)}\ket{k}\ket{\phi_k}.
    \end{align}
    \takashi{
   We take sum over $k\in \{0,1\}^n\setminus \{0\}$ in the above, but we take sum over $k\in \{0,1\}^n$ in the main proof. We should make them consistent.
    }
    Let $Gen^{\{\ket{\phi_k}\}_k}$ be the CPTP oracle which acts as follows
    \begin{enumerate}
        \item Sample $k\gets \{0,1\}^n\setminus \{0\}$.
        \item Output $\ket{k}\ket{\phi_k}$.
    \end{enumerate}
    Then for all $t$, we have
    \begin{align}
    \E_f\left[|\phi^f\rangle\langle\phi^f|^{\otimes t}\right] = \Sim^{Gen^{\{|\phi_k\rangle\}_k}}(1^t).
    \end{align}
\end{lemma}

\begin{proof}
    For ease of notation, we will write $\ket{\phi_0}=\ket{0}$ and $f(0)=0$. 

    $\Sim(1^t)$ will be defined as follows
    \begin{enumerate}
        \item Sample $\eta\gets Bin(t,1/2)$ where $Bin(t,1/2)$ is the binomial distribution. That is,
        \begin{align}
        \Pr[Bin(t,1/2)=\eta] = \frac{{t\choose \eta}}{2^t}.
        \end{align}
        \item Query $Gen$ a total of $\eta$ times, receiving output states $\ket{k_1}\ket{\phi_{k_1}}\dots \ket{k_\eta}\ket{\phi_{k_\eta}}$.
        \item For $i>\eta$, set $k_i=0$. For ease of notation, we will define $\ket{\phi_0}=\ket{0}$.
        \item Generate the state
        \begin{align}
        \left(\sum_{\pi \in Sym(t)} \bigotimes_{i=1}^t\ket{k_{\pi(i)}}_{\mathbf{A_i}}\ket{0}_{\mathbf{B_i}}\right)\otimes \left(\bigotimes_{j=1}^t \ket{k_j}_{\mathbf{C_j}}\ket{\phi_{k_j}}_{\mathbf{D_j}}\right).
        \end{align}
        \item For all $i,j$, apply $SWAP_{\mathbf{B_i},\mathbf{D_j}}$ controlled on registers 
        $\mathbf{A_i},\mathbf{C_j}$ containing the same value in the standard basis. This results in state
        
        \begin{align}
        \left(\sum_{\pi \in Sym(t)} \bigotimes_{i=1}^t\ket{k_{\pi(i)}}_{\mathbf{A_i}}\ket{\phi_{k_{\pi(i)}}}_{\mathbf{B_i}}\right)\otimes \left(\bigotimes_{j=1}^t \ket{k_j}_{\mathbf{C_j}}\ket{0}_{\mathbf{D_j}}\right).
        \end{align}
        \item Finally, output registers $\mathbf{A_1B_1}\dots \mathbf{A_tB_t}$. This will contain the state
      
        \begin{align}
        \sum_{\pi \in Sym(t)} \bigotimes_{i=1}^t\ket{k_{\pi(i)}}_{\mathbf{A_i}}\ket{\phi_{k_{\pi(i)}}}_{\mathbf{B_i}}.
        \end{align}
    \end{enumerate}

    For a vector $\vec{k}\in (\{0,1\}^n)^t$, define $type(\vec{k}) \in [2^n]^{2^n}$ to be the vector defined by
    \begin{align}
    type(\vec{k})_i\coloneqq\#\text{ times }i\text{ appears in }\vec{k}.
    \end{align}
    We will denote $|T|$ by the number of $\vec{k}$ such that $type(\vec{k})=T$.
    \mor{This is slightly hard to understand because there are two $=$. How about defining $|T|$ by word, such as
    "We will denote $|T|$ by the number of $\vec{k}$ such that $type(\vec{k})=T$"}
    and
    \begin{align}
    types \coloneqq \left\{T\in [2^n]^{2^n}:\sum_i T_i = t\right\}.
    \end{align}
    We will then define
    \begin{align}
    \ket{\psi_T}\coloneqq \frac{1}{\sqrt{|T|}}\sum_{\vec{k}\in T}\bigotimes_{j=1}^t \ket{k_j}\ket{\phi_{k_j}}.
    \end{align}
    Note that after sampling $\vec{k}$, $\Sim^{Gen}(1^t)$ outputs $\ket{\psi_{type(\vec{k})}}$.

    Note that for any given $T$, the probability that $c=T_0$ is chosen is $\frac{{t\choose t-T_0}}{2^t}$. After this, the probability that $k_1,\dots,k_c$ are chosen such that $(k_1,\dots,k_c,0,\dots,0)\in T$ is
    \begin{align}
    \frac{\frac{|T|}{{t\choose t-T_0}}}{(2^n-1)^{t-T_0}}.
    \end{align}
    This is because $T/{t\choose t-T_0}$ is the number of vectors in $T$ which begin with $t-T_0$ non-zeros.
    
    In particular, consider the distribution over types (i.e. elements of $[2^n]^{2^n}$ with total weight $t$) defined by
    \begin{align}
    \Pr[\mathcal{D} \to T] = \frac{{t\choose t-T_0}}{2^t}\cdot \frac{\frac{|T|}{{t\choose t-T_0}}}{(2^n-1)^{t-T_0}} = \frac{|T|}{2^t\cdot (2^n-1)^{t-T_0}}.
    \end{align}
    $\Sim^{Gen}(1^t)$ outputs the state 
    \begin{align}
    \E_{T\gets \mathcal{D}}[\ketbra{\phi_{T}}].
    \end{align}

    We will now take a closer look at the state $\ket{\phi^f}^{\otimes t}$.
    \begin{align}
            \ket{\phi^f}^{\otimes t} = \left(\frac{1}{\sqrt{2}}\ket{0} - \frac{1}{\sqrt{2}}\frac{1}{\sqrt{2^n-1}}\sum_{k\in \{0,1\}^n\setminus \{0\}}\omega_{2^n}^{f(x)}\ket{k}\ket{\phi_k}\right)^{\otimes t}\nonumber\\
            \sum_{\vec{k}\in \{0,1\}^n} \left(\frac{1}{\sqrt{2}}\right)^t\left(\frac{-1}{\sqrt{2^{n}-1}}\right)^{t-type(\vec{k})_0}\bigotimes_{j=1}^t \omega_{2^n}^{f(k_j)}\ket{k_j}\ket{\phi_{k_j}}\nonumber\\
            \sum_{\vec{k}\in \{0,1\}^n} \left(\frac{1}{\sqrt{2}}\right)^t\left(\frac{-1}{\sqrt{2^{n}-1}}\right)^{t-type(\vec{k})_0}\omega_{2^n}^{\sum_{j=1}^tf(k_j)}\bigotimes_{j=1}^t \ket{k_j}\ket{\phi_{k_j}}.
    \label{eq:simeq1}
    \end{align}
    But note that $\sum_{j=1}^t f(k_j) = f \cdot type(k_j)$ where $f$ is viewed as a vector in $[2^n]^{2^n}$. 
    
    And so
    \begin{align}
            \ket{\phi^f}^{\otimes t} &= \sum_{T\in types} \omega_{2^n}^{f\cdot T}\left(\frac{1}{\sqrt{2}}\right)^t\left(\frac{-1}{\sqrt{2^{n}-1}}\right)^{t-T_0} \sum_{\vec{k} \in T}\bigotimes_{j=1}^t \ket{k_j}\ket{\phi_{k_j}}\\
            &=\sum_{T\in types}\omega_{2^n}^{f\cdot T}(-1)^{t-T_0}\left(\sqrt{\frac{|T|}{2^t\cdot (2^n-1)^{t-T_0}}}\right)\ket{\psi_T}.
    \end{align}

    We can then explicitly compute
   
    \begin{align}
           & \E_{f}\left[|\psi^f\rangle\langle\psi^f|^{\otimes t}\right] \\
           &= \frac{1}{(2^n)^{2^n}}\sum_{f,T,T'\in types} \omega_{2^n}^{f\cdot T - f\cdot T'}(-1)^{T_0+T_0'}\sqrt{\frac{|T|}{2^t\cdot (2^n-1)^{t-T_0}}} \sqrt{\frac{|T'|}{2^t\cdot (2^n-1)^{t-T_0'}}} \ket{\psi_T}\bra{\psi_{T'}}\\
           &= \frac{1}{(2^n)^{2^n}}\sum_{f,T\in types}\left(\sum_f \omega_{2^n}^{f\cdot 0}\right) \frac{|T|}{2^t\cdot (2^n-1)^{t-T_0}} \ketbra{\psi_T}\\
            &+\frac{1}{(2^n)^{2^n}}\sum_{T\neq T'\in types} \left(\sum_f \omega_{2^n}^{f\cdot (T-T')}\right)(-1)^{T_0+T_0'}\sqrt{\frac{|T|\cdot |T'|}{2^{2t}\cdot (2^n-1)^{2t-T_0-T_0'}}}  \ket{\phi_T}\bra{\psi_{T'}}.
    \end{align}
    But when $T-T'\neq 0$, 
    \begin{align}
    \sum_f \omega_{2^n}^{f\cdot (T-T')} = 0.
    \end{align}
    and so we finally can compute
    \begin{align}
        \E_{f}\left[\ketbra{\phi^f}^{\otimes t}\right] &= \sum_{T\in types} \frac{|T|}{2^t\cdot (2^n-1)^{t-T_0}}\ketbra{\psi_T}\\
       & =\E_{T\gets \mathcal{D}}\left[\ketbra{\psi_T}\right]\\
       & =\Sim^{Gen}(1^t).
    \end{align}
\end{proof}

%% file: lemmas/ortholem.tex
\section{Proof of~\Cref{lem:randomortho}}\label{sec:randomortho}

\begin{lemma}[\Cref{lem:randomortho}, restated]
    Let $N,N>0$ be integers with $N \leq M$. Let $\rho$ be the mixed state of dimension $NM$ defined by sampling $N$ states $\ket{\phi_k}\randfrom \HaarSt{M}$ for $k\in[N]$. That is, 
    \begin{align}
        \rho=\E_{\ket{\phi_k}\gets \HaarSt{M}, k\in [N]}\left[\bigotimes_{k=1}^N \ketbra{\phi_k}\right].
    \end{align}
    %$$\rho = \bigotimes_{k=1}^n \HaarSt{M}$$
    
    Let $\rho'$ be the same state but where we require that each $\ket{\phi_k},\ket{\phi_{k'}}$ are orthogonal. 
    Formally, $\rho'$ is the mixed state of dimension $NM$ defined as
    \begin{align}
    \rho' = \E_{U\gets\mu_M}\left[\bigotimes_{k=1}^N U|k\rangle\langle k|U^\dagger\right].
    \end{align}

    Then $\TD(\rho,\rho') \leq \frac{N(N+1)}{2M}$.
\end{lemma}

\begin{proof}
    Define $\rho_i$ to be the state defined by sampling the first $k$ states to all be orthogonal and the last $N-i$ states uniformly. Formally,
    \begin{align}
    \rho_i = \E_{U\gets\mu_M,|\phi_k\rangle\gets\mu_M^s}\left(\bigotimes_{k=1}^i \U|k\rangle\langle k|U^\dagger\right) \otimes \left(\bigotimes_{k=i+1}^N |\phi_k\rangle\langle\phi_k|\right).
    \end{align}
    Here $\rho_0 = \rho'$ and $\rho_n = \rho$. 
    We will show that for all $i$, $\TD(\rho_i,\rho_{i+1}) \leq \frac{i}{M}$.
    In particular, define 
    \begin{align}
    S_i = \left\{\bigotimes_{k=1}^N\ket{\psi_k} : \text{for all }j,j'\leq i, \braket{\psi_j|\psi_{j'}} = 0\right\}. 
    \end{align}
    Define $\Pi_i$ to be the projector onto the subspace spanned by $S_i$.
    We observe that $\Pi_{i+1} \rho_i$ is proportional to $\rho_{i+1}$. That is, $I - \Pi_{i+1}$ is the measurement which optimally distinguishes $\rho_i$ and $\rho_{i+1}$. In particular,
    \begin{align}
    \TD(\rho_{i},\rho_{i+1}) \leq {\rm Tr}((I-\Pi_{i+1})(\rho_{i} - \rho_{i+1})) = 1 - {\rm Tr}(\Pi_{i+1}\rho_i).
    \end{align}
    But ${\rm Tr}(\Pi_{i+1}\rho_i)$ is the probability that applying $\{\sum_{k\leq i}\ketbra{\psi_k}, I - \sum_{k<i}\ketbra{\psi_k}\}$ to $\rho_i$ results in the first measurement outcome (here $\ket{\psi_k} = U\ket{k}$). So
    \begin{align}
            {\rm Tr}(\Pi_{i+1}\rho_i) &= \BigE{\sum_{k\leq i}\abs{\bra{\psi_{i+1}}U\ket{k}}^2}{U\randfrom \HaarUn{m}\\ \ket{\psi_{i+1}}\randfrom \HaarSt{M}}\\
            &\leq \sum_{k\leq i} \BigE{\abs{\bra{\psi_{i+1}}U\ket{k}}^2}{U\randfrom \HaarUn{m}\\ \ket{\psi_{i+1}}\randfrom \HaarSt{M}}\\
            &\leq i\cdot \E\left[\abs{\braket{\psi|\phi}}^2:\ket{\psi},\ket{\phi}\randfrom \HaarSt{M}\right]\\
            &= \frac{i}{M}.
    \end{align}

    And so, we get that 
    \begin{align}
        \TD(\rho,\rho') \leq \sum_{i=1}^N \frac{i}{M}\\
        = \frac{N(N+1)}{2M}.
    \end{align}
\end{proof}